\documentclass[a4paper,11pt]{article}
\usepackage{jheppub}
\usepackage[utf8]{inputenc}
\DeclareUnicodeCharacter{2212}{\ensuremath{-}}

\usepackage{amsthm}
\newtheorem{theorem}{Theorem}

\usepackage{lineno}
\usepackage{xcolor}
\usepackage{tensor}
\usepackage{booktabs}
\usepackage{multirow}
\usepackage{hyperref}
\usepackage{subcaption}
\usepackage{array}

\newcommand{\tsc}[1]{\textsc{#1}}

\DeclareRobustCommand{\Sec}[1]{Sec.~\ref{#1}}
\DeclareRobustCommand{\Secs}[2]{Secs.~\ref{#1} and \ref{#2}}
\DeclareRobustCommand{\Eq}[1]{Eq.~(\ref{#1})}
\DeclareRobustCommand{\Eqs}[2]{Eqs.~(\ref{#1}) and (\ref{#2})}
\DeclareRobustCommand{\Fig}[1]{Fig.~\ref{#1}}
\DeclareRobustCommand{\Figs}[2]{Figs.~\ref{#1} and \ref{#2}}
\DeclareRobustCommand{\Thrm}[1]{Theorem~\ref{#1}}
\DeclareRobustCommand{\App}[1]{App.~\ref{#1}}
\DeclareRobustCommand{\Tab}[1]{Table~\ref{#1}}
\DeclareRobustCommand{\Reference}[1]{Ref.~\cite{#1}}
\DeclareRobustCommand{\References}[1]{Refs.~\cite{#1}}

\preprint{MIT-CTP/6077}

\title{\boldmath Simplex Demixing:  Disentangling Multiple Light-Flavor Jets at Colliders}

\author[a,b]{Gregorio de la Fuente}
\author[a,b,c,d]{and Jesse Thaler}
\affiliation[a]{Center for Theoretical Physics -- a Leinweber Institute, Massachusetts Institute of Technology, \\
Cambridge, Massachusetts, United States}
\affiliation[b]{The NSF AI Institute for Artificial Intelligence and Fundamental Interactions}
\affiliation[c]{Institut des Hautes \'Etudes Scientifiques, 91440 Bures-sur-Yvette, France}
\affiliation[d]{Institut de Physique Th\'eorique, CEA Paris-Saclay, 91191 Gif-sur-Yvette, France}

\emailAdd{grego137@mit.edu}
\emailAdd{jthaler@mit.edu}

\abstract{
Providing a practical and hadron-level definition of multiple jet flavors has been a long-standing challenge in collider physics.
Previous work has introduced a data-driven, operational definition of quark and gluon jets, but no robust generalization beyond two jet categories presently exists.
To address this, we introduce a machine-learning framework called ``simplex demixing'' to extract $T$ jet flavors (or topics in the statistics literature) from $M$ data samples (or mixtures) with minimal constraints.
Intuitively, our procedure identifies the maximally separable categories in the data, translating a multi-category classifier on the $M$ mixtures into a bounded geometric object with $T$ vertices.
We first demonstrate our procedure on a toy problem to infer the truth-level fractions of down-quark, up-quark, and gluon jets from synthetic mixtures of the three pure samples.
We then propose a tag-and-probe strategy to extract multiple light-flavor categories in a more realistic collider setting involving dijet production.
As expected, the identifiability of jet flavors depends on their relative abundance in the samples and the hadron-level information available to the classifier architecture. 
Our work opens the door to data-driven extractions of multiple jet flavor properties at the Large Hadron Collider.}

\begin{document}

\maketitle
\flushbottom

\section{Introduction}
\label{sec:introduction}

Quarks and gluons are produced copiously at the Large Hadron Collider (LHC).
Partons produced through hard scattering radiate gluons, which in turn split into other gluons and quark--anti-quark pairs through a process of perturbative showering.
At nonperturbative length scales of order $\Lambda_{\rm QCD}^{-1}$, these partons hadronize into color-neutral hadrons, forming collimated sprays of particles known as jets.
We can try to assign flavor labels, such as ``quark'' or ``gluon'', to the observed jets \cite{Nilles:1980ys, Jones:1988ay, Fodor:1989ir, Jones:1990rz, Lonnblad:1990qp, Pumplin:1991kc, Gallicchio:2011xq, Gallicchio:2012ez, Bhattacherjee:2015psa, FerreiradeLima:2016gcz, Bhattacherjee:2016bpy, Davighi:2017hok, Larkoski:2019nwj}.
This process, known as jet tagging, is important to many analyses at the LHC.
For example, it can enhance the fraction of signal-like jets in searches for physics beyond the standard model \cite{CMS:2017mbm, CMS:2017asf,  CMS:2017qxu, ATLAS:2017zuf, CMS:2019emo, ATLAS:2021kxv, ATLAS:2021suo, ATLAS:2023dyu, ATLAS:2025ycg}.

Traditional approaches to jet tagging use hand-crafted substructure observables based on insights from quantum chromodynamics (QCD) \cite{Salam:2010nqg, Abdesselam:2010pt, Plehn:2011tg, Altheimer:2012mn, Shelton:2013an, Altheimer:2013yza, Adams:2015hiv, Cacciari:2015jwa, Kogler:2018hem, Marzani:2019hun, Larkoski:2017jix,   Larkoski:2024uoc}. 
More recently, the rise of deep learning has led to the development of classifiers that leverage lower-level features to discriminate jet flavors \cite{Guest:2018yhq, Albertsson:2018maf, Radovic:2018dip, Carleo:2019ptp, Bourilkov:2019yoi, Schwartz:2021ftp,  Feickert:2021ajf, Boehnlein:2021eym, Karagiorgi:2022qnh, Plehn:2022ftl, Bonilla:2022wzp, DeZoort:2023vrm, Zhou:2023pti, Belis:2023mqs, Mondal:2024nsa, Larkoski:2026zvn}.
Typically, jet tagging relies on simulated jets whose flavor labels come from parton-level information in a Monte Carlo event generator, such as the event record of a parton shower. 
This naive generator definition assumes a perfect correlation between the UV flavor from the hard process and the IR flavor from the jet measurement, which is known to fail \cite{Gras:2017jty}. 
Progress has been made toward replacing the generator definition with sharp parton-level definitions \cite{Banfi:2006hf, Buckley:2015gua, Caletti:2022hnc, Caletti:2022glq, Czakon:2022wam, Gauld:2022lem, Caola:2023wpj, Behring:2025ilo}, factorization-based definitions \cite{Gallicchio:2011xc, Frye:2016aiz, Frye:2016okc}, and even fully hadron-level definitions \cite{Gras:2017jty,Komiske:2018vkc,Stewart:2022ari}.
The hadron-level approach advocated for in \Reference{Gras:2017jty} defines a given jet flavor as a phase space region that yields an enriched sample of partons of that flavor in a specific process.
It is conceptually attractive because it is explicitly tied to measurable hadronic final states.
Furthermore, it acknowledges that a ``quark'' in one process may not look like a ``quark'' in other processes due to color correlations with the rest of the event.

The operational definition in \Reference{Komiske:2018vkc} translates this conceptual hadron-level definition into a practical procedure for extracting the distributions of quark and gluon jets from collider data. 
Given two mixed jet samples whose quark/gluon compositions are unknown \textit{a priori}, the operational quark and gluon distributions are defined as the maximally separable categories in the data.
With this definition, one can extract the quark and gluon mixing fractions of the two mixtures using the classification without labels (CWoLa) \cite{Metodiev:2017vrx} and jet topics frameworks \cite{Metodiev:2018ftz}. 
Crucially, once these mixing fractions are known, one can determine separate quark and gluon distributions for any jet observable.
The operational definition has already been used to extract quark-like and gluon-like substructure 
distributions in CMS Open Data \cite{Komiske:2022vxg, Dolan:2023abg} and in ATLAS analyses \cite{ATLAS:2019rqw, LlorenteMerino:2019zov, VillaplanaPerez:2019yxt, ATLAS:2025ycg, Lenz:2020eam}, and it is also viable for testing quark- and 
gluon-jet modification in heavy-ion collisions \cite{Brewer:2020och}. 
However, its current implementation is limited to extracting at most two jet flavors from two mixtures.

In this paper, we extend the operational definition to any number of samples $M$ and operational topics $T$, as long as $T \leq M$.
We show that the output distribution of a classifier trained on the $M$ samples is bounded by a $(T-1)$-simplex, which is the higher-dimensional generalization of a line segment ($T=2$), a triangle ($T=3$), or a tetrahedron ($T=4$).
The classifier learns a mapping from the complex, high-dimensional feature space of jets to the lower-dimensional classifier space.
Hence, the vertices of the $(T-1)$-simplex correspond to regions of phase space, known as anchor regions, that yield pure samples of jets of each operational topic in a specific process.
In this sense, the vertices of the simplex \textit{are} the hadron-level definition of jet flavor.
Using this result, we design a practical machine learning procedure called \emph{simplex demixing} to extract the mixing fractions of operational topics from the measured mixtures.
As an initial toy study, we use simplex demixing to recover the truth-level fractions from several synthetic mixtures of down-quark, up-quark, and gluon jets.

Anticipating applications of simplex demixing at the LHC, we then propose a data-driven strategy for defining multiple light-flavor jet categories from dijet events.
In a quasi-realistic proof-of-concept study using jets from \tsc{Pythia} 8.310 \cite{Bierlich:2022pfr}, we construct $M = 14$ mixtures with different flavor compositions using a tag-and-probe method.
When we train our classifier on the 14 mixtures, we find that the distribution lies on a simplex spanned by $T = 7$ operational topics.
With perfect detector information and HL-LHC statistics, five operational topics are strongly associated with the \tsc{Pythia} down, up, strange, anti-strange, and gluon flavors, whereas the remaining two topics are weakly associated with the \tsc{Pythia} anti-down and anti-up flavors.
These results show that light-flavor jet categories can emerge clearly at the ensemble level, avoiding the issue of assigning flavor labels to individual jets.
This suggests a path toward follow-up studies that define multiple jet flavors directly from experimental data to extract their corresponding jet substructure distributions.

Beyond our collider physics context, simplex demixing is a powerful method for topic modeling.
Topic models were originally developed for identifying abstract themes that constitute a collection of documents. 
More recently, they have been applied to collider datasets for unsupervised learning \cite{Dillon:2019cqt, Dillon:2020quc, Alvarez:2021zje, LeBlanc:2022bwd, Dolan:2023abg}.
A classic topic model is latent Dirichlet allocation (LDA), which models each document as a convex combination of themes, and each theme as a probability distribution over a finite set of words \cite{blei2003latent}.
Topic modeling can also be formulated as a non-negative matrix factorization problem, where the matrix of documents over words is factorized into a matrix of topics over words and a matrix of fractions over topics \cite{10.1145/275487.275505, arora2012learning, pmlr-v28-arora13}.
Closely related matrix factorization problems arise in many other fields, including signal processing \cite{6200362} and archetypal analysis \cite{cutler1994archetypal, eugster2009spider, suleman2017validation}.
Some approaches to these problems exploit the simplex geometry and assume the existence of anchor regions, analogous to our strategy \cite{10.1117/12.366289, 4779330}.
However, these methods share an important limitation: they require that the mixed distributions be defined on finite feature spaces. 
By contrast, simplex demixing can be combined with any flexible deep architecture, allowing it to operate on continuous features and even structured inputs like variable-length arrays. 
Related techniques in label-noise learning combine deep classification with the simplex geometry \cite{Li:2021LabelNoise}, but focus on light-to-moderate label contamination and settings in which the latent categories are known \textit{a priori}.
Our emphasis, instead, is on demixing highly-mixed samples and discovering their latent topics through a weakly supervised approach.
Another closely related method is the \tsc{Demix} framework developed in \Reference{3322706.3361982}.
\tsc{Demix} provides a theoretical algorithm that performs topic modeling on continuous feature spaces in the asymptotic limit; our simplex demixing approach can be viewed as a practical machine learning alternative with good finite-sample performance.

The remainder of this paper is organized as follows.
In \Sec{sec:reviewing_the_two_mixture_operational_definition}, we review the operational definition for two jet mixtures, including the necessary ingredients from the CWoLa and jet topics frameworks, and then recast the operational definition as a geometric inference problem. 
In \Sec{sec:introducing_simplex_demixing}, we use our geometric framing to extend the operational definition to $M$ mixtures and $T$ operational topics.
We then introduce simplex demixing, a practical machine-learning framework compatible with any classifier architecture to infer the mixing fractions in realistic datasets.
In \Sec{sec:three_mixture_toy_study}, we validate simplex demixing by recovering the mixing fractions from different mixtures of down-quark, up-quark, and gluon jets.
We highlight that the classifier distribution has a distinct geometric signature depending on the number of underlying topics.
In \Sec{sec:realistic_physics_application}, we perform a proof-of-concept study with \tsc{Pythia}-generated dijet mixtures with realistic QCD correlations and show how simplex demixing yields operational topics for multiple light-flavor jets.
We provide a conclusion and outlook in \Sec{sec:conclusions}.
~\\

\noindent \textbf{Note added:} While finalizing our analysis, we learned of \Reference{JohannRaphael}, which also proposes a simplex method for topic modeling and applies it to standard machine-learning benchmark datasets, as well as to Galaxy10 DECals in a synthetic-mixture setting. The overall mathematical framework is consistent with ours, though we use complementary strategies for identifying the simplex geometry: post-hoc simplex fitting and architectural bottlenecks in \Reference{JohannRaphael} versus simplex regularization via loss terms in \Sec{sec:three_stage_training_procedure} below.

\section{Reviewing the two-mixture operational definition} \label{sec:reviewing_the_two_mixture_operational_definition}

In this section, we review the original operational definition of quark and gluon jets~\cite{Komiske:2018vkc}, which extracts two jet topics from two jet mixtures. 
We then provide an alternative, geometric approach that is easier to generalize, laying the necessary groundwork for simplex demixing in \Sec{sec:introducing_simplex_demixing}.

\subsection{Review of CWoLa and jet topics}\label{sec:review_of_cwola_and_jet_topics}

The goal of CWoLa is to learn the optimal classifier for distinguishing two categories using only two mixed samples of those categories~\cite{Metodiev:2017vrx}.
We focus on the task of quark/gluon discrimination from two jet mixtures, $M_1$ and $M_2$, but the same logic holds for any two categories/mixtures.
The mixtures are assumed to be convex combinations of quark and gluon distributions in the feature space $x$:%
\footnote{The feature space can be high dimensional or even variable dimensional, but we use just the symbol $x$ for notational simplicity.}
\begin{align}
    p_{M_1} (x) &= f_{1} \, p_q(x) + (1-f_{1}) \, p_g(x) \label{eq:firstmixture}, \\ 
    p_{M_2}(x) &= f_2 \, p_q(x) + (1-f_2) \, p_g(x) \label{eq:secondmixture},
\end{align}
with $f_1, f_2 \in [0, 1]$.
By the Neyman--Pearson lemma~\cite{Neyman:1933wgr}, the likelihood ratio is the optimal classifier for discriminating between $M_1$ and $M_2$, given by:
\begin{equation} \label{eq:l1_to_lqg}
    L_{M_1 M_2} (x) \equiv \frac{p_{M_1}(x)}{p_{M_2}(x)}
    = \frac{f_1 \, L_{qg} (x) + (1-f_1)}{f_2 \, L_{qg} (x) + (1-f_2)},
\end{equation}
where the quark-gluon likelihood ratio is:
\begin{equation}
    L_{qg} (x) \equiv \frac{p_q (x)}{p_g(x)}.
\end{equation}
If $f_1 > f_2$, one can show that $L_{M_1 M_2} (x)$ increases monotonically with $L_{qg} (x)$.
If $f_1 < f_2$, $L_{M_1 M_2} (x)$ instead decreases monotonically with $L_{qg} (x)$.
Therefore, $L_{M_1 M_2} (x)$ and $L_{qg} (x)$ define classifiers with the same decision boundaries up to an overall sign.

The most quark-enriched and gluon-enriched regions in phase space occur at the extrema of $L_{M_1 M_2} (x)$, suggesting that the endpoints are relevant for recovering the base quark and gluon distributions. 
Using the formalism introduced in \Reference{3322706.3361982}, the jet topics framework \cite{Metodiev:2018ftz} defines the \emph{reducibility factor} between any two distributions, $p_A$ and $p_B$, as the minimum of their likelihood ratio over feature space: 
\begin{equation}
    \kappa_{AB} = \min_{x} \frac{p_A(x)}{p_B(x)}.
\end{equation}
If both reducibility factors $\kappa_{AB}$ and $\kappa_{BA}$ vanish, then $p_A (x)$ and $p_B (x)$ are said to be \textit{mutually irreducible} \cite{pmlr-v30-Scott13}. 
Intuitively, two distributions are mutually irreducible when each of them has a pure region in phase space, also known as an anchor region, where it is nonzero but the other vanishes.

According to the jet topics framework~\cite{Metodiev:2018ftz}, the $M_1$ and $M_2$ mixtures can be expressed in terms of two mutually irreducible distributions $T_1$ and $T_2$, also known as topics:
\begin{align}
    p_{M_1}(x) &= \widetilde{f}_1 \, p_{T_1} (x) + (1-\widetilde{f}_1) \, p_{T_2} (x), \label{eq:first_topic_mixture} \\ 
    p_{M_2}(x) &= \widetilde{f}_2 \, p_{T_1} (x) + (1- \widetilde{f}_2) \, p_{T_2}(x). \label{eq:second_topic_mixture}
\end{align}
These equations imply that $L_{M_1 M_2} (x)$ and $L_{T_1 T_2} (x)$ are also monotonically related to one another:
\begin{equation}
    L_{M_1 M_2} (x) = \frac{\widetilde{f}_1 \, L_{T_1 T_2} (x) + (1-\widetilde{f}_1)}{\widetilde{f}_2 \, L_{T_1 T_2}(x) + (1-\widetilde{f}_2)},
\end{equation}
where $\widetilde{f}_1, \widetilde{f}_2 \in [0, 1]$ and $\widetilde{f}_1 > \widetilde{f}_2$.
Because $p_{T_1} (x)$ and $p_{T_2} (x)$ are mutually irreducible by assumption, the reducibility factors $\kappa_{T_1 T_2}$ and $\kappa_{T_2 T_1}$ are zero. 
This fact, together with the monotonicity of $L_{M_1 M_2} (x)$ and $L_{T_1 T_2} (x)$, yields the mixing fractions of the topics in the data, $\widetilde{f}_1$ and $\widetilde{f}_2$, in terms of the reducibility factors of the two jet mixtures, $\kappa_{M_1 M_2}$ and $\kappa_{M_2 M_1}$:
\begin{align}
    \kappa_{M_1 M_2} \equiv \min_{x} L_{M_1 M_2} (x) &= \frac{1-\widetilde{f}_1}{1-\widetilde{f}_2} \label{eq:k12}, \\
    \frac{1}{\kappa_{M_2 M_1}} \equiv \max_{x} L_{M_1 M_2} (x) &= \frac{\widetilde{f}_1}{\widetilde{f}_2} \label{eq:k21}.
\end{align}
The above system of equations has a unique solution as long as $\widetilde{f}_1 \neq \widetilde{f}_2$.

The key assumption of the operational definition of quark and gluon jets is that these two topics should be \emph{identified} with the ``quark'' and ``gluon'' distributions up to permutation~\cite{Komiske:2018vkc}.
In this definition, the word ``operational'' emphasizes that it has a concrete implementation that can be applied to both real and synthetic datasets, even if the labels ``quark'' and ``gluon'' are fundamentally ambiguous~\cite{Gras:2017jty}.
Assuming that $T_1$ is the quark topic and $T_2$ is the gluon topic, this implies:
\begin{equation} \label{eq:topics_are_flavors}
f_1 = \widetilde{f}_1, \qquad f_2 = \widetilde{f}_2.
\end{equation}
Inverting \Eqs{eq:k12}{eq:k21} and substituting the result into \Eqs{eq:first_topic_mixture}{eq:second_topic_mixture} yields an expression for the pure operational quark and gluon distributions in terms of the jet mixtures and the reducibility factors:
\begin{align} \label{eq:pure_quark_dist}
p_{q} (x) \equiv p_{T_1} (x) &= \frac{p_{M_1} (x) - \kappa_{M_1 M_2} \, p_{M_2} (x)}{1 - \kappa_{M_1 M_2}}, \\ \label{eq:pure_gluon_dist}
p_{g} (x) \equiv p_{T_2} (x) &= \frac{p_{M_2} (x) - \kappa_{M_2 M_1} \, p_{M_1} (x)}{1 - \kappa_{M_2 M_1}}.
\end{align}
Because the $p_{M_1} (x)$ and $p_{M_2} (x)$ distributions are known, \Eqs{eq:pure_quark_dist}{eq:pure_gluon_dist} enable us to plot the topic distributions over feature space, providing a fully data-driven way to study quark and gluon jet properties.

All that remains to recover the $f_1$ and $f_2$ fractions is to extract the reducibility factors from the mixtures $M_1$ and $M_2$.
There are several methods to accomplish this in the literature:
\begin{itemize}
    \item \textbf{Anchor bin method: } The $p_{M_1} (x)$ and $p_{M_2} (x)$ mixtures are binned in $x$, and their reducibility factors are calculated as the extrema of the log-likelihood ratios of bins.
    This method is straightforward, but it is sensitive to the statistics at the endpoints of the distribution.
    The anchor bin method is used in \References{Komiske:2022vxg, Dolan:2023abg, ATLAS:2019rqw, LlorenteMerino:2019zov, VillaplanaPerez:2019yxt, ATLAS:2025ycg, Lenz:2020eam}.

    \item \textbf{Log-likelihood ratio method: } To mitigate the sensitivity to the endpoint statistics, the log-likelihood ratio $\ln L_{M_1 M_2} (x)$ is fit to a polynomial. 
    This method, which is compatible with any quark-gluon discriminant $x$, is developed in \Reference{Komiske:2022vxg}.

    \item \textbf{ROC curve method: } The reducibility factors are calculated from the slopes at the endpoints of the receiver operating characteristic (ROC) curve. 
    This method is only applicable when the likelihood ratio is a monotonic function of the discriminant. 
    The ROC curve method is used in \References{Komiske:2022vxg, Dolan:2023abg}.
\end{itemize}
In \Sec{sec:introducing_simplex_demixing}, we describe a new geometric approach to demixing, called simplex demixing, that extracts the mixing fractions directly from a trained classifier, removing the need for a separate processing step to extract the reducibility factors.

\subsection{Geometric interpretation of the operational definition} \label{sec:geometric_interpretation}

The two-mixture operational definition discussed in \Sec{sec:review_of_cwola_and_jet_topics} relies on the notion of a likelihood ratio, which does not have a unique generalization to more than two distributions.
We now present a geometric picture of this operational definition, which will facilitate a generalization to any number of mixtures and topics in \Sec{sec:introducing_simplex_demixing}.
Instead of working directly with the likelihood ratio, we train a classifier on the mixtures using the cross-entropy loss.
In the two-mixture case, this classifier is monotonically related to the likelihood ratio and has useful geometric properties that generalize to the multi-mixture case.

\begin{figure}[t]
\centering
\subfloat[\label{subfig:reducibilityToGeometry_a}]{\includegraphics[width=0.45\textwidth]{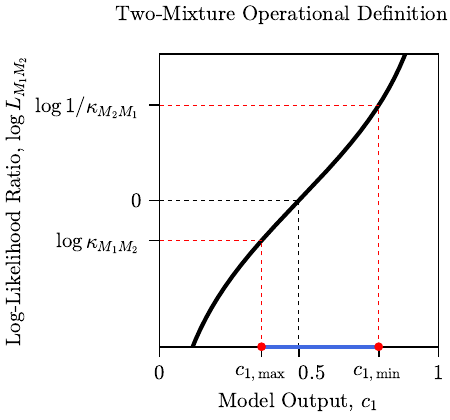}}
\hfill
\subfloat[\label{subfig:reducibilityToGeometry_b}]{\includegraphics[width=0.45\textwidth]{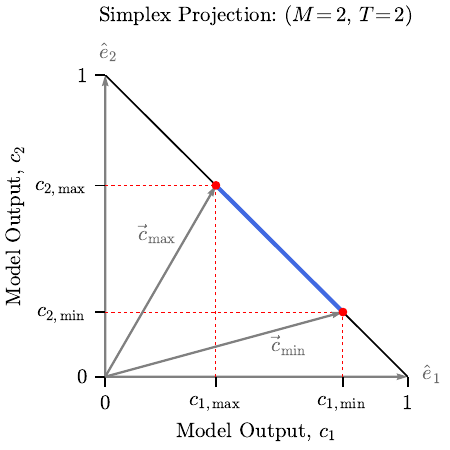}}
\caption{
(a) The mixed likelihood ratio is plotted against the classifier output according to \Eq{eq:mixedLHFunctional}.
The jet topics approach, based on extracting reducibility factors $\kappa_{M_1 M_2}$ and $\kappa_{M_2 M_1}$, is equivalent to finding the endpoints $\{c_{1, \text{min}}, c_{1, \text{max}}\}$ (in red) of the classifier distribution (in blue).
The positions of the two red points give the mixing fractions $f_1$ and $f_2$ using \Eqs{eq:cmin}{eq:cmax}. 
(b) The same classifier output is shown embedded in the probability 1-simplex, which is the convex hull of the orthonormal vectors $\hat{e}_1$ and $\hat{e}_2$, according to \Eq{eq:2_mixtures_2_topics}. \label{fig:reducibilityToGeometry}}
\end{figure}

Let us train a classifier on jets from two mixtures, $M_1$ and $M_2$, using the binary cross-entropy (BCE) loss. 
Each jet is labeled according to the mixture from which it is drawn. 
If the mixtures are finite, then their empirical distributions are sums of Dirac deltas over their jets in feature space.
In the limit of infinite training data, the distributions become continuous functions, rendering them amenable to analysis with the calculus of variations. 
In this asymptotic limit, the binary cross-entropy loss is:
\begin{equation}
    \label{eq:BCE}
    \mathcal{L}_{\rm BCE}[c_1, c_2] = -\int d x \big[ p_{M_1}(x) \, \ln c_1(x) + p_{M_2} (x) \, \ln c_2(x)\big],
\end{equation}
where $c_1(x) + c_2 (x) = 1$ and $c_1 (x), c_2(x) \geq 0$ for all $x$. 
At the loss minimum, where $\delta \mathcal{L} = 0$, the learned classifier becomes a monotonic function of the likelihood ratio $L_{M_1 M_2} (x)$, as shown in \Fig{subfig:reducibilityToGeometry_a}.
Specifically, the functional extremum of the loss is at:
\begin{equation} \label{eq:mixedLHFunctional}
    L_{M_1 M_2} (x) = \frac{c_1(x)}{1-c_1(x)} \quad \Rightarrow \quad c_1(x) = \frac{L_{M_1 M_2} (x)}{1+L_{M_1 M_2} (x)}.
\end{equation}
Just as the likelihood ratio is bounded between the two reducibility factors, the classifier distribution $c_1(x)$, depicted as a blue line segment on the model output axis, is bounded between the $c_{1, \text{max}}$ and $c_{1, \text{min}}$ points, in red. 
The classifier $c_1(x)$ is monotonic in $L_{M_1 M_2}(x)$, and $L_{M_1 M_2}(x)$ is in turn monotonic in $L_{q g} (x)$, so the endpoints of $c_1(x)$ are the quark- and gluon-enriched regions of phase space by transitivity. 
Using \Eqs{eq:k12}{eq:k21} and assuming \Eq{eq:topics_are_flavors}, we can relate $c_{1, \text{min}}$ and $c_{1, \text{max}}$ to the mixing fractions of quarks and gluons in the data:
\begin{align}
    c_{1, \text{min}} &= \frac{(1-f_1)}{(1-f_1) + (1-f_2)} \label{eq:cmin}, \\ 
    c_{1, \text{max}} &= \frac{f_1}{f_1 + f_2} \label{eq:cmax}.
\end{align}

Crucially, we can arrive at the same result by viewing the classifier distribution as a geometric object. 
To see this, we first parameterize the classifier distribution as:
\begin{equation}
    \vec{c}(x) = c_1 (x) \, \hat{e}_1 + c_2 (x) \, \hat{e}_2,
\end{equation}
where $\hat{e}_1$ and $\hat{e}_2$ are orthonormal vectors of a two-dimensional plane. 
The $c_1(x) + c_2(x) = 1$ and $c_1 (x), c_2 (x) \geq 0$ conditions restrict the space of possible outputs to the convex hull of $\hat{e}_1$ and $\hat{e}_2$, also known as the probability 1-simplex.
At the loss minimum, the distribution can be written as a vector-valued function from feature space $x$ to the 1-simplex:
\begin{align} \label{eq:classifierSegment}
    \vec{c} (x) = \frac{p_{M_1} (x) \,  \hat{e}_1 +   p_{M_2} (x) \,  \hat{e}_2}{p_{M_1}(x)+p_{M_2}(x)}.
\end{align}
Substituting the mixture model from \Eqs{eq:first_topic_mixture}{eq:second_topic_mixture} (assuming \Eq{eq:topics_are_flavors}) into \Eq{eq:classifierSegment},
we can rewrite the distribution $\vec{c}(x)$ in a suggestive form:
\begin{align} \label{eq:2_mixtures_2_topics}
    \vec{c}(x)= (1-\lambda (x)) \,  \vec{c}_{\text{min}} + \lambda (x) \,  \vec{c}_{\text{max}},
\end{align}
where $\lambda(x)$, known as a barycentric coordinate, is defined as:%
\footnote{In the special case that $f_1 = 1-f_2$, i.e.~exactly complementary mixing fractions, we have $\lambda(x) \equiv p_q(x)/ (p_q(x)+ p_g(x))$, i.e.~the quark/gluon classifier for pure samples.
Deviations from the pure classifier can be interpreted as an effective sample size imbalance between quarks and gluons in the training data, which is compensated by an asymmetry in the classifier endpoints.}
\begin{equation}
    \lambda(x) \equiv \frac{(f_1+f_2) \, p_q(x)}{(f_1+f_2) \, p_q(x)+(2-f_1-f_2) \, p_g(x)},
\end{equation}
and the classifier endpoints are:
\begin{equation}
   \vec{c}_{\text{min}} \equiv \begin{pmatrix}
        \frac{1-f_1}{2-f_1-f_2} \\ \frac{1-f_2}{2-f_1-f_2} \end{pmatrix}, \qquad
   \vec{c}_{\text{max}} \equiv \begin{pmatrix}
        \frac{f_1}{f_1 + f_2} \\ \frac{f_2}{f_1 + f_2}
    \end{pmatrix}.
\end{equation}

Let us now interpret \Eq{eq:2_mixtures_2_topics}.
The column vectors $\vec{c}_{\text{min}}$ and $\vec{c}_{\text{max}}$ have nonnegative entries that add up to one, so they lie in the probability 1-simplex. 
Because $p_q(x)$ and $p_g(x)$ are nonnegative and $f_1, f_2 \in [0, 1]$, the $\lambda (x)$ coordinate lies in the interval $[0, 1]$.
Furthermore, $p_q(x)$ and $p_g(x)$ are mutually irreducible by assumption, so the quarks have an anchor region $X_q$ where $p_q(X_q) > 0$ and $p_g(X_q) = 0$, and vice versa.%
\footnote{We use capital $X_t$ to denote a set, since there might be multiple $x_t$ values that satisfy the anchor property.}
As a result, $\lambda (X_q) = 1$ and $\lambda (X_g) = 0$, so $\lambda(x)$ saturates the bounds of the $[0, 1]$ interval over the feature space $x$. 
Geometrically, this implies that the classifier distribution lies in the convex hull of $\vec{c}_{\text{max}}$ and $\vec{c}_{\text{min}}$.
In other words, the classifier distribution is bounded by a 1-simplex within the probability 1-simplex, as shown in \Fig{subfig:reducibilityToGeometry_b}.
From this picture, it follows that the $f_1$ and $f_2$ fractions of quarks and gluons are recoverable unless $\vec{c}_{\text{min}} = \vec{c}_{\text{max}}$ or, equivalently, $f_1 = f_2$. The advantage of the geometric approach is that it generalizes straightforwardly to multiple topics, as we describe in the next section.

\section{Introducing simplex demixing} \label{sec:introducing_simplex_demixing}

In this section, we first extend the geometric approach from \Sec{sec:geometric_interpretation} to the case of three mixtures. 
This $M = 3$ example builds intuition for how the classifier distribution depends on the number of latent topics, which will be useful for interpreting the results of our first case study in \Sec{sec:three_mixture_toy_study}.
We then provide a fully general proof of simplex demixing for an arbitrary number of mixtures and topics.
With this proof in hand, we design a practical machine-learning procedure to perform simplex demixing and discuss how to compare the inferred topics against test datasets of known compositions.

\subsection{Intuition for demixing with three mixtures} \label{sec:intuition_for_demixing_with_three_mixtures}

Going from two samples to three suggests going from the binary cross-entropy loss in \Eq{eq:BCE} to the $M=3$ categorical cross-entropy (CCE) loss.
In the asymptotic limit, the categorical cross-entropy loss is:
\begin{equation}
    L_{\text{CCE}} = - \int dx \big[ p_{M_1}(x) \ln c_1(x) + p_{M_2}(x) \ln c_2(x) + p_{M_3}(x) \ln c_3(x)\big],
\end{equation}
where the learned $c_1(x)$, $c_2(x)$, and $c_3 (x)$ functions are subject to the constraint $c_1 (x) + c_2 (x) + c_3 (x) = 1$ and $c_1 (x), c_2 (x), c_3(x) \geq 0$.
The classifier distribution is now a vector in three-dimensional Euclidean space:
\begin{equation}
\vec{c} (x) = c_1 (x) \, \hat{e}_1 + c_2 (x) \, \hat{e}_2 + c_3 (x) \, \hat{e}_3,
\end{equation}
where $\hat{e}_1$, $\hat{e}_2$, and $\hat{e}_3$ are orthonormal vectors. 
The output of the classifier is therefore restricted to the convex hull of the three unit vectors, also known as the probability 2-simplex.
As in the previous section, the stationary condition on the loss implies that the classifier learns the posterior class probabilities:
\begin{equation} \label{eq:m_3_posterior_class_probabilities}
    \vec{c} (x)  = \frac{p_{M_1} (x ) \, \hat{e}_1 + p_{M_2} (x ) \, \hat{e}_2 + p_{M_3} (x ) \, \hat{e}_3}{p_{M_1} (x) + p_{M_2} (x) + p_{M_3}(x)}.
\end{equation}

We then model the three mixed probability distributions as convex combinations of three (or fewer) mutually irreducible operational topics:
\begin{equation} \label{eq:m_3_mixture_model}
    p_m (x)= \sum_{t=1}^T F_{mt} \, p_t (x).
\end{equation}
Here, $m \in \{1, 2, M=3\}$ indexes the jet samples, while $t \in \{ 1, \ldots, T \leq M\}$ indexes the underlying operational topics.
The matrix of mixing fractions $F_{mt}$ contains nonnegative entries that sum to one per row, such that $\sum_t F_{mt} = 1$.
By substituting the mixture model from \Eq{eq:m_3_mixture_model} into \Eq{eq:m_3_posterior_class_probabilities}, we find that $\vec{c}(x)$ lies in the convex hull of the $T$ vertices $\vec{V}_t$:
\begin{equation} \label{eq:m_3_classifier_vertex_parameterization}
    \vec{c} (x) = \sum_{t=1}^T \lambda_t (x) \, \vec{V}_t.
\end{equation}
In the above equation, the barycentric coordinates $\lambda_t (x)$ are defined as:
\begin{equation} \label{eq:barycentric_coords_m=3}
    \lambda_t(x) = \frac{\sum_{m=1}^M F_{mt} \, p_t(x)}{\sum_{m'=1}^M \sum_{t'=1}^T F_{m't'} \, p_{t'}(x)}.
\end{equation}
The vertices $\vec{V}_t$ are the columns of the fraction matrix rescaled to the embedding 2-simplex:
\begin{equation}
  \vec{V}_t \equiv \frac{\vec{F}_t}{\sum_{m = 1}^M F_{mt}}, \qquad \vec{F}_t \equiv \sum_{m = 1}^M F_{mt} \, \hat{e}_m. 
\end{equation}
Notably, the $\lambda_t (x)$ coordinates absorb all of the dependence on the feature space, so the vertices are global, input-independent properties of the classifier distribution.

\begin{figure}[t]
\centering
\subfloat[\label{subfig:proof_illustration_a}]{
    \includegraphics[width=0.45\textwidth]{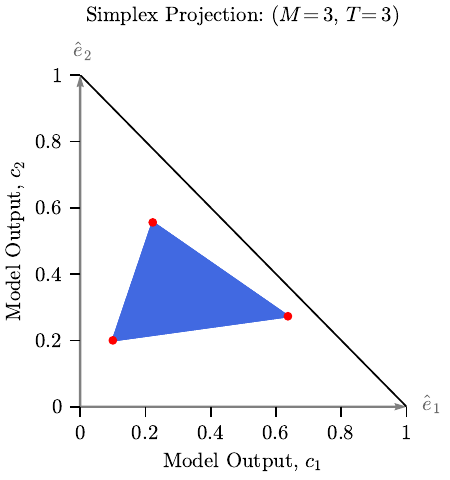}}%
\hfill
\subfloat[\label{subfig:proof_illustration_b}]{
    \includegraphics[width=0.45\textwidth]{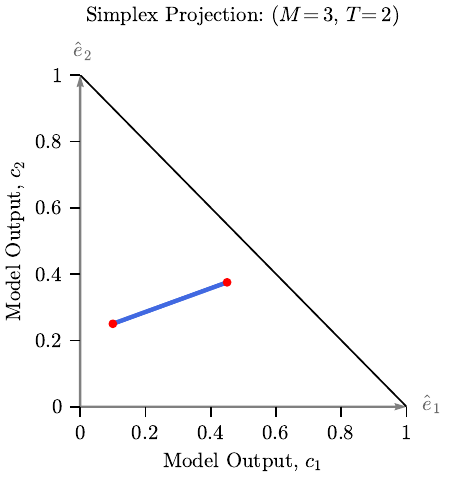}
}
\caption{\label{fig:proof_illustration} (a) The convex hull of the classifier distribution, projected onto the $c_1$--$c_2$ plane, when the $M=3$ mixtures are composed of $T=3$ irreducible topics, such as three distinct flavors of jets. (b) The same projection when the three mixtures are composed solely of $T=2$ topics. In this case, the simplex collapses to a line.}
\end{figure}

The learned distribution has support at the vertices by an argument similar to the one presented at the end of \Sec{sec:geometric_interpretation} (and proved more generally in \Sec{sec:demixing_any_number_of_topics} below).
Therefore, the vertex positions can be used to infer the fractions of operational topics in the three mixtures. 
Before proceeding to the general case, let us first build intuition for \Eq{eq:m_3_classifier_vertex_parameterization}.
This equation tells us that the number of vertices bounding the classifier geometry equals the number of underlying topics. 
In \Fig{fig:proof_illustration}, we illustrate this geometry for different numbers of topics $T$.
In both panels, the classifier space is projected onto the $c_1$--$c_2$ plane, and the triangular black outline denotes the corresponding projection of the bounding probability simplex.
If there are three underlying topics, then the convex hull of the classifier distribution is a 2-simplex defined by three vertices, as shown in \Fig{subfig:proof_illustration_a}. 
Each vertex corresponds to an anchor region where one topic distribution is nonzero and the other two vanish.
If there are two underlying topics, then the convex hull of the classifier distribution is a 1-simplex defined by two vertices, as shown in \Fig{subfig:proof_illustration_b}. 
If there is a single underlying topic, then the classifier distribution collapses to a single point, a 0-simplex, within the probability 2-simplex. 
In this way, the learned classifier geometry provides the complete information about the demixed topics.

\subsection{Demixing for any number of mixtures and topics}\label{sec:demixing_any_number_of_topics}

We now extend the previous argument to an arbitrary number of mixtures $M$ and latent topics $T$.
We find that the bounding shape of the classifier distribution is a $(T-1)$-simplex, which is the higher-dimensional generalization of a line or a triangle.
The logic is essentially the same as the arguments provided in \Secs{sec:geometric_interpretation}{sec:intuition_for_demixing_with_three_mixtures}, but we first state the problem formally and introduce important definitions to support the general proof.

Suppose we have $M$ mixed distributions $\{p_m (x) \}_{m=1}^M$ that are modeled as linear combinations of $T$ mutually irreducible topics $\{p_t (x) \}_{t = 1}^T$:
\begin{equation} \label{eq:general_mixture_decomposition}
    p_m (x) = \sum_{t=1}^T F_{mt} \, p_t (x), \qquad \sum_{t=1}^T F_{mt} = 1,
\end{equation}
where $m \in \{ 1, \dots, M \}$ indexes the mixed distributions and $t \in \{1, \dots, T \}$ indexes the topics.
The goal is to recover the fraction matrix $F_{mt}$, whose entries are nonnegative and sum to one per row. 
We define the vertices $\vec{V}_t$ as the columns of the fraction matrix rescaled to the probability simplex:
\begin{equation} \label{eq:vertex_definition}
    \vec{V}_t \equiv \frac{\vec{F}_t}{S_t}, \qquad \vec{F}_t \equiv \sum_{m=1}^M F_{mt} \, \hat{e}_m,
    \qquad S_t \equiv \sum_{m = 1}^{M} F_{mt}.
\end{equation}
To emphasize, the vector arrow indicates mixture space, which lives in the probability $(M-1)$-simplex, while the $t$ subscript indicates topic space, which spans a smaller $(T-1)$-simplex.
In the asymptotic limit, we train a classifier on the $M$ distributions using the categorical cross-entropy loss: 
\begin{equation} \label{eq:cce_loss}
    L_{\text{CCE}} \, [c_1, \dots, c_M] = -\int dx \,  \sum_{m=1}^M p_m(x) \ln c_m(x),
\end{equation}
where the $c_m (x)$ functions are constrained to be nonnegative and sum to one.

\begin{theorem} \label{theorem:demixing}

At $\delta L_{\text{CCE}} = 0$, the convex hull of the classifier distribution is a $(T-1)$-simplex $\mathcal{S}$ within the probability $(M-1)$-simplex. The $\mathcal{S}$ simplex has $T$ vertices, each corresponding to a different operational topic, provided that $T \leq M$.
The mixing fractions are recoverable up to permutations of the operational topics as long as the fraction matrix has full column rank.
\end{theorem}

\begin{proof}[Proof]
We augment the loss functional with a Lagrange multiplier $\mu (x)$ that explicitly enforces the classifier normalization-to-one condition at each $x$:
\begin{equation}
    L_{\text{CCE}} [c_1, \ldots, c_M, \mu] = - \int dx \left[\left( \sum_{m=1}^M p_m(x) \ln c_m(x) \right) + \mu(x) \left( 1-\sum_{m=1}^M c_m(x)\right)\right].
\end{equation}
At $\delta L_\text{CCE} = 0$, the functional derivatives of $L_\text{CCE}$ with respect to each $c_m(x)$ and $\mu(x)$ must be zero. 
Specifically:
\begin{align}
    \frac{\delta L_\text{CCE} [c_1, \ldots, c_M, \mu]}{\delta c_m (x)} &= - \frac{p_m(x)}{c_m(x)} + \mu(x) =0, \label{eq:LCC_by_cm}\\
    \frac{\delta L_\text{CCE} [c_1, \ldots, c_M, \mu]}{\delta \mu (x)} &= \sum_{m=1}^M c_m(x) - 1 =0.
\end{align}
These two equalities, taken together, imply that $\mu(x) = \sum_{m'=1}^M p_{m'}(x)$, so the classifier learns the posterior class probabilities:
\begin{equation}
    c_m(x) = \frac{p_m(x)}{\sum_{m'=1}^M p_{m'}(x)}.
\end{equation}
Equivalently, using vector notation $\vec{c}(x) = \sum_{m=1}^M c_m(x) \hat{e}_m$, we may write
\begin{equation} \label{eq:posterior_class_probabilities}
    \vec{c}(x) = \frac{\sum_{m=1}^M p_m (x) \, \hat{e}_m}{\sum_{m'=1}^M p_{m'} (x)}.
\end{equation}
Substituting the mixture model from \Eq{eq:general_mixture_decomposition} into \Eq{eq:posterior_class_probabilities}, we find:
\begin{align}
    \vec{c} (x) &= \frac{1}{\sum_{m'=1}^M \sum_{t'=1}^T F_{m't'} p_{t'} (x)} \sum_{m=1}^M \sum_{t=1}^T F_{mt} \, p_t (x) \, \hat{e}_m \nonumber \\
    &= \sum_{t=1}^T \left( \frac{S_t \, p_t (x)}{\sum_{t'=1}^T S_{t'} \, p_{t'} (x)} \right) \left( \sum_{m=1}^M \frac{F_{mt} \,  \hat{e}_m}{S_t} \right) \nonumber \\
    &= \sum_{t=1}^T \lambda_t (x) \, \vec{V}_t,
\end{align}
where we have used $\vec{V}_t \equiv \vec{F}_t / S_t$ and defined:
\begin{equation}
    \lambda_t(x) \equiv \frac{S_t \, p_t(x)}{\sum_{t'} S_{t'} \, p_{t'} (x)},
\end{equation}
which is the general version of \Eq{eq:barycentric_coords_m=3}. The barycentric coordinates $\lambda_t (x)$ are nonnegative because $S_t > 0$ and $p_t (x) \geq 0$ for all $t$ and $x$.
Furthermore, they sum to one at each $x$ by definition. 
Geometrically, this means that the classifier distribution lies in the simplex defined by the $T$ vertices $\vec{V}_t$.

Now, we show that the classifier distribution saturates the bounds of the simplex at the vertices.
By assumption, each topic $t$ has an anchor region $X_t$ in feature space where $p_{t} (X_t) > 0$ and $p_{t'} (X_t)= 0$ for $t' \neq t$.
Using the definition of $\lambda_t (x)$, this implies that $\lambda_t (X_t) = 1$ and $\lambda_{t'} (X_t) = 0$ at the anchor region.
In other words, the classifier distribution has support at the $T$ vertices.
This property is all we need to identify the $T$ topics. 
We do not need the classifier distribution to saturate the entire simplex, which is actually not guaranteed by mutual irreducibility alone.

To solve for the mixing fractions from the vertices, we use the condition that every row of the fraction matrix must add up to one: 
\begin{equation} \label{eq:vertices_to_fractions}
    \sum_t \vec{F}_t = \sum_t S_t \, \vec{V}_t = \vec{1},
\end{equation}
where $\vec{1} = (1 \ldots 1)^\top$ is a column vector of ones. 
The second equality of \Eq{eq:vertices_to_fractions} yields a system of equations for $S_t$. 
Under the mixture model, the vector of ones is in the span of the vertices, so the system is consistent. 
There is a unique solution for the $S_t$ coefficients if and only if the $T$ vertices are linearly independent or, equivalently, the fraction matrix has full column rank, completing the proof.

\end{proof}

\subsection{Machine learning architecture} \label{sec:machine_learning_architecture}

\Thrm{theorem:demixing} shows that the classifier distribution is bounded by a $(T-1)$-simplex within the probability $(M-1)$-simplex. 
Furthermore, the $T$ vertices are known functions of the entries of the fraction matrix $F_{mt}$.
In principle, one could use standard computational procedures to determine the convex hull of the points that make up the classifier distribution \cite{mair2024archetypal, 1411995}.
However, deep networks are susceptible to mismodeling, which can distort the learned geometry, especially when the samples are highly mixed.
We therefore need a way to regularize the classifier.
In addition, we would like a way to remove vertices that do not contribute to the convex hull of the classifier distribution. 
To address both challenges, we present a general framework compatible with any base classifier architecture. 

In standard cross-entropy training, the classifier outputs points within the probability $(M-1)$-simplex of the form:
\begin{equation}
    \vec{c}(x) = c_1(x) \, \hat{e}_1 + \ldots + c_M(x) \, \hat{e}_M.
\end{equation}
Instead, the core idea of our method is to parameterize the classifier so that its output lives in a learnable simplex:
\begin{equation} \label{eq:cce_parameterization}
    \vec{c}(x) = w_1 \, \lambda_1 (x) \, \vec{V}_1 + \ldots + w_M \, \lambda_M (x) \, \vec{V}_M.
\end{equation}
The three elements that appear in this functional form are roughly related to the three different training stages described in \Sec{sec:three_stage_training_procedure} below:
\begin{itemize}
    \item The vertices $\vec{V}_t$ define the simplex that encloses the learned distribution.%
    
    In the network, the vertices are implemented as an $M \times M$ matrix of trainable weights.
    In each forward pass, the weights are rescaled to ensure that the vertices always lie in the probability simplex.%
    \footnote{Concretely, we first apply a softmax transformation, defined generally in \Eq{eq:softmax_transformation}, to each row of the weight matrix. 
    We then transpose the matrix and normalize each row to one.
    This technical trick ensures that the vertices, given by the rows of the transformed matrix, span the column vector of ones, which is a key assumption of the mixture model as discussed in the formal proof of \Thrm{theorem:demixing}.}
    We consider $M$ vertices because there can be up to $T=M$ operational topics in the $M$ mixtures.
    The vertices are initialized to the $M$ corners of the probability simplex, $\vec{V}_m = \hat{e}_m$, to ensure that the learned classifier initially has access to all of classifier space. 

    \item The barycentric coordinates $\lambda_m(x)$  are functions that map each input, at a location $x$ in feature space, to a location in classifier space within the convex hull of the learned vertices. 
    These coordinates must be nonnegative and sum to one so that each $x$ is mapped into the convex hull of the vertices.

    In the network, the $\lambda_m(x)$ functions are implemented using any deep architecture for classification.
    The raw outputs of the base classifier, $( z_1 (x), \ldots, z_M (x))$, are rescaled using a softmax transformation, given by:
    \begin{equation} \label{eq:softmax_transformation}
        (\lambda_1(x), \ldots, \lambda_M(x)) = \frac{1}{\sum_{m=1}^M e^{z_m(x)}}\left( e^{z_1(x)}, \ldots, e^{z_M(x)} \right),\end{equation}
        which guarantees nonnegativity and unit normalization.
        The initialization scheme of the base classifier weights depends on the selected architecture.

    \item The vertex weights $w_m$ in $[0, 1]$ control whether the $\vec{V}_m$ vertices contribute to the bounding simplex of the classifier distribution or not. 

    In the network, the $w_m$ are implemented using a $1 \times M$ array of trainable logits, $\left( \ell_1, \ldots, \ell_M \right)$.
    The logits are scaled by a sigmoid transformation, given by:
    \begin{equation}
        (w_1, \ldots, w_M) = \left(  \frac{1}{1 + e^{-\ell_1}}, \ldots, \frac{1}{1 + e^{-\ell_M}}\right),
\end{equation}
    ensuring that the vertex weights range from zero to one. 
    These weights are all initialized to one because there can be up to $M$ topics in the mixtures. 
\end{itemize}
Each forward pass of the network multiplies the vertices, coordinates, and vertex weights according to \Eq{eq:cce_parameterization} and evaluates the cross-entropy loss on the final product $\vec{c}(x)$.
Intuitively, the vertices $\vec{V}_t$ shape the classifier distribution during training, while the weights $w_t$ sparsify the representation of the classifier distribution down to the true $T$ operational topics latent in the mixtures, as we now explain.

\subsection{Three-stage training procedure} \label{sec:three_stage_training_procedure}

\begin{figure}[t]
\centering
\subfloat[\label{fig:901000_results_a}]{\includegraphics[width=0.32\textwidth]{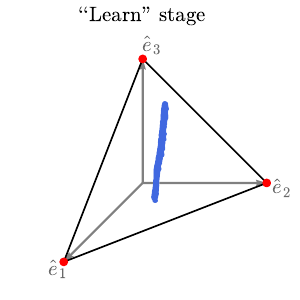}}
\hspace{0mm}
\subfloat[\label{fig:901000_results_b}]{\includegraphics[width=0.32\textwidth]{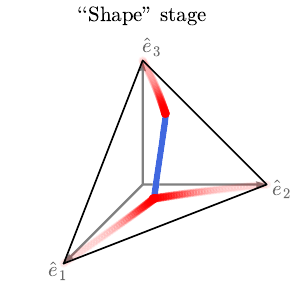}}
\hspace{0mm}
\subfloat[\label{fig:901000_results_c}]{\includegraphics[width=0.32\textwidth]{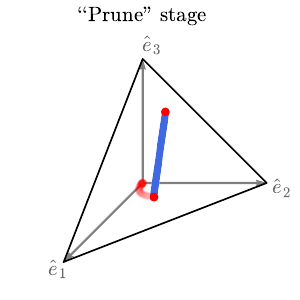}}
\caption{\label{fig:901010_results} 
A schematic representation of the three-stage training procedure for simplex demixing, for the case of $M = 3$ and $T = 2$.
(a) During the ``learn'' stage, the vertices $\vec{V}_m$ (in red) are frozen at the corners of the simplex, $\vec{V}_m = \hat{e}_m$. 
After training, the classifier distribution (in blue) resembles a 1-simplex with noticeable irregularities. 
(b) During the ``shape'' stage, the vertices are unfrozen and pulled toward the distribution by an auxiliary loss, $L_\text{edge}$, thereby regularizing the shape of the 1-simplex.
(c) During the ``prune'' stage, the excess vertex is sequestered at the origin of the coordinate system by a Lasso-type loss, leaving two vertices to span the distribution.
The two vertices are the two operational topics in the mixtures.}
\end{figure}

While it is in principle possible to directly train the classifier in \Eq{eq:cce_parameterization}, we found more robust results from using a three-stage training strategy.
Each stage corresponds to a different task, as illustrated in \Fig{fig:901010_results}:
\begin{itemize}
    \item \textbf{Stage one -- Learn:} The purpose of this stage is for the classifier to first learn a  good representation of the data.

The vertices are fixed at the corners of the probability $(M-1)$-simplex, $\vec{V}_m = \hat{e}_m$, to ensure the network has access to all of classifier space. 
For the same reason, we fix the vertex weights $w_m = 1$. 
Therefore, the function $\vec{c}(x)$ entering the cross-entropy calculation is effectively just:
\begin{equation}
    \vec{c}(x) = \lambda_1 (x) \, \hat{e}_1 + \ldots + \lambda_M (x) \, \hat{e}_M.
\end{equation}
In addition to the cross-entropy loss, we introduce an L2 regularization term on the neural-network weights $\theta$ of the $\lambda_m(x)$ functions:
\begin{equation}
    L_{\text{L2}} = \sum_{\theta \in \Theta_\lambda} \theta^2.
\end{equation}
The L2 regularizer induces a simplicity bias on the learned features of the model, which helps prevent overfitting~\cite{Hoerl01021970}.
The total loss in the ``learn'' stage is:
\begin{equation}
    L_\text{learn} = L_{\text{CCE}} + \gamma \, L_{\text{L2}},
\end{equation}
where $L_\text{CCE}$ is defined in \Eq{eq:cce_loss}.
We emphasize that $L_\text{L2}$ is just one example of a regularizer that can be used.
Other examples include dropout \cite{Srivastava:2014kpo} or batch normalization \cite{Ioffe:2015ovl}.
Ultimately, the optimal choice of regularizer depends on the base classifier used to implement the $\lambda_m(x)$ functions.

The network is trained in stage one until the validation cross entropy reaches a stable plateau.
By the end of stage one, the classifier distribution resembles a simplex, but with noticeable distortions, as shown in \Fig{fig:901000_results_a}.
The $\gamma$ hyperparameter and other hyperparameters of the $\lambda_m(x)$ functions are tuned to achieve a sustained plateau in the cross entropy through the end of the third training stage.
If overfitting sets in earlier, it interferes with the tasks of stages two and three, which we discuss next.

\item \textbf{Stage two -- Shape: } The goal of this stage is to encourage the vertices to tightly enclose the classifier distribution, enforcing the simplex geometry.
This stage is especially useful when the mixtures are highly mixed.

During this stage, the vertices are unfrozen and allowed to move, so the $\vec{c}(x)$ function entering the cross entropy is:
\begin{equation}
\vec{c}(x) = \lambda_1(x) \, \vec{V}_1 + \ldots + \lambda_M(x) \, \vec{V}_M.
\end{equation}
We need an auxiliary loss that attracts the vertices $\vec{V}_t$ toward the classifier distribution without overly distorting it.
In addition, the loss should continue to contract the vertices even after the configuration has collapsed to a lower-dimensional simplex, as occurs when there are fewer than $M$ topics in the mixtures.%
\footnote{While the volume of the learned simplex might seem like an obvious choice for regularization, the volume vanishes as soon as one of the vertices collapses to the subspace spanned by the remaining ones. Therefore, it does not work well when there are fewer than $M-1$ topics in the mixtures, and can even fail with $M$ topics if one of the edges collapses before the simplex fully shrinks.}
A simple loss that satisfies both requirements is the sum of the Euclidean distances between each pair of vertices (the simplex ``edges'').
For $t, t' \in \{1, \ldots, T=M\}$, the loss is:
\begin{equation} \label{eq:aux_loss}
    L_{\text{edge}} = \sum_{t > t'} \| \vec{V}_t - \vec{V}_{t'}\|_2.
\end{equation}
The total loss in the ``shape'' stage is therefore:
\begin{equation} \label{eq:shape_loss}
    L_\text{shape} = L_\text{CCE} + \gamma \, L_\text{L2} + \alpha \, L_\text{edge},
\end{equation}
where the hyperparameter $\alpha$ controls the strength of the auxiliary loss relative to the objective loss.
As the network trains on data, the vertices fall toward the distribution, as shown by the flow of the vertices in \Fig{fig:901000_results_b}.
The network is trained in stage two until the vertices converge to their final positions and $L_\text{edge}$ plateaus. 
How tightly the vertices enclose the distribution at the end of stage two depends on the value of $\alpha$.
If $\alpha$ is too weak, the cross entropy remains stable, but the vertices fail to reach the distribution.
If $\alpha$ is too strong, the cross entropy worsens, and the vertices distort the distribution by shrinking it below its true size.
The $\alpha$ parameter is tuned to an intermediate value that balances the two objectives.
\end{itemize}

 After the second stage, the classifier distribution is bounded by a set of $M$ vertices, which is overcomplete if the mixtures contain fewer than $M$ physical topics.
When $M$ is small, the number of physical vertices is evident from visual inspection of the classifier distribution: physical endpoints lie at the extremes, whereas extraneous vertices remain inside the cloud.
For instance, in the simple case study of \Sec{sec:three_mixture_toy_study}, the distinction between $T=2$ and $T=3$ is visually apparent.
When $M$ is larger, though, identifying the number of endpoints becomes challenging.
In \Sec{sec:realistic_physics_application}, we use physics domain knowledge to select $T$, and leave a discussion of domain-agnostic methods for choosing $T$ to future work. 

\begin{itemize}

\item \textbf{Stage three -- Prune: } Once $T$ is known, the ``prune'' stage uses Lasso (L1 regularization) \cite{51791361-8fe2-38d5-959f-ae8d048b490d} to remove any excess vertices. 
With the unnecessary vertices removed, the jets are redistributed among the physical vertices corresponding to operational topics. 

To accomplish this, we unfreeze the weights $w_t$, so the function entering the cross-entropy calculation is now the full form from \Eq{eq:cce_parameterization}:
\begin{equation}
\vec{c} (x) = w_1 \, \lambda_1 (x) \, \vec{V}_1 + \ldots + w_M \, \lambda_M (x) \, \vec{V}_M.
\end{equation}
Besides the cross-entropy and L2 losses, we also turn on an L1 loss on the vertex weights:
\begin{equation}
    L_\text{L1} = \sum_{m=1}^M w_m.
\end{equation}
The total loss in the ``prune'' stage is:
\begin{equation}
    L_{\text{prune}} = L_\text{CCE} + \gamma \, L_\text{L2} + \alpha \, L_\text{edge} +  \beta \, L_\text{L1}.
\end{equation}
The L1 loss promotes sparsity: if a vertex $\vec{V}_t$ is redundant for describing the distribution, its corresponding weight $w_t$ decays to zero.
In contrast, the relevant vertices' weights remain one, as illustrated in \Fig{fig:901000_results_c}.
The network is trained in stage three until $L_\text{L1}$ selects the relevant features and stabilizes, while ensuring that the cross-entropy and edge losses have plateaued.
In our experiments, stage three is a short fine-tuning step.
The $\beta$ hyperparameter is tuned to keep exactly $T$ vertices.

Finally, at the end of stage three, the vertices are inverted using \Eq{eq:vertices_to_fractions} to obtain the $M \times T$ matrix of mixing fractions $F_{mt}$, which constitutes the complete demixing information.
\end{itemize}

The L1 penalty in the third stage drives the effective vertices $w_t \, \vec{V}_t$ toward the origin of the $\hat{e}_m$ coordinate system.
If some of the vertices collapse to the origin, one might worry that some of the outputs $\vec{c}(x)$ would get pulled off the $\sum_m c_m(x) = 1$ hyperplane, distorting the simplex geometry from the second stage.
Fortunately, this failure mode does not occur because the cross-entropy gradients counteract the collapse.
Writing $p_m(x)$ for the $m$-th mixture, the gradient of the loss with respect to $c_m(x)$ is (repeating the relevant part of \Eq{eq:LCC_by_cm}):
\begin{equation}
    \frac{\delta L[c_1, \ldots, c_m, \ldots, c_M]}{\delta c_m (x)} = - \frac{p_m (x)}{c_m (x)},
\end{equation}
so the cross-entropy loss decreases as $c_m(x)$ increases.
To understand the effect of this gradient, we observe that the vertices are constrained to live on the simplex, the barycentric coordinates add up to one, and the weights range from zero to one, implying:
\begin{equation}
    \sum_m c_m(x) = w_1 \, \lambda_1(x) + \ldots + w_T \, \lambda_T(x) \leq 1,
\end{equation}
such that the distribution cannot grow past the simplex. 
Therefore, for the subset of active vertices $\mathcal{S} \subseteq \{ 1, \ldots, T \}$, the weights $w_t$ are one and the barycentric coordinates satisfy $\sum_{t \in \mathcal{S}} \lambda_t(x) = 1$ by the end of training.
Said another way, training will drive the $w_m$ to either 0 (inactive vertex, dominated by $L_{\rm L1}$) or 1 (active vertex, dominated by $L_{\rm CCE}$).

Our reasoning for separating the three training stages is that each stage serves a different objective that builds on the previous ones.
Regularizing the geometry before the classifier has learned a good representation risks driving the network into a suboptimal local minimum. 
Similarly, pruning vertices before they have stabilized risks eliminating important features of the distribution before their role has become clear.
With a staged approach, the $\alpha$ and $\beta$ hyperparameters also become easier to tune. 
Empirically, this framework consistently leads to robust minima.
We discuss the details of tuning the $\alpha$, $\beta$, and $\gamma$ in the case studies in \Secs{sec:three_mixture_toy_study}{sec:realistic_physics_application}.

\subsection{Identifying the operational topics} \label{sec:identifying_the_operational_topics}

To validate the performance of simplex demixing, we want to check that the extracted mixing fractions $F_{mt}$ agree with known benchmarks.
Imagine that we have a set of mixtures $p_m(x)$ of data that are \emph{known} convex combinations of $C$ categories $p_c (x)$:
\begin{equation} \label{eq:truth_level_decomposition}
    p_m (x) = \sum_{c=1}^C G_{mc} \, p_c(x),
\end{equation}
where $G_{mc}$ can be determined \emph{a priori} and $M \geq C \geq T$.
Of course, such a benchmark sample would not be available when deploying on real experimental data.
In our synthetic case studies, the categories correspond to \tsc{Pythia} flavors, but they could refer to any other ``truth'' label.

Naively, we could evaluate our method's performance by running simplex demixing on the mixtures $p_m (x)$, inferring the mixing fractions $F_{mt}$, and comparing the entries of $F_{mt}$ and $G_{mc}$. 
Such a comparison assumes that there is a one-to-one correspondence between the operational topics and the truth-level categories, but this will in general not be true for at least two reasons:
\begin{itemize}
    \item There will be a systematic offset between the entries of $F_{mt}$ and $G_{mc}$ if the definition of the categories $p_c(x)$ does not imply mutual irreducibility.
    Because the operational topics $p_t(x)$ are mutually irreducible by construction, perfect agreement is only possible if the category definition also enforces mutual irreducibility.
    
    \item The entries of $F_{mt}$ are noisy away from the infinite-data limit.
    \Thrm{theorem:demixing} is proven in the asymptotic limit, so we expect our estimator to recover the true mixing fractions exactly only in this limit.
    Intuitively, finite samples may not contain enough data points from a given topic to populate that topic's anchor region in phase space.
    In addition, the more weakly identifiable an operational topic is in the samples, the more data points are required to resolve its anchor region.  
\end{itemize}

To sidestep these complications, we use a comparison method that quantifies the degree to which the operational topics align with the categories rather than assuming that exact agreement should hold.
Our method measures two different quantities for all inferred topics $t$ and truth-level categories $c$:
\begin{itemize}
    \item  $p (t | c)$, the probability that the data are assigned to the $t$-th topic given that they belong to the $c$-th category.

    \item $p (c | t)$, the quasi-probability that the data belong to the $c$-th category given that they are assigned to the $t$-th topic. 
\end{itemize}
If both $p(c | t)$ and $p(t | c)$ are one, then there is a one-to-one correspondence between the $c$-th category and the $t$-th topic.
To emphasize, $p (t | c)$ has a true probability interpretation (at least asymptotically), whereas $p (c | t)$ only has a quasi-probability interpretation, as we now explain.

The $p(t|c)$ coefficients arise naturally from expanding the truth-level categories in terms of the operational topics.
This is possible because all probability distributions, including those of the categories, are assumed to be convex combinations of the topics:
\begin{equation}
    p_c(x) \equiv p(x|c) = \sum_{t=1}^T p(x|t, c) \, p(t|c) = \sum_{t=1}^T p(t|c) \, p_t(x),
\end{equation}
where we have used $p(x | t, c) = p_t(x)$ because the topics are independent of the categories used to construct them.
The $p(t|c)$ coefficients are therefore \textit{bona fide} conditional probabilities, which are nonnegative and sum to one in the asymptotic limit. 
On real datasets, any negative values of $\widehat{p}(t|c)$ provide a measure of finite-statistics effects and systematic errors from the choice of network hyperparameters.

Unlike the $p(t|c)$ coefficients, the $p(c|t)$ coefficients do not fall out naturally from a mixture model. 
Performing the analogous expansion gives:
\begin{equation} \label{eq:pt_to_pc_expansion}
    p_t(x) \equiv p(x|t) = \sum_{c=1}^C p(x|c,t) \, p(c|t) \not= \sum_{c=1}^C p(c|t) \, p_c(x),
\end{equation}
where the last step would require $p(x|c,t) = p_c(x)$, which directly contradicts the assumption of topic independence.
As a result, the $p(c|t)$ in \Eq{eq:pt_to_pc_expansion} are generally not equal to the true conditional probabilities obtained from $p(t|c)$ using Bayes's theorem: they can be negative even in the asymptotic limit. 
However, each diagonal $p(c|t)$ still becomes one when the $p_c(x)$ and $p_t(x)$ distributions coincide, so it is useful to interpret them as conditional quasi-probabilities.
On real datasets, the deviations of the diagonal $p(c|t)$ from one provide a measure of finite-statistics effects, systematic errors from the choice of network hyperparameters, and any mismatch between the category and topic definitions.

To perform a benchmark study and extract $p(t|c)$ and $p(c|t)$, we need a way to express the truth-level categories in terms of the operational topics, and vice versa.
With simplex demixing, we can relate the categories to the topics through their respective known compositions in the mixtures. 
Specifically, given the inferred $F_{mt}$ and the truth-level matrix $G_{mc}$, our estimators of these conditional (quasi-)probabilities are given by:
\begin{align} 
    \widehat{p}(t|c) &= \sum_{m=1}^M (G^+)_{cm} F_{mt} \label{eq:p(t|c)}, \\
    \widehat{p}(c|t) &= \sum_{m=1}^M (F^+)_{tm} G_{mc} \label{eq:p(c|t)},
\end{align}
where $(G^+)_{cm}$ and $(F^+)_{tm}$ denote the Moore--Penrose pseudoinverses of $G_{mc}$ and $F_{mt}$, respectively.
These systems of equations are either exactly determined or overdetermined because $M \geq T$ and $M \geq C$ by assumption.
In the asymptotic limit, \Eq{eq:p(t|c)} is consistent under the mixture model.
The reverse decomposition in \Eq{eq:p(c|t)} is exactly solvable when $C=T$, but generally only solvable in a least-squares sense when $C > T$.
Away from the limit of infinite training data, \Eqs{eq:p(t|c)}{eq:p(c|t)} yield least-squares solutions for both coefficient matrices.

\section{Three-mixture toy study} \label{sec:three_mixture_toy_study}

As a first validation of our simplex demixing procedure, we apply it to $M = 3$ synthetic mixtures of $d$-quark, $u$-quark, and gluon jets from \tsc{Pythia}.
Our method shows that the categories of ``down'', ``up'', and ``gluon'' are indeed latent in the mixed, unlabeled data.

\subsection{Event generation} \label{subsec:event_generation}

For this study, dijet events are generated using \tsc{Pythia} 8.310 \cite{Bierlich:2022pfr} at $\sqrt{s} = 14$ TeV with the default tunings and shower parameters, including hadronization and multiple parton interactions.
A parton-level $p_T$ cut of 400 GeV is applied to the two outgoing partons from the hard scattering process.
All two-to-two hard QCD processes that produce light quarks and gluons are turned on: \texttt{HardQCD:gg2gg}, \texttt{HardQCD:gg2qqbar}, \texttt{HardQCD:qg2qg}, \texttt{HardQCD:qq2qq}, \texttt{HardQCD:qqbar2gg}, and \texttt{HardQCD:qqbar2qqbarNew} with the default \texttt{HardQCD:nQuarkNew = 3}.

Final state non-neutrino particles are clustered using \tsc{FastJet} 3.5.0  \cite{Cacciari:2005hq, Cacciari:2011ma} into anti-$k_T$ jets \cite{Cacciari:2008gp} with radius $R = 0.4$.
All reconstructed jets are required to have $p_T \in [500, 550]$ GeV and pseudorapidity $|\eta| \leq 2$.
Per event, the hardest two jets are considered and kept if they pass the above cuts.
Each jet is saved as a $200 \times 4$ array, where each row corresponds to a jet constituent and contains its transverse momentum $p_T$, rapidity $y$, azimuthal angle $\phi$, and particle identification (PID) number.
The empty rows are padded with zeros. 
The jets are normalized in $p_T$ and centered in $y$ and $\phi$.
The PIDs corresponding to $\gamma$, $\pi^+$, $\pi^-$, $K^+$, $K^-$, $K_L$, $n$, $\bar{n}$, $p$, $\bar{p}$, $e^-$, $e^+$ are remapped to a float value starting at 0 and increasing by 0.1 for each type, following the prescription in \Reference{Komiske:2018cqr} (known as the PFN-ID scheme).

\begin{table}[t]
    \centering
    \begin{tabular}{c c c c c}
        \toprule
        Toy Set & Mixture 1 & Mixture 2 & Mixture 3 & Purpose\\
        \midrule
        A & $20/70/10$ & $50/30/20$ & $20/10/70$ & Lightly mixed 2-simplex \\
        B & $40/40/20$ & $20/40/40$ & $40/20/40$ & Heavily mixed 2-simplex \\
        C & $10/90/0$  & $25/75/0$  & $65/35/0$ & 1-simplex\\
        \bottomrule
    \end{tabular}
    \caption{Mixing fractions, in percent, for the three toy-study sets. Each entry gives the fractions $f_d/f_u/f_g$ for one artificial jet sample.}
    \label{tab:toy_mixing_fractions}
\end{table}

The unphysical parton-shower-labeled jet flavor is assigned based on the flavor of the closest parton in the hard process, required to be within $2R$ of the jet rapidity-azimuth direction~\cite{Larkoski:2019nwj}.
Based on the unphysical jet flavors, we create three pure samples of down-quark, up-quark, and gluon jets.
We then combine the samples into three toy sets (A, B, and C) with three mixtures of $4 \times 10^5$  jets per set. 
The compositions of the sets, summarized in \Tab{tab:toy_mixing_fractions}, are hand-picked to test our procedure on different classifier geometries:
\begin{itemize}
    \item Set A is lightly mixed.
    Each mixture is heavily enriched in a distinct jet flavor.
    This set tests a classifier geometry that, while falling short of filling the whole probability simplex, should be an easily identifiable 2-simplex.

    \item Set B is heavily mixed.
    No one category is predominant in any mixture, which is a situation closer to what one might expect in a real experiment.
    This set tests a classifier geometry that is significantly smaller than the full probability simplex but still looks like a 2-simplex.

    \item Set C contains only down- and up-quark jets, so it tests the case where the classifier distribution collapses to a 1-simplex within the probability 2-simplex.
\end{itemize}
Of the $1.2 \times 10^6$ jets in each set (i.e.~three mixtures of $4 \times 10^5$ jets each), 70\% are used for training, 20\% for validation, and 10\% for testing.

\subsection{Training the three-mixture demixer} \label{sec:training_the_three_mixture_demixer}

All training is performed using Keras with the TensorFlow backend \cite{chollet2015keras, tensorflow2015-whitepaper}.
We build the architecture described in \Sec{sec:machine_learning_architecture}. 
The base classifiers $z_m(x)$ used in the $\lambda_m(x)$ functions are particle flow networks (PFNs), a set-based architecture for jet classification introduced in \Reference{Komiske:2018cqr} and implemented by the \texttt{EnergyFlow} package \cite{EnergyFlow}.
The PFN architecture is given by:
\begin{equation}
    \text{PFN}(x) = F\left( \sum_{x_i \in x} \Phi(x_i)\right),
\end{equation}
where $\Phi$ maps each particle $x_i$ in jet $x$ to an $\ell$-dimensional latent space, and $F$ maps the sum of the latent representations over all particles in the jet to classifier space.
In our application, the per-particle $\Phi$ function is parameterized using three dense layers with 100, 100, and $\ell=128$ nodes respectively.
We choose $\ell=128$ because it is well within the plateau of quark/gluon discrimination performance. 
The per-jet $F$ function is parameterized using three layers, each with 100 nodes.  
We use Leaky ReLU activation functions with leakiness 0.3 to mitigate the dying-ReLU problem. 
The PFN weights are initialized with He-uniform initialization \cite{He:2015dtg}.
Across all three training stages, we use a batch size of 1000 and optimize the simplex demixer with Adam \cite{Kingma:2014vow} using a learning rate of $0.001$.

Using this architecture, we implement the general training procedure described in \Sec{sec:three_stage_training_procedure} to compute the $\widehat{p}(t|c)$ and $\widehat{p}(c|t)$ matrices introduced in \Sec{sec:identifying_the_operational_topics}.
The procedure yields point estimates of these two matrices, which are subject to statistical uncertainties. 
To estimate the uncertainty in these and other downstream estimands, we combine our training scheme with the nonparametric bootstrap \cite{10.1214/aos/1176344552, Efron_Hastie_2016}. 
In the bootstrap, we draw multiple resamples with replacement from the full training set, each of the same size as the original training set, and compute all desired estimands per resample.
The distribution of the estimands across resamples allows us to estimate their standard errors.

Within the bootstrap framework, we first obtain a single random resample and tune the $\gamma$ and $\alpha$ hyperparameters on this resample.
We tune the hyperparameters on a resample, rather than on the original dataset, because we report only bootstrap results and therefore want the tuning procedure to reflect the effective information content of a resample.
Each resample contains only $63.2\%$ unique original data points on average and therefore effectively carries less information for learning purposes, so it often requires stronger L2 regularization than the original dataset would.

To start tuning, we allocate 400 epochs for training overall.
To tune $\gamma$, we perform a run using the stage one settings extended to the full 400 epochs. 
We choose $\gamma$ so that, in this run, the validation cross entropy has converged by epoch 200 and stays on the plateau until epoch 400.
The optimal value depends on the dataset, but we find $\gamma = 1 \times 10^{-4}$ to be a good default. 
Once $\gamma$ is tuned, we run stage one from scratch for 200 epochs. 
Then, we train stage two for 200 epochs using 40 different values of the $L_\text{edge}$ hyperparameter $\alpha$.
We focus on the $\alpha$ values for which the $L_\text{edge}$ loss has converged by the end of stage two. 
For these values, we plot the validation cross entropy as a function of $\alpha$. 
Typically, the curve is flat at low $\alpha$ and increases at high $\alpha$, once $L_\text{edge}$ is strong enough to compress the classifier distribution below its true size.
We select $\alpha$ at the elbow between the flat and increasing regimes. 
The position of the elbow varies depending on the dataset, but we find empirically that $\alpha = 5 \times 10^{-4}$ is a good value to start scanning around.
In this toy case study, we omit the ``prune'' stage because the number of vertices is obvious by visual inspection: three for Sets A and B, and two for Set C. 
We have checked that our results are robust to different initialization seeds and $\alpha$ values near the elbow.

After tuning, we run simplex demixing on 20 bootstrap resamples of the training set, using the same 200/200 epoch split for stages one and two respectively, and fixed $\gamma$ and $\alpha$ hyperparameters.
Each run yields a slightly different estimate of the mixing fractions $F_{mt}$.
For each resample, we also recalculate the truth-level mixing fractions $G_{mc}$ using the \tsc{Pythia} flavors of the jets drawn in that resample.
For each $F_{mt}$ and $G_{mc}$ pair, we compute the $\widehat{p} (t|c)$ and $\widehat{p} (c|t)$ estimates using \Eqs{eq:p(t|c)}{eq:p(c|t)}.
The columns of $\widehat{p} (t|c)$ and $\widehat{p} (c|t)$ might be permuted across runs because of the permutation symmetry inherent to our procedure.
To order the columns consistently, we select a single run and minimize the distances between its vertices and those of every other run using the \texttt{linear\_sum\_assignment} method implemented in SciPy \cite{2020SciPy-NMeth}.
We report the estimated $\widehat{p}(t|c)$ and $\widehat{p}(c|t)$ as their bootstrap medians with their corresponding 15\%--85\% intervals:
\begin{equation}
    \big[ Q_{0.15}\left( \widehat{p}(t|c) \right), Q_{0.85}\left( \widehat{p}(t|c) \right)\big],  \qquad
    \big[ Q_{0.15}\left( \widehat{p}(c|t)\right), Q_{0.85}\left( \widehat{p}(c|t)\right)\big].
\end{equation}
We emphasize that we do not report simplex demixing results on the original dataset because we are primarily interested in quantifying the uncertainty in the inferred matrices across bootstrap resamples.

Note that this fixed-hyperparameter procedure to estimate uncertainties neglects the dependence of the optimal hyperparameter values on the specific bootstrap resample.
A full nested bootstrap would retune the hyperparameters per resample, capturing this additional source of variance. 
Doing the bootstrap with fixed hyperparameters yields narrower error bars on the entries of the coefficient matrices than a nested bootstrap would.
We do not perform the nested bootstrap here because it is computationally prohibitive.
That said, we expect the underestimate of the uncertainty to be small because the curves of validation cross entropy versus hyperparameter have a weak dependence on the specific resample, so retuning does not seem to change the hyperparameters significantly.

\subsection{Identifying the three operational topics} \label{sec:identifying_the_three_operational_topics}

\begin{figure}[htbp]

\centering

\subfloat[]{%
\vspace{-0.1cm}
    \includegraphics[width=0.32\textwidth]{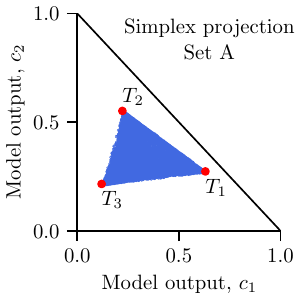}
}
\hfill
\subfloat[]{%
    \includegraphics[width=0.32\textwidth]{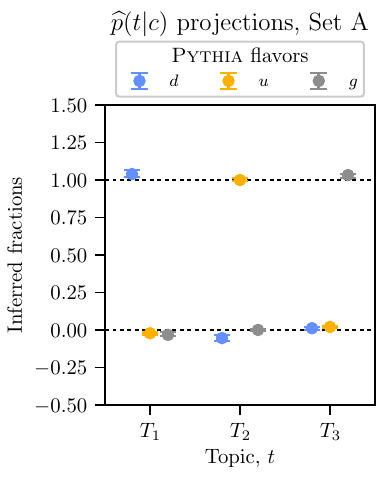}
}
\hfill
\subfloat[]{%
    \includegraphics[width=0.32\textwidth]{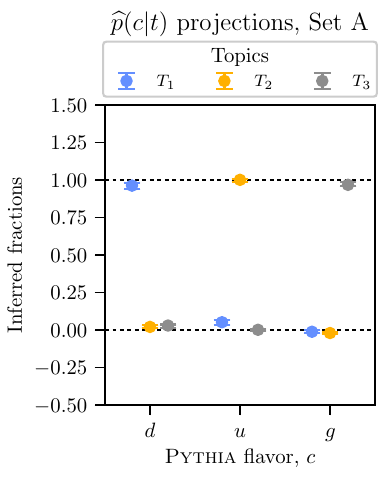}
}

\vfill

\subfloat[]{%
\vspace{-0.1cm}
    \includegraphics[width=0.32\textwidth]{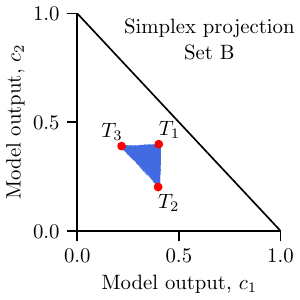}
}
\hfill
\subfloat[]{%
    \includegraphics[width=0.32\textwidth]{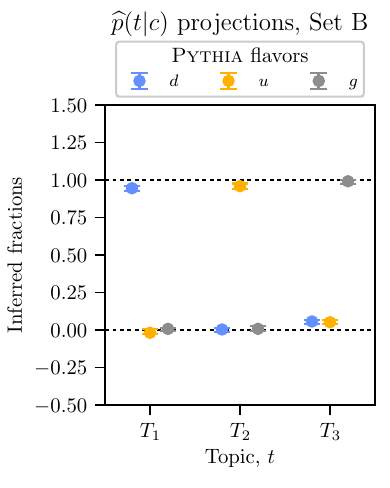}
}
\hfill
\subfloat[]{%
    \includegraphics[width=0.32\textwidth]{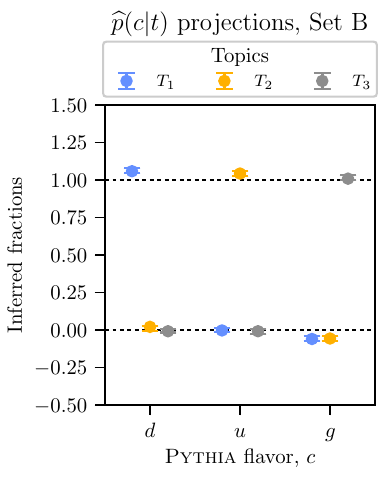}
}

\vfill

\subfloat[]{%
\vspace{-0.1cm}
    \includegraphics[width=0.32\textwidth]{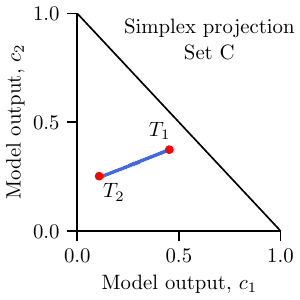}
}
\hfill
\subfloat[]{%
    \includegraphics[width=0.32\textwidth]{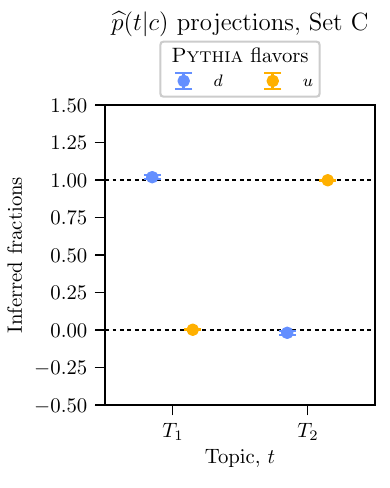}
}
\hfill
\subfloat[]{%
    \includegraphics[width=0.32\textwidth]{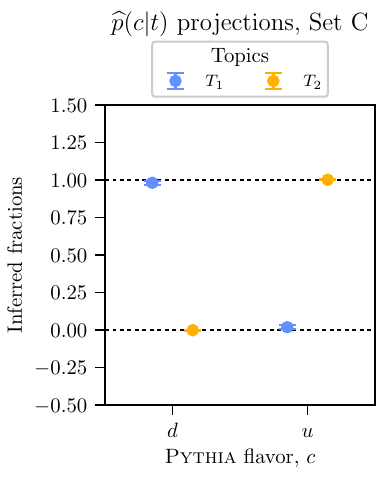}
}

\caption{\label{fig:three_mixture_results}
Results of simplex demixing for the three-mixture toy study, for Set A (top row), Set B (middle row), and Set C (bottom row).
The left column shows the classifier distributions (in blue) projected onto the $c_1$--$c_2$ plane.
The middle and right columns show the entries of $\widehat{p}(t|c)$ and $\widehat{p}(c|t)$, respectively.
The one-to-one correspondence between the \tsc{Pythia} flavors and operational topics shows that the concepts of ``up'', ``down'', and ``gluon'' jets are latent in the data across different levels of contamination.
When gluon jets are absent from the mixtures, as in Set C, the geometry collapses to a 1-simplex, as predicted by \Thrm{theorem:demixing}.
}
\end{figure}

The final results of simplex demixing for Sets A, B, and C are shown in \Fig{fig:three_mixture_results}.
The left column shows the classifier distributions projected onto the $c_1$--$c_2$ plane.
The output space is bounded by the probability 2-simplex, outlined in black.
In the classifier distribution, each blue dot in classifier space corresponds to a different input jet.
Together, the jets populate a 2-simplex bounded by the three learned vertices, in red. 
Comparing Sets A and B, the classifier distributions become smaller within the probability 2-simplex the more mixed the mixtures are.
For Set C, where the gluon jets are absent from the samples, the classifier distribution collapses to a 1-simplex spanned by the down and up vertices alone, as predicted by \Thrm{theorem:demixing}.

In the middle and right columns of \Fig{fig:three_mixture_results}, we show the median entries of $\widehat{p}(t|c)$ and $\widehat{p}(c|t)$ from the bootstrap together with their respective 15\%--85\% intervals.
In each set, both matrices are approximately diagonal, establishing a clean one-to-one correspondence between the \tsc{Pythia} down-quark, up-quark, and gluon flavors and the $T_1$, $T_2$, and $T_3$ topics.
In Set B, the largest gluon fraction of $\widehat{p}(t|c)$ is consistent with one, though the largest up and down fractions are slightly lower than one. 
As discussed in \Sec{sec:identifying_the_operational_topics}, these deviations are consistent with finite-statistics effects, which become more prominent for highly mixed datasets, and with a small mismatch between the \tsc{Pythia} and operational topic definitions.
In particular, the \tsc{Pythia} quark and gluon flavors have a small but nonzero reducibility factor $\kappa_{qg} > 0$, implying that every gluon region in phase space is contaminated by quarks~\cite{Larkoski:2019nwj}. 
A nonzero $\kappa_{qg}$ pulls the largest quark fractions below one, matching the direction of the offset we see in Set B. 
We emphasize, though, that the \tsc{Pythia} flavors are not hadron-level concepts; they solely provide a semi-quantitative baseline for the performance of simplex demixing. 
Overall, the agreement between the \tsc{Pythia} flavors and the operational topics indicates that the concepts of ``up'', ``down'', and ``gluon'' jets are latent within the mixtures across different levels of contamination.

\section{Proof-of-concept physics application} \label{sec:realistic_physics_application}

We now study a quasi-realistic application of simplex demixing to operationally define multiple light flavors from dijet events at the LHC.
We find that the seven light flavors emerge (five strongly and two weakly) at the ensemble level, assuming large datasets with perfect detector information, even if tagging them individually is notoriously hard. 

\subsection{Event generation}
\label{subsec:tag_probe_event_generation}

While our main goal is to replace supervised jet tagging with a data-driven extraction of operational topics, we begin with the same ingredients as a typical fully supervised analysis.
Supervised jet tagging uses two datasets: simulated Monte Carlo jets and jets observed in experimental data.
In the supervised paradigm, a classifier is first trained to distinguish the simulated jet flavors and then used to assign flavor labels to the experimentally measured jets.

To obtain the simulated dataset, we generate 30 million dijet events in \tsc{Pythia}, assign each jet a parton-level flavor by the closest-hard-parton prescription, and preprocess the individual jets following the same procedure as in \Sec{subsec:event_generation}.
Using the parton-level flavors, we create seven pure samples of \tsc{Pythia} $d$, $\bar{d}$, $u$, $\bar{u}$, $s$, $\bar{s}$, and $g$ jets.
We then train a standard supervised classifier on the seven samples with the categorical cross-entropy loss and the same PFN architecture as in \Sec{sec:training_the_three_mixture_demixer}.

For this proof-of-concept study, we create the ``experimental'' dataset using dijet events generated in \tsc{Pythia}.
We keep only events in which the two $p_T$-leading jets both lie in the [500, 550] GeV range and satisfy $|\eta| \leq 2$.
We save a total of 128 million events, which is close to the statistics expected at the HL-LHC given our cuts \cite{ATLAS:2017ble, CMS:2023fix}. 
Each jet is preprocessed and assigned a parton-level flavor following \Sec{subsec:event_generation}.
The parton-level flavors are used only for validation and for the heavy-flavor exclusion described next; they do not enter the demixing procedure.
For the purpose of operationally defining the light flavors, we treat all jets labeled as charm or bottom as perfectly taggable and exclude them from our analyses.
Together, these heavy-flavor jets make up less than 4\% of the total.
A more realistic study would need to include them, since mixture-dependent fractions of heavy flavors can yield additional topics.
We emphasize that, although we use Monte Carlo jets here, the general procedure would remain unchanged if we were to replace the \tsc{Pythia} jets with jets observed experimentally, apart from the idealized exclusion of heavy flavor.

\subsection{The tag-and-probe method}
\label{subsec:tag_and_probe}

The goal of our study is to recover the operational topics latent in the experimental dataset. 
The dataset is described by the joint distribution $p(x_1, x_2)$ of the two jets in each event.
This distribution factorizes into the operational topic distributions as:
\begin{equation} \label{eq:universal_topic_factorization}
    p(x_1, x_2) = \sum_{t, t'} p(x_1 | t) \, p(x_2|t') \, p(t'|t) \, p(t),
\end{equation}
where $t$ and $t'$ index the same topics over $x_1$ and $x_2$, respectively.
Here, $p(x_1 |t) \equiv p_t(x_1)$ and $ p(x_2|t') \equiv p_{t'}(x_2)$ denote the corresponding topic distributions, and $p(t'|t) \, p(t)$ is the topic co-occurrence matrix. 
The entries of $p(t'|t) \, p(t)$ are determined not only by the hard process but also by showering and nonperturbative physics, since the operational-topic assignments are defined in terms of hadron-level jet features.
If the physics is symmetric under the criterion for splitting each event into $x_1$ and $x_2$, then $p(x_1, x_2) = p(x_2, x_1)$.
In this case, the marginal distributions, $p(x_1)$ and $p(x_2)$, are the same linear combination of operational topics.

\begin{figure}[t]
\centering
\includegraphics[width=0.90\textwidth]{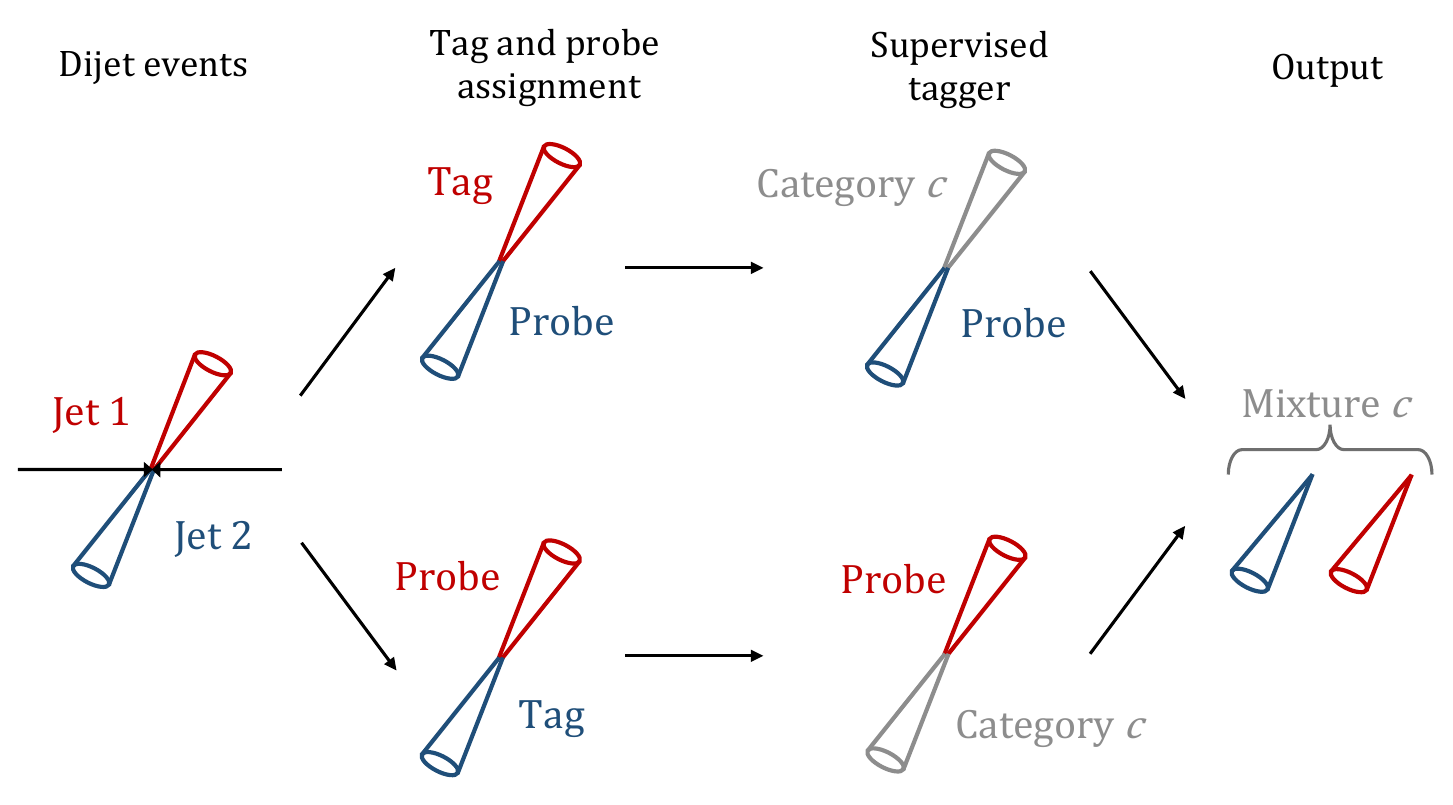}
\caption{\label{fig:tag_and_probe_illustration}
Schematic illustration of our tag-and-probe method.
Each dijet event is split into a tag jet and a probe jet using the two possible assignments.
For each assignment, the flavor of the tag jet is predicted by a supervised tagger pre-trained on the seven light-flavor jets from \tsc{Pythia}. 
In this example, the supervised tagger predicts the same flavor $c$ for the two jets, but the two flavors can (and often do) differ.
Finally, each probe jet is placed in a mixture according to the flavor of its corresponding tag jet. 
Above, the two jets are placed in the $c$-th mixture, which is then subdivided into a low-$\eta$ bin and a high-$\eta$ bin.
This procedure yields fourteen jet mixtures with varying flavor compositions.
}
\end{figure}

To expose the latent topics, we process the experimental dataset using a tag-and-probe method, which is similar in spirit to the Tag N' Train technique~\cite{Amram:2020ykb}.
A schematic illustration of the following four steps is shown in \Fig{fig:tag_and_probe_illustration}.
\begin{enumerate}
    \item For each event, we classify the jet with the smaller azimuthal angle as $x_1$, and the other as $x_2$.
    
    \item We process each event of the experimental data using two assignments: one where $x_1$ is designated as the tag jet and $x_2$ is designated as the probe jet, and the other with the roles reversed. 
    
    \item For each assignment, we evaluate the supervised classifier on the tag jet and use the result to label the probe jet.
    After applying both assignments, each jet is conditioned on its partner's predicted softmax output.

    \item We keep jets whose assigned outputs are 0.8 or higher for some flavor.
    After this cut, 25 million jets survive from the 256 million jets in the experimental dataset (i.e.~typically two jets from each of the 128 million dijet events). 
    The surviving jets are grouped according to the maximum-score flavor, resulting in seven mixtures. 
    Each mixture is subdivided into two roughly balanced pseudorapidity bins---a low-$\eta$ bin with $|\eta| \leq 0.75$, and a high-$\eta$ bin with $|\eta| > 0.75$---yielding 14 mixtures in total.
\end{enumerate}

From a probability perspective, our tag-and-probe method decomposes the marginal distributions as:
\begin{align}
    (\text{probe } x_1, \text{tag } x_2) \implies  p(x_1 | \tau ) &= \sum_{c, \eta} p(x_1 | c(x_2), \eta, \tau) \, p(\eta | c(x_2), \tau) \,  p(c(x_2)|\tau), \\
    (\text{probe } x_2, \text{tag } x_1) \implies p(x_2|\tau) &= \sum_{c, \eta} \, p(x_2 | c(x_1), \eta, \tau) \, p(\eta | c(x_1), \tau) \, p(c(x_1)|\tau),
\end{align}
where $\tau$ is the event that the tag jet passes the confidence threshold, $c(x)$ is the predicted \tsc{Pythia} flavor for jet $x$, and $\eta$ indexes the pseudorapidity bins.
Here, the conditional distributions $p(x_1|c(x_2), \eta, \tau)$ are the 14 mixtures obtained by one tag-and-probe assignment, whereas the distributions $p(x_2|c(x_1), \eta, \tau)$ are the mixtures from the opposite assignment.
Because the physics is invariant under rotations in the azimuthal angle $\phi$, and the feature space $x$ contains no information about the azimuthal angles used to order the jets, the joint distribution is symmetric: $p(x_1, x_2) = p(x_2, x_1)$.
Together with the fact that we use the same supervised classifier, confidence cut, and $\eta$ binning for both assignments, this implies that $p(x_1 | c(x_2), \eta, \tau)$ and $p(x_2|c(x_1), \eta, \tau)$ are asymptotically identical.
As a result, we can combine them into a single set of 14 mixtures, as described in step 4 above. 
For the sake of clarity, we refer to the 14 pooled mixtures as $p(x|c(x'), \eta, \tau)$ from now on.

With the mixtures $p(x|c(x'), \eta, \tau)$ in hand, we can decompose them into the constituent operational topics:
\begin{equation}
    p(x|c(x'),\eta,\tau) = \sum_t p(x|t, c(x'), \eta, \tau) \, p(t|c(x'), \eta, \tau).
\end{equation}
By the universal factorization in \Eq{eq:universal_topic_factorization} and assuming that the topic distributions do not depend on the probe jet's pseudorapidity bin, $p(x|t, c(x'), \eta, \tau) = p(x|t) \equiv p_t(x)$.
Furthermore, if we let $m$ index the mixtures such that $p(x|c(x'), \eta, \tau) \equiv p_m(x)$, then $p(t|c(x'), \eta, \tau) \equiv p(t|m) \equiv F_{mt}$.
Thus, we arrive at the mixture model required for simplex demixing:
\begin{equation}
    p_m(x) = \sum_t F_{mt} \, p_t(x).
\end{equation}
Therefore, by \Thrm{theorem:demixing}, training our simplex demixer on the 14 mixtures recovers the universal operational topics in the asymptotic limit, as long as $F_{mt}$ has full column rank.
Assuming that the \tsc{Pythia} flavors are (at least partially) correlated with the operational topics and that the topic co-occurrence matrix $p(t'|t) \, p(t)$ has nontrivial correlations, the full rank condition should be satisfied with real datasets.
Of course, the stronger the correspondence between \tsc{Pythia} flavors and topics, and the stronger the topic correlations between the two jets in a given dijet event, the more robust the topic extraction will be with finite datasets.

\begin{table}[t]
    \centering
    \begin{tabular}{c >{\hspace*{1.0cm}}l c c c c c c c >{\hspace*{0.0cm}}r}
    \toprule
    Mixture
    & \multicolumn{1}{c}{\raisebox{0.2cm}{$\Bigl($}
      \shortstack{partner \\ tag},
      \shortstack{probe \\ $\eta$ bin}
      \raisebox{0.2cm}{$\Bigr)$}}
    & $d$
    & $\bar{d}$
    & $u$
    & $\bar{u}$
    & $s$
    & $\bar{s}$
    & $g$
    & \multicolumn{1}{r}{\shortstack[r]{Jet count \\ ($\times 10^3$)}} \\
    \midrule
    1  & ($d$, low)        & $0.11$ & $0.05$ & $0.19$ & $0.03$ & $0.02$ & $0.02$ & $0.57$ & 728  \\
    2  & ($d$, high)       & $0.15$ & $0.04$ & $0.31$ & $0.02$ & $0.02$ & $0.01$ & $0.45$ & 749  \\
    3  & ($\bar{d}$, low)  & $0.17$ & $0.04$ & $0.15$ & $0.04$ & $0.02$ & $0.02$ & $0.55$ & 137  \\
    4  & ($\bar{d}$, high) & $0.20$ & $0.03$ & $0.29$ & $0.03$ & $0.02$ & $0.02$ & $0.41$ & 155  \\
    5  & ($u$, low)        & $0.13$ & $0.03$ & $0.16$ & $0.04$ & $0.03$ & $0.02$ & $0.59$ & 3181 \\
    6  & ($u$, high)       & $0.15$ & $0.03$ & $0.26$ & $0.03$ & $0.02$ & $0.01$ & $0.50$ & 3216 \\
    7  & ($\bar{u}$, low)  & $0.10$ & $0.05$ & $0.23$ & $0.03$ & $0.02$ & $0.02$ & $0.55$ & 462  \\
    8  & ($\bar{u}$, high) & $0.15$ & $0.03$ & $0.36$ & $0.02$ & $0.01$ & $0.02$ & $0.40$ & 546  \\
    9  & ($s$, low)        & $0.12$ & $0.03$ & $0.19$ & $0.02$ & $0.02$ & $0.09$ & $0.52$ & 1065 \\
    10 & ($s$, high)       & $0.16$ & $0.02$ & $0.34$ & $0.02$ & $0.01$ & $0.07$ & $0.37$ & 1292 \\
    11 & ($\bar{s}$, low)  & $0.11$ & $0.04$ & $0.15$ & $0.03$ & $0.09$ & $0.02$ & $0.55$ & 960  \\
    12 & ($\bar{s}$, high) & $0.15$ & $0.03$ & $0.28$ & $0.03$ & $0.07$ & $0.02$ & $0.42$ & 1088 \\
    13 & ($g$, low)        & $0.12$ & $0.04$ & $0.20$ & $0.03$ & $0.02$ & $0.02$ & $0.56$ & 5093 \\
    14 & ($g$, high)       & $0.17$ & $0.03$ & $0.35$ & $0.02$ & $0.01$ & $0.01$ & $0.40$ & 6289 \\
    \bottomrule
\end{tabular}
    \caption{Mixing fractions for the fourteen jet mixtures obtained using our tag-and-probe method.
    In the (partner tag, probe $\eta$ bin) column, the first entry is the flavor of the tag jet, and the second one is the pseudorapidity bin of the probe jet, either low ($|\eta| \leq 0.75$) or high ($|\eta| > 0.75$).
    Each row gives the flavor fractions for $d$, $\bar{d}$, $u$, $\bar{u}$, $s$, $\bar{s}$, and $g$ jets, rounded to two decimal places.
    The final column contains the number of jets in each mixture, rounded to the nearest thousand.}
    \label{tab:light_flavor_mixing_fractions}
\end{table}

To validate these assumptions on our experimental dataset, we show the compositions of the resulting 14 mixtures in terms of \tsc{Pythia} flavors, $G_{mc}$, in \Tab{tab:light_flavor_mixing_fractions}.
Notably, the relative quark/gluon fractions depend mainly on the $\eta$ bin.
The high-$\eta$ jet samples are quark-enriched, whereas the low-$\eta$ ones are gluon-enriched, matching previous studies \cite{Komiske:2022vxg}.
As one would expect from the valence quarks of the proton, the fractions of up-quark probe jets are approximately twice the fractions of down-quark probe jets across the samples.%
\footnote{Because of differences in tagging performance, the number of tagged up- and down-quark partner jets is not in this ratio.}
The sea quarks ($\bar{d}$, $\bar{u}$, $s$, and $\bar{s}$) make up smaller fractions throughout. 
The $s$ and $\bar{s}$ flavors have stronger correlations because a larger fraction of strange jets come from pair-production processes such as $gg \rightarrow s \bar{s}$ or $q \bar{q} \rightarrow s \bar{s}$. 
Pair-production processes also contribute to the $d/\bar{d}$ and $u/\bar{u}$ samples, but the signal is diluted by other channels involving valence quarks. 
The $\bar{d}$ and $\bar{u}$ flavors have small and near-uniform fractions across samples, so we expect them to be the hardest flavors to recover.

Given these fractions, we anticipate good support for the topics corresponding to $d$, $u$, $s$, $\bar{s}$, and $g$, with marginal support for the $\bar{d}$ and $\bar{u}$ topics.
Because the performance of simplex demixing is ultimately limited by statistics, we focus exclusively on dijets in this study, as they have the highest production cross section at hadron colliders like the LHC.
Nonetheless, other processes like vector-boson-plus-jet production offer a cleaner final state and might be worth studying as well in future work.

Finally, to prepare these mixtures for simplex demixing, we construct training, validation, and test datasets using the same 70\%/20\%/10\% split as in \Sec{subsec:event_generation}.
The training set contains 17.6 million jets in total.
In the training and validation sets, each mixture's distribution is normalized using class weights so that the $p_m(x)$ are valid probability densities.
When bootstrapping, the mixtures are renormalized accordingly, and all probe entries derived from the same dijet event are resampled together to prevent data leakage.

\subsection{Training the fourteen-mixture demixer} \label{sec:training_the_fourteen_mixture_demixer}

To train the simplex demixer on the fourteen mixtures, we select the same base classifier architecture (PFN) as in \Sec{sec:training_the_three_mixture_demixer}.
For tuning, we fix a training budget of 360 epochs in total, split into 150 epochs for stage one, 150 epochs for stage two, and 60 epochs for stage three.
We determine the optimal L2 strength $\gamma = 0$ (the training set is large enough that no L2 regularization is needed) and $\alpha = 5\times 10^{-5}$ using the same tuning procedure as in \Sec{sec:training_the_three_mixture_demixer}, replacing the 200/200 epoch split with a 150/150 epoch split.

Unlike the toy study, we also train stage three using the L1 loss to sparsify the representation of the classifier distribution. 
The L1 loss hyperparameter $\beta$ is scanned along a lasso path of 10 log-spaced values from 0.02 to 0.05. 
We pick $\beta$ to keep only $T=7$ active vertices, using our domain knowledge that there should be seven light flavors in the samples.
In most resamples, $T=7$ is achieved at $\beta = 0.0245$.
We discuss potential approaches to data-driven vertex-number selection in \Sec{sec:future_directions}.
Using fixed $\gamma$, $\alpha$, and $T$, we rerun the procedure on 34 bootstrap resamples to estimate the variance of the coefficient matrices $\widehat{p}(t|c)$ and $\widehat{p}(c|t)$, following the same procedure as in \Sec{sec:training_the_three_mixture_demixer} up to the epoch split.

\subsection{Identifying the seven operational topics} \label{sec:identifying_the_seven_operational_topics}

\begin{figure}[t]
\centering
\includegraphics[width=0.66\textwidth]{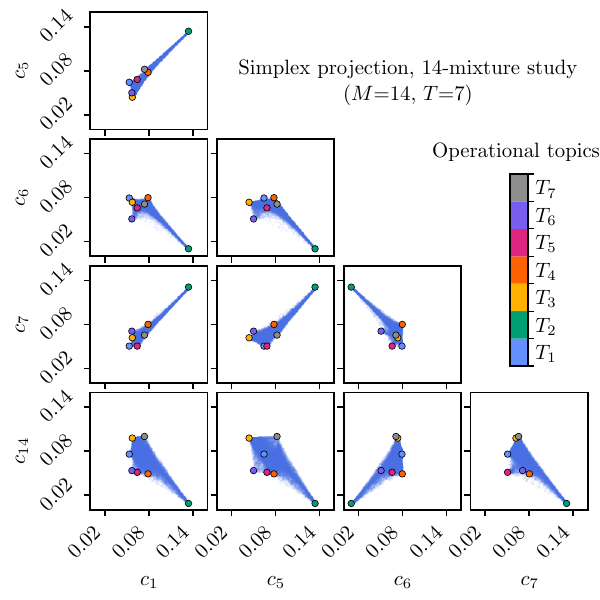}
\caption{\label{fig:fourteen_mixture_distribution}
The 14-mixture classifier distribution (in blue) is projected onto several two-dimensional planes.
Out of the $\binom{14}{2}$ unique projections, we plot only ten for clarity. 
The classifier distribution is bounded by $T=7$ vertices, which are identified by their color according to the bar on the right.
We observe that each vertex is an endpoint of the classifier distribution in at least one projection, so all seven vertices contribute to the distribution's convex hull.
It is important to note that we have restricted the projections to the $[0, 0.16] \times [0, 0.16]$ window, so the support of the classifier distribution is small compared to the bounding probability 13-simplex.
}
\end{figure}

To gain an intuition for the complexity of the demixing process, we first show the resulting classifier distribution for a single bootstrap resample in \Fig{fig:fourteen_mixture_distribution}.
Because the classifier is trained on $14$ samples, the classifier distribution lies inside the probability 13-simplex defined by $c_1 (x) + \ldots + c_{14} (x) = 1$ and $c_m (x) \geq 0$ for all $m$.
For visualization purposes, we project the distribution onto 10 representative planes out of the $\binom{14}{2}$ possible two-dimensional projections.
Within the probability 13-simplex, the individual jets (in blue) fill out a 6-simplex bounded by the seven learned vertices.
For each vertex, there is at least one projection in which that vertex clearly lies at an endpoint of the distribution, demonstrating that all seven vertices contribute to the convex hull.
Remarkably, the seven directions are visible even though the classifier distribution occupies only a small region inside the probability simplex, as shown by the small range of the axes.

\begin{figure}[t]
\centering

\vspace{1ex}

\subfloat[\label{fig:identifiable_fourteen_mixture_results_a}]{%
    \includegraphics[width=0.9\textwidth]{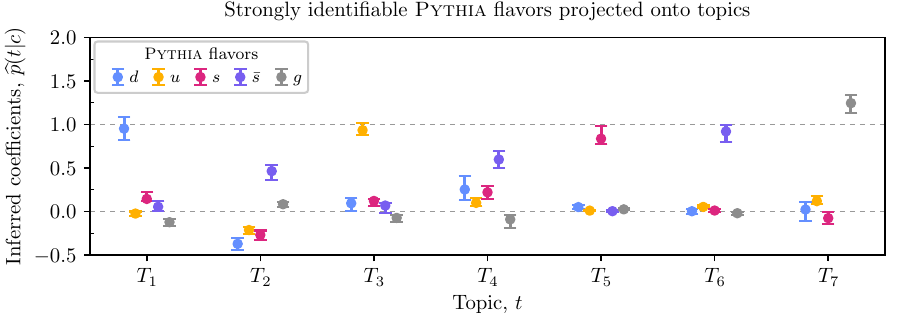}
}
\vspace{1ex}
\subfloat[\label{fig:identifiable_fourteen_mixture_results_b}]{%
    \includegraphics[width=0.9\textwidth]{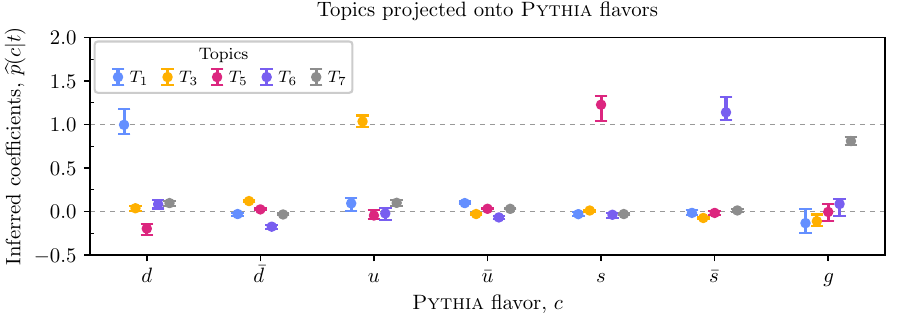}
}

\caption{\label{fig:identifiable_fourteen_mixture_results} 
Inferred (quasi-)probabilities for the strongly identifiable flavors and their associated topics.
(a) The $d$, $u$, $s$, $\bar{s}$, and $g$ rows of the inferred $\widehat{p}(t|c)$ matrix.
The rows exhibit near-one entries in the $T_1$, $T_3$, $T_5$, $T_6$, and $T_7$ topics, respectively, so we associate these topics with the corresponding flavors.
(b) The $T_1$, $T_3$, $T_5$, $T_6$, and $T_7$ rows of the $\widehat{p}(c|t)$ matrix.
As expected, these exhibit near-one entries in the $d$, $u$, $s$, $\bar{s}$, and $g$ flavors, which is an important closure test for strong identifiability.
Together, these panels show that the concepts of $d$, $u$, $s$, $\bar{s}$, and $g$ jets are latent in the mixtures. 
}
\end{figure}

\begin{figure}[t]
\centering

\vspace{1ex}

\subfloat[\label{fig:weak_fourteen_mixture_results_a}]{%
    \includegraphics[width=0.9\textwidth]{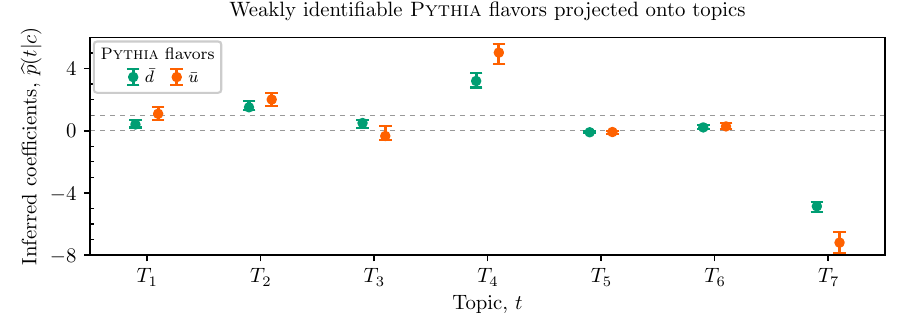}
}
\vspace{1ex}
\subfloat[\label{fig:weak_fourteen_mixture_results_b}]{%
    \includegraphics[width=0.9\textwidth]{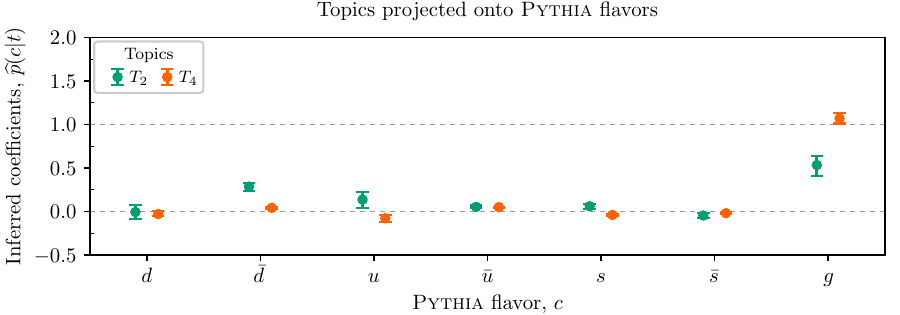}
}

\caption{\label{fig:weak_fourteen_mixture_results}
Inferred (quasi-)probabilities for the weakly identifiable flavors and their associated topics.
(a) The $\bar{d}$ and $\bar{u}$ rows of the inferred $\widehat{p}(t|c)$ matrix.
These two flavors map preferentially onto the $T_2$ and $T_4$ topics, though with substantial bleed-through, including negative representation in topic $T_7$.
(b) The $T_2$ and $T_4$ rows of the $\widehat{p}(c|t)$ matrix.
This projection shows that these operational topics are heavily contaminated by gluons.
Together, these panels show that the $\bar{d}$ and $\bar{u}$ flavors are only weakly identified by the operational topics.
}
\end{figure}

In \Fig{fig:identifiable_fourteen_mixture_results}, we show the estimated entries of the $\widehat{p}(t|c)$ and $\widehat{p}(c|t)$ matrices for the strongly identifiable \tsc{Pythia} flavors.
In \Fig{fig:identifiable_fourteen_mixture_results_a}, the probabilities that \tsc{Pythia} $d$, $u$, $s$, $\bar{s}$, and $g$ jets are mapped to the $T_1$, $T_3$, $T_5$, $T_6$, and $T_7$ vertices, respectively, are all close to one.
In \Fig{fig:identifiable_fourteen_mixture_results_b}, the quasi-probabilities that jets belonging to the $T_1$, $T_3$, $T_5$, $T_6$, and $T_7$ topics are \tsc{Pythia} $d$, $u$, $s$, $\bar{s}$, and $g$ jets, respectively, are also close to one.
These two directions establish a robust one-to-one correspondence between the \tsc{Pythia} $d$, $u$, $s$, $\bar{s}$, and $g$ flavors and the $T_1$, $T_3$, $T_5$, $T_6$, and $T_7$ operational topics.
Interestingly, in \Fig{fig:identifiable_fourteen_mixture_results_b}, the gluon-like $T_7$ topic exhibits some small but nonzero quark contamination, as discussed further in the next subsection.

By contrast, \Fig{fig:weak_fourteen_mixture_results} shows the entries of the $\widehat{p}(t|c)$ and $\widehat{p}(c|t)$ matrices for the weakly identifiable \tsc{Pythia} flavors.
In \Fig{fig:weak_fourteen_mixture_results_a}, the \tsc{Pythia} $\bar{d}$ and $\bar{u}$ flavors map preferentially onto the $T_2$ and $T_4$ vertices.
However, according to \Fig{fig:weak_fourteen_mixture_results_b}, there is a high quasi-probability that the jets assigned to the $T_2$ and $T_4$ topics are actually \tsc{Pythia} gluons. 
This implies that the $T_2$ and $T_4$ operational topics are heavily contaminated by gluons, as discussed in more detail below.

Overall, these results show that the notions of $d$, $u$, $s$, $\bar{s}$, and $g$ jets are latent in the mixtures, whereas $\bar{d}$ and $\bar{u}$ are only weakly identifiable. 
In the following subsection, we discuss the weak identifiability of $\bar{d}$ and $\bar{u}$ and the quark contamination in the gluon-like topic.

\subsection{Probabilistic interpretation of the demixing results} \label{sec:a_probabilistic_interpretation_of_the_demixing_results}

The weak identifiability of the $\bar{d}$ and $\bar{u}$ flavors is expected from the small truth-level \tsc{Pythia} fractions shown in \Tab{tab:light_flavor_mixing_fractions}.
When there are too few jets of a given flavor, none of the inferred topics resolves its anchor region in phase space.
As a result, the flavor's distribution can only be reconstructed as a difference of the resolved topics, which appears as negative entries in the noisy estimates of $\widehat{p}(t|c)$ and $\widehat{p}(c|t)$.
Therefore, we attribute the negative entries of the coefficient matrices to insufficient statistics.

The low $\bar{d}$ and $\bar{u}$ statistics can also provide a semi-quantitative understanding of why $T_2$ and $T_4$ are contaminated by gluons.
To see this, we restrict our attention to the subspace spanned by the $\bar{d}$-like vertex $T_2$ and the gluon-like vertex $T_7$.
Ignoring small contributions from the other topics, we write the likelihood ratio between \tsc{Pythia} gluons and \tsc{Pythia} anti-down quarks as 
\begin{equation}
    L_{g\bar{d}} (x) \equiv \frac{p_g(x)}{p_{\bar{d}} (x)} \approx \frac{p(T_7|g) \, p_{T_7}(x) + p(T_2|g) \, p_{T_2}(x)}{p(T_7|\bar{d}) \, p_{T_7}(x) + p(T_2|\bar{d}) \, p_{T_2}(x)} = \frac{p(T_7|g) \, L_{T_7T_2} (x) + p(T_2|g)}{p(T_7|\bar{d}) \, L_{T_7T_2} (x) + p(T_2|\bar{d})},
\end{equation}
where
\begin{equation}
    L_{T_7 T_2} (x) \equiv \frac{p_{T_7}(x)}{p_{T_2} (x)}.
\end{equation}
Since $p(T_7|g) > p(T_2|g)$ and $p(T_2 | \bar{d}) > p(T_7 | \bar{d})$, $L_{g\bar{d}}(x)$ is monotonic in $L_{T_7 T_2} (x)$.
By the mutual irreducibility of the operational topics, we have:
\begin{equation}
    \min_{x} L_{g\bar{d}} (x)  = \frac{p(T_2|g)}{p(T_2|\bar{d})} = \frac{p(g|T_2)}{p(\bar{d}|T_2)} \frac{ p(\bar{d})}{p(g)},
\end{equation}
where we have used Bayes' theorem in the second equality.
If there are not enough instances of \tsc{Pythia} $\bar{d}$ jets in the entire jet sample, then the $\bar{d}$ anchor region is not perfectly resolved and therefore the left-hand side of the above equation is nonzero.
Given that there are many more gluon jets than $\bar{d}$ jets in the mixtures, we have $p(g) \gg p(\bar{d})$ and the ratio $p(g | T_2) / p(\bar{d} | T_2)$ becomes large, implying that the $T_2$ topic is heavily contaminated by gluons.
The $p(g | T_4) / p(\bar{u}|T_4)$ ratio is even larger because $\bar{u}$ jets are rarer and have weaker correlations in the mixtures.
To make the $\bar{d}$ and $\bar{u}$ fully identifiable, we would require either stronger correlations of $\bar{d}$ and $\bar{u}$ in the mixtures or more training data.

We can use a similar argument to explain why the $T_7$ vertex (primarily gluons) is contaminated by quarks of various flavors.
We can write:
\begin{equation}
    L_{qg}(x) \equiv \frac{p_q(x)}{p_g(x)} = \frac{p(T_q|q) \, p_{T_q}(x) + p(T_7|q) \, p_{T_7}(x)}{p(T_q|g) \, p_{T_q}(x) + p(T_7|g)\, p_{T_7}(x)} = \frac{p(T_q|q) \, L_{T_qT_7}(x) + p(T_7|q) }{p(T_q|g) \, L_{T_qT_7} (x) + p(T_7|g)},
\end{equation}
where $p_q(x)$ is the distribution of quarks in the mixtures and $p_{T_q}(x)$ is a (suitably weighted) combination of quark-like topic distributions.
We have $p(T_q|q) > p(T_7|q)$ and $p(T_7 | g) > p(T_q|g)$, so $L_{qg}(x)$ is monotonic in $L_{T_q T_7} (x)$, and thus:
\begin{equation}
    \kappa_{qg} = \frac{p(T_7|q)}{p(T_7|g)} = \frac{p(q|T_7)}{p(g|T_7)} \frac{p(g)}{p(q)}.
\end{equation}
As discussed in \Sec{sec:identifying_the_three_operational_topics}, the intrinsic reducibility factor $\kappa_{qg}$ is nonzero for \tsc{Pythia} jets~\cite{Larkoski:2019nwj}, so the $p(q|T_7)/p(g|T_7)$ ratio is nonzero and, as a consequence, the gluon-like operational topic necessarily contains a nonzero quark component.
This contamination, unlike the gluon contamination in $T_2$ and $T_4$, does not decrease as one increases the size of the dataset because it is caused by a systematic mismatch between the definitions of \tsc{Pythia} flavors and operational topics.

\subsection{Extracting substructure distributions}
\label{subsec:substructure_distributions}

Having established the relationship between \tsc{Pythia} flavors and operational topics in \Secs{sec:identifying_the_seven_operational_topics}{sec:a_probabilistic_interpretation_of_the_demixing_results}, we now extract observable distributions for each operational topic and compare them to those of the corresponding \tsc{Pythia} flavors.
To solve for the distributions of operational topics, we invert \Eq{eq:general_mixture_decomposition} in a least-squares sense:
\begin{equation} \label{eq:invert_general_mixture_decomposition}
    p_t(x) = \sum_{m=1}^M (F^+)_{tm} \, p_m(x),
\end{equation}
where $(F^+)_{tm}$ denotes the pseudoinverse of $F_{mt}$, and $p_m(x)$ are the measured mixture distributions.

\begin{figure}[p]
\centering

\subfloat[\label{fig:substructure_distributions_a}]{%
    \includegraphics[width=0.48\textwidth]{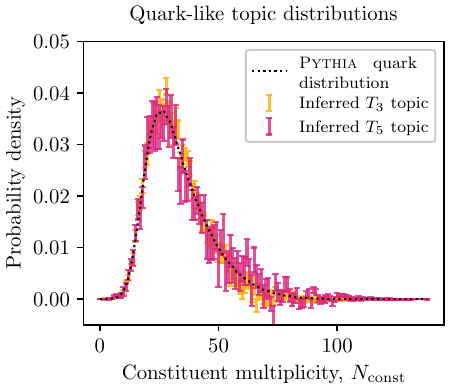}
}
\hfill
\subfloat[\label{fig:substructure_distributions_b}]{%
    \includegraphics[width=0.48\textwidth]{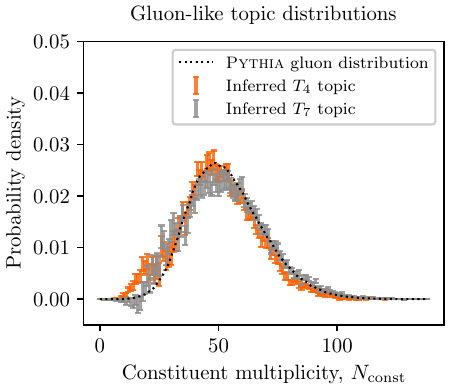}
}

\vfill

\subfloat[\label{fig:substructure_distributions_c}]{%
    \includegraphics[width=0.48\textwidth]{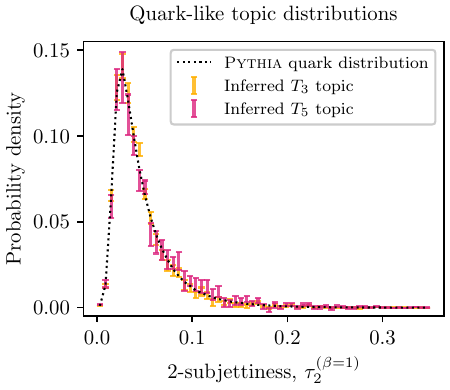}
}
\hfill
\subfloat[\label{fig:substructure_distributions_d}]{%
    \includegraphics[width=0.48\textwidth]{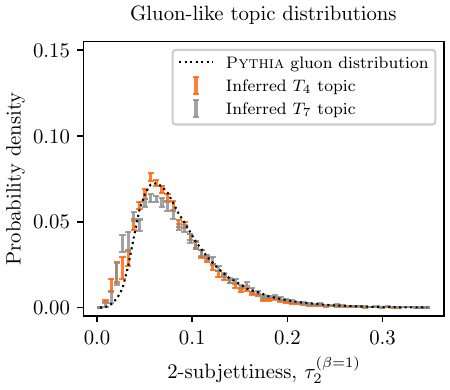}
}

\vfill

\subfloat[\label{fig:substructure_distributions_e}]{%
    \includegraphics[width=0.48\textwidth]{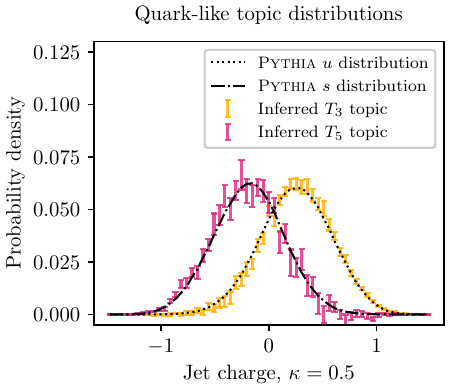}
}
\hfill
\subfloat[\label{fig:substructure_distributions_f}]{%
    \includegraphics[width=0.48\textwidth]{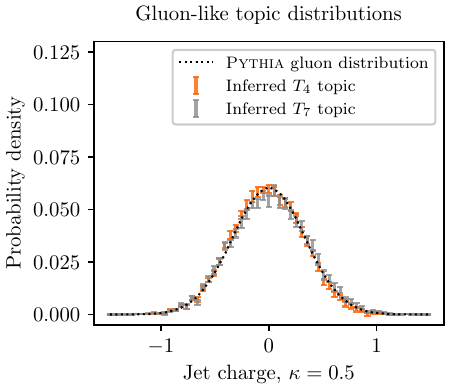}
}

\caption{\label{fig:substructure_distributions}
Distributions of (top row) constituent multiplicity $N_\text{const}$, (middle row) 2-subjettiness $\tau_2^{(\beta=1)}$, and (bottom row) jet charge $\mathcal{Q}_{\kappa=0.5}$, comparing the distributions of \tsc{Pythia} flavors to those of the extracted operational topics.
(Left column) The quark-like $T_3$ and $T_5$ topics are in good agreement with the \tsc{Pythia} $u$ and $s$ flavors, respectively.
(Right column) The gluon-like $T_4$ and $T_7$ topics are in good agreement with \tsc{Pythia} gluon distributions, despite $T_4$'s alignment with the \tsc{Pythia} $\bar{u}$ flavor.
The remaining three topics are shown in \Fig{fig:substructure_distributions_backup}.
}
\end{figure}

In our study, $x$ is the full representation of a jet in terms of its variable-length set of particle constituents.  
Because the mixing fractions are obtained from this highly expressive representation, they are more general than those obtained from any individual observable.
We can therefore use the same fractions to reconstruct the operational topic distributions over \emph{any} observable $\mathcal{O}$.
To do so, we simply apply the inversion in \Eq{eq:invert_general_mixture_decomposition} to the mixture distributions over $\mathcal{O}$, $p_m(\mathcal{O})$:
\begin{equation} \label{eq:invert_general_mixture_decomposition_observable}
    p_t(\mathcal{O}) = \sum_{m=1}^M (F^+)_{tm} \, p_m (\mathcal{O}).
\end{equation}

To validate this approach, we consider two representative observables for quark/gluon discrimination: constituent multiplicity $N_\text{const}$ and 2-subjettiness $\tau_2^{(\beta=1)}$~\cite{Thaler:2010tr, Thaler:2011gf}.
We also consider an additional observable sensitive to light quark flavor: the jet charge $\mathcal{Q}_{\kappa=0.5}$ \cite{Krohn:2012fg, EuropeanMuon:1984xji}. 
We compute constituent multiplicity using a custom implementation,  2-subjettiness using the \tsc{Nsubjettiness} 2.3.2 module \cite{fastjet_contrib}, and jet charge using the definition provided in \Reference{Krohn:2012fg}.

For each observable, we obtain the binned distributions of each mixture and bootstrap resample of the training dataset.
Then, using \Eq{eq:invert_general_mixture_decomposition_observable}, we apply the inversion bin by bin to the histogrammed mixture distributions.
Finally, we report the per-bin medians of the $p_t(\mathcal{O})$ distributions across bootstrap iterations together with the per-bin 15\%--85\% interval as a measure of statistical uncertainty:
\begin{equation}
    \big[ Q_{0.15}\left( p_t(\mathcal{O}) \right), Q_{0.85} \left( p_t(\mathcal{O}) \right) \big].
\end{equation}

In \Fig{fig:substructure_distributions}, we show the inferred topic distributions over the three substructure observables defined above:  $N_\text{const}$, $\tau_2^{(\beta=1)}$, and $\mathcal{Q}_{\kappa=0.5}$. 
For the sake of readability, we only show two quark-like topics, $T_3$ (\tsc{Pythia} up) and $T_5$ (\tsc{Pythia} strange), and two gluon-like topics, $T_4$ (\tsc{Pythia} anti-up) and $T_7$ (\tsc{Pythia} gluon); the remaining topic distributions are plotted and discussed in \App{sec:substructure_distributions_for_the_remaining_topics}.
We emphasize that $T_4$, despite acting as a sink for $\bar{u}$ jets, yields a gluon-like distribution because it is heavily contaminated by gluons, as discussed in \Sec{sec:identifying_the_seven_operational_topics} and shown in \Fig{fig:weak_fourteen_mixture_results_b}.

All three rows of \Fig{fig:substructure_distributions} show good agreement with the truth-level \tsc{Pythia} distributions.
In \Fig{fig:substructure_distributions_d}, the $T_7$ 2-subjettiness distribution exhibits some broadening to the left of the peak, consistent with the quark contamination of the gluon topic shown in \Fig{fig:identifiable_fourteen_mixture_results_b}.
This ``broadening'' feature is also visible, albeit to a lesser degree, in \Fig{fig:substructure_distributions_b}.
The jet charge distributions in \Fig{fig:substructure_distributions_e} show a striking separation between the $T_3$ and $T_5$ topics, matching the positive and negative charges of the $u$ and $s$ quarks, respectively.
In contrast, the gluon-like jet charge distributions in \Fig{fig:substructure_distributions_f} are centered around zero, consistent with the gluons having zero electric charge.
Overall, these results demonstrate the potential of simplex demixing for performing data-driven extractions of multiple jet flavor substructure distributions at the LHC.

\subsection{Topic flow along the regularization path} \label{sec:topic_flow_along_the_regularization_path}

\begin{figure}[t]
\centering

\vspace{1ex}

\includegraphics[width=0.66\textwidth]{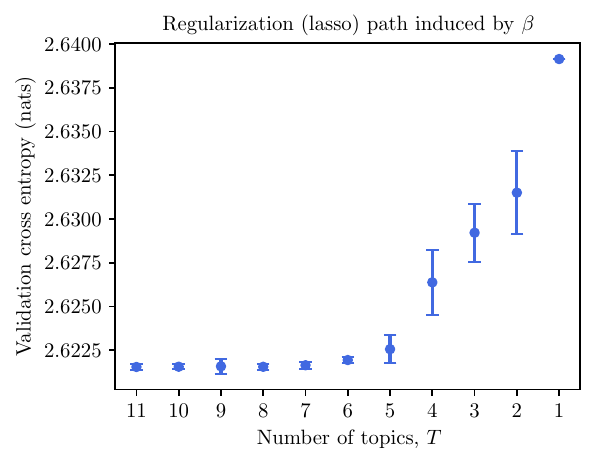}

\caption{\label{fig:val_ce_vs_t}
Lasso path as a function of the hyperparameter $\beta$, translated into the number of obtained topics $T$.
The curve is nearly flat at high $T$ and starts growing sharply at $T\leq4$.
}
\end{figure}

\begin{figure}[p]
\centering

\vspace{-2cm}

\includegraphics[width=0.83\textwidth]{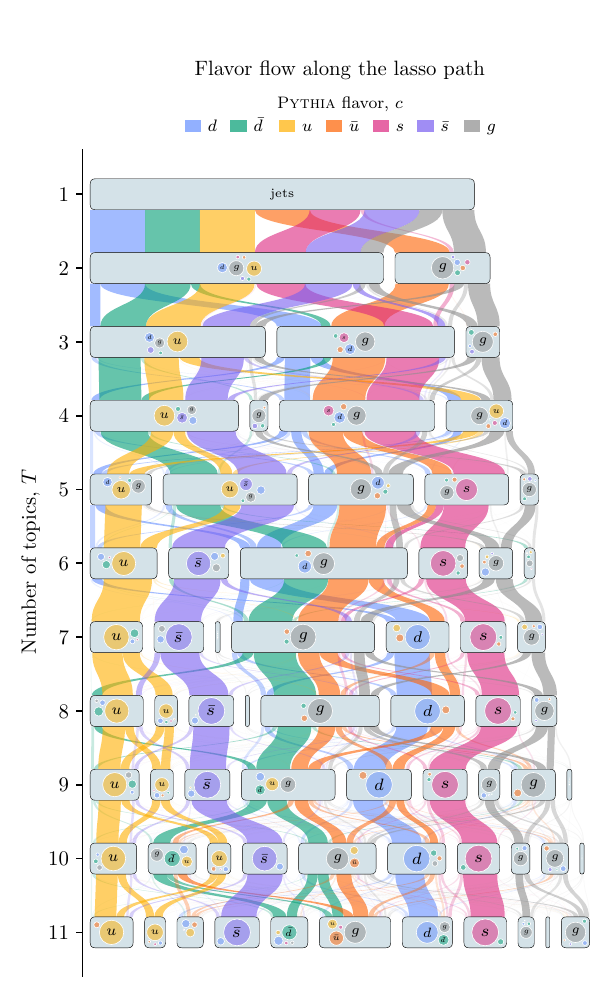}

\caption{\label{fig:topic_hierarchy}
Flow of operational topics in the quasi-realistic study as the regularization parameter $\beta$ is decreased (number of topics $T$ increased). 
On each level, each gray rectangle corresponds to a different vertex.
The bubbles inside each vertex $t$ have areas proportional to the entries of $\widehat{p}(c|t)$.
There are seven flows for $\widehat{p}(t, t'|c)$, showing how the jets of each flavor flow from model $T$ to model $T+1$.
We find that the \tsc{Pythia} $d$, $u$, $s$, $\bar{s}$ and $g$ flavors remain strongly identifiable for $7 \leq T \leq 11$.
At $T=2$, the distribution simplifies to quark-like and gluon-like regions of phase space.
}
\end{figure}

Thus far, we have hardcoded the number of vertices $T = 7$ using our prior physics knowledge. 
However, $T = 7$ is merely a single point along the regularization (or lasso) path induced by the hyperparameter $\beta$.
To view the full path, we scan 26 log-spaced $\beta$ values over the range $[7\cdot 10^{-3}, 1.4\cdot 10^{-1}]$ at fixed $\alpha$ for each bootstrap resample. 
Then, we average the validation cross entropy over the resamples that attain each $T$.
Note that, because we scan $\beta$ over a discrete grid, not all resamples attain every value of $T$.
In \Fig{fig:val_ce_vs_t}, we show the cross-entropy averages together with the corresponding standard deviations, displayed as error bars. 
The curve remains nearly constant from $T = 11$ through $T=7$, starts increasing at $T = 6$ and $T = 5$, and deteriorates more substantially for $T \leq 4$.
Although this curve does not provide a rigorous model selection criterion, it does show that our physics-motivated choice lies on the plateau of models with comparable performance.

Given this regularization path, it is interesting to examine whether the vertices retain a valid physics interpretation for other models on the plateau, $T=11$ through $T=7$, as well as for simpler models, $T \leq 6$.
To do so, we would like to visualize the entries of the $\widehat{p}(t|c)$ and $\widehat{p}(c|t)$ matrices for each model $T$.

In \Fig{fig:topic_hierarchy}, we show the information contained in the projection matrices across different model sizes.
For each vertex $t'$ and model $T$, the entries of $\widehat{p}(c|t')$ are encoded as bubbles inside the rectangular node corresponding to $t'$.
There is one bubble for each \tsc{Pythia} flavor $c'$; its area is proportional to the bootstrap median of the quasi-probability $\widehat{p}(c'|t')$ when this is positive, and it is zero otherwise. 
Physically, the largest bubbles indicate the dominant shape of the topic $t'$ distribution.
The $\widehat{p}(t|c)$ matrix does not admit a similarly direct representation because the $\widehat{p}(t|c)$ values are properties of the flavors, not of the vertices.
Instead, we encode the $\widehat{p}(t|c)$ values as flows of width $\widehat{p}(t, t'|c)$ between pairs of vertices $(t, t')$ in consecutive models $(T, T+1)$.
Intuitively, the flows describe the proportion of jets of each truth-level flavor that move from $t$ to $t'$.

To calculate the $\widehat{p}(t, t'|c)$ entries, we write:
\begin{align}
    p(t, t'|c) &= \int dx \, p(t, t'|x, c) \, p(x|c) \\
    &\approx \frac{1}{N_c} \sum_{j \in c} p(t, t'|x_j, c) \\ 
    &= \frac{1}{N_c} \sum_{j \in c} p(t|x_j, c) \, p(t'|x_j, c), \label{eq:flow_probabilities}
\end{align}
where $x_j$ refers to the properties of the $j$-th jet of flavor $c$ in the feature space $x$, and $N_c$ is the total number of jets of flavor $c$.
In the second line, we have inserted the empirical distribution of flavor $c$ over feature space. 
In the third line, we have chosen to construct the joint distribution $p(t, t'|x_j, c)$ as the product of marginals, $p(t|x_j, c)$ and $p(t'|x_j, c)$.
We emphasize here that $p(t|x_j, c)$ is, in general, not equal to $p(t|x_j)$.
In particular, $p(t|x_j, c)$ updates the $p(t|x_j)$ probabilities based on the correlations between truth-level flavors and operational topics inferred from the mixing matrix $F_{mt}$.
Under the assumption of topic universality, this Bayesian update is:
\begin{equation}
    p(t|x_j, c) = \frac{p(t|x_j) \, \left(\frac{p(t|c)}{p(t)}\right)}{\sum_{t'} p(t'|x_j) \left(\frac{p(t'|c)}{p(t')}\right)} \approx 
    \frac{\lambda_t (x_j) \left(\frac{\widehat{p}(t|c)}{\sum_{c'} \widehat{p}(t|c')p(c')}\right)}{\sum_{t'} \lambda_{t'} (x_j) \left(\frac{\widehat{p}(t'|c)}{\sum_{c'} \widehat{p}(t'|c')p(c')}\right)}.
\end{equation}
In the approximation, we have used that $p(t|x_j) = \lambda_t(x_j)$: the function learned by the simplex demixer, $\lambda_t(x)$, evaluated at $x_j$.
In this expression, the flavor prior $p(c)$ is:
\begin{equation}
    p(c) = \frac{1}{M} \sum_m G_{mc},
\end{equation}
where $G_{mc}$ is the mixing matrix of \tsc{Pythia} flavors in the mixtures.

To estimate $\widehat{p}(t, t'|c)$ on our fourteen-mixture dataset, we compute $\widehat{p} (t|c)$, $\lambda_t(x_j)$, and $\lambda_{t'}(x_j)$ for each bootstrap resample that attains models $T$ and $T+1$ at different values of $\beta$ and on each $x_j$ in a fixed sample of 5,000 test-set jets of flavor $c$.
The $\widehat{p}(t|c)$ are obtained from a least-squares fit to \Eq{eq:p(t|c)}, subject to non-negativity and normalization-to-one constraints to ensure that the resulting flow widths are nonnegative.
We note that this is different than the unconstrained fit used for \Figs{fig:identifiable_fourteen_mixture_results}{fig:weak_fourteen_mixture_results}.
Finally, we compute $\widehat{p}(t, t'|c)$ using \Eq{eq:flow_probabilities} for each bootstrap and report the width of each flow as the mean over bootstraps.

According to \Fig{fig:topic_hierarchy}, the models on the plateau of the lasso path ($T \in [7,11]$) have largely consistent physical interpretations.
In particular, all of these models contain strong $d$, $u$, $s$, $\bar{s}$, and $g$ vertices.
There are also $\bar{d}/\bar{u}$ vertices in each model, but the assignments are unstable due to the weak identifiability of the flavors.
Interestingly, the strong $u$ and $d$ vertices absorb a significant portion of $\bar{d}$ and $\bar{u}$ jets, respectively, effectively aggregating the weak flavors by charge.
The upper half of \Fig{fig:topic_hierarchy} shows that further increasing $\beta$ induces a sensible clustering of flavors.
Notably, at $T=4$, $d$ merges with $s$, and $u$ merges with $\bar{s}$, which is expected due to the electric charges of the quarks.
Finally, at $T=2$, we are left with gluon-dominated and quark-dominated regions of phase space.

Overall, these results show that the lasso loss drives the simplex demixer through a physically interpretable path, shaped by the intrinsic separability of the different flavors and their abundances in the mixtures.
Simplex demixing identifies useful regions of phase space for defining jet flavor, even in the absence of a domain-agnostic procedure for selecting $T$.
Despite this success, a caveat is in order: we have used the same value of $\alpha$ for every value of $\beta$, even though $\alpha$ was tuned on the original, unregularized geometry and therefore need not remain optimal for smaller $T$. 
We leave a systematic study of the $\alpha$ dependence of the lasso path to future work.

\subsection{Future directions} \label{sec:future_directions}

Despite the success of simplex demixing on our proof-of-concept physics application, a number of issues deserve further study before deployment at the LHC.
We briefly mention three issues related to choosing the number of topics, increasing the effective sample size, and handling realistic detector capabilities, which we plan to study in future work.

In \Sec{sec:topic_flow_along_the_regularization_path}, we showed that the anchor regions identified by simplex demixing retain a useful physics interpretation for different numbers of vertices $T$.
That said, it would still be useful to incorporate simplex demixing into a domain-agnostic procedure for choosing $T$, especially as the number of mixtures $M$ grows.
In general, choosing the number of features or clusters in a dataset is a challenging problem and an active area of research.
Common approaches for determining $T$ in other methods include the scree plot, or elbow method, which selects $T$ by identifying a kink in the residual-versus-$T$ curve \cite{eugster2009spider}. 
Information-theoretic metrics formalize this by balancing goodness of fit against model complexity when selecting the optimal $T$ \cite{schwarz1978estimating, suleman2017validation}.
Cross-validation can assess model performance on held-out validation data \cite{10.1214/08-AOAS227}, while clustering methods often estimate $T$ using the average silhouette score or the gap statistic \cite{ROUSSEEUW198753, tibshirani2001gap}.
Finally, stability metrics select the relevant features based on their reproducibility under bootstrap resampling or subsampling \cite{yu2019bootstrapping, meinshausen2010stability, 2834535}. 
We believe that \References{meinshausen2010stability, 2834535} are particularly well-suited to our simplex demixing architecture, which uses the lasso to perform feature selection.

In terms of the effective sample size, we based this study on the projected integrated luminosity of the HL-LHC. 
If we instead restricted ourselves to Run 3 statistics, we would expect around 15 million dijet events satisfying the stringent cuts described in \Sec{subsec:tag_and_probe} \cite{CMS:2023fix, ATLAS:2017ble}.
We found that running simplex demixing on these smaller datasets increased the finite-statistics artifacts.
To increase the effective statistics, one could of course move the jet selection cut to a lower $p_T$ range, at the expense of more soft contamination in jets.
Alternatively, instead of discarding events by imposing a hard 0.8 threshold on the supervised classifier output, it would be interesting to find a way to augment our method to also include probe jets with less certain tag-label assignments.

Finally, in terms of detector capabilities, our results demonstrate what would be possible with perfect detector information, including $\pi/K/p$ discrimination.
In reality, the ATLAS and CMS detectors have limited $\pi/K$ separation capabilities, especially at high $p_T$.
Without $\pi/K$ separation, the \tsc{Pythia} light quark flavors become less distinguishable, and some operational topics may lose their resolvable anchor regions in phase space.
For example, $d$ versus $s$ discrimination relies heavily on the larger ratio of charged kaons to charged pions in $s$-quark jets \cite{Nakai:2020kuu}. 
The momentum-weighted fraction of neutral kaons (in particular $K_S \to \pi^+ \pi^-$) still provides a means to distinguish $d$ from $s$ but is insufficient on its own to produce anchor regions in phase space. 
In \App{sec:ten_mixture_study_with_pfn_ex_information}, we show how the performance of our tag-and-probe method degrades without $\pi/K/p$ separation.
Additional experimental information is available at ALICE and LHCb, or at ATLAS and CMS at low $p_T$ \cite{ALICE:2020nkc, Calabrese:2022eju, CMS:2017eoq}.
Along with studies of strangeness, it would be interesting to incorporate heavy-flavor jets into the simplex demixing framework.

\section{Conclusions}
\label{sec:conclusions}

In this paper, we proposed a data-driven strategy for extracting multiple light-flavor categories from jets at the LHC.
The key to our approach is a novel simplex demixing framework, which extends the operational definition of quark and gluon jets to any number of jet flavors $T$ in any number of jet samples $M$, provided that $T \leq M$. 
In the asymptotic limit, we proved that a classifier trained on the $M$ jet samples outputs points within a $(T-1)$-simplex in classifier space, whose $T$ vertices correspond to the latent jet flavors.
The positions of these vertices allow us to solve for the mixing fractions of the operational topics in the samples, and thereby extract the distributional properties of light-flavor jets.

With this proof in hand, we designed a practical three-stage machine-learning procedure to demix the samples, compatible with any classifier architecture and robust to finite statistics. 
In the first stage, the classifier learns a good representation of the data.
In the second stage, an additional edge loss is imposed to corral the bounding simplex of the learned distribution. 
In the third stage, an additional L1 loss is imposed to remove unnecessary vertices and retain only those corresponding to the operational topics.
Given a benchmark of known categories, the final result can be interpreted as a (quasi-)probability that a jet of a known category is assigned to a particular topic, and vice versa.

After validating simplex demixing on toy datasets with down, up, and gluon jets in varying proportions, we applied it to dijet events from \tsc{Pythia}.
This proof-of-concept study involved realistic QCD correlations and HL-LHC statistics, though it (unrealistically) assumed access to perfect detector information.
Using a tag-and-probe strategy, we translated a supervised jet tagger into an unsupervised approach for topic modeling.
We found that the learned 14-category classifier distribution was bounded by a 6-simplex, consistent with having seven light-flavor jets.
The down, up, strange, anti-strange, and gluon flavors were strongly identifiable in the data.
The anti-down and anti-up flavors were also identifiable, but only weakly, due to their weak fraction correlations in the dataset and contamination from gluon jets.

Looking to the LHC, we hope simplex demixing can help facilitate data-driven extractions of light-flavor jet categories, together with their individual substructure-observable distributions.
Tagging individual light-flavor jets is notoriously difficult, even with fully supervised techniques.
Nevertheless, our results demonstrate that flavor categories can emerge at the ensemble level, without needing to rely on an ambiguous parton-level jet flavor definition.
To better understand what features drive the mutual irreducibility between light-quark flavors, it would be interesting to examine the per-particle $\Phi$ layer learned by the inner PFN.
In addition, it would be interesting to study the sample dependence of jet flavor by comparing the operational jet flavors extracted from different jet samples.

Beyond particle physics, simplex demixing provides a practical way to extract mutually irreducible topics latent in mixed datasets for broader machine-learning and data-science applications. 
Whereas most approaches to topic modeling focus on demixing probability distributions over discrete feature spaces, our approach works for continuous and complex spaces, including for variable-dimensional point clouds. 
More philosophically, simplex demixing offers an interesting interpretation of classification with mixed data.
A classifier trained on mixed samples is simply a ``rescaled'' version of a classifier trained on pure samples. 
By choosing working points near the vertices of the learned simplex, rather than the vertices of the probability simplex, we can still assign data points to underlying pure categories.
In this context, it would be interesting to test whether simplex regularization could improve classification performance on datasets with label noise.

\section*{Code Availability}

Our simplex demixing implementation can be found at \Reference{SimplexDemixingCode}, and the three-mixture datasets are at \Reference{SimplexDemixingDatasets}.  

\acknowledgments

We thank Sean Benevedes for helpful discussions, and Patrick Komiske and Eric Metodiev for early brainstorming on multi-category jet topics and subsequent feedback.
We thank Johann Ioannou-Nikolaides and Rapha\"{e}l Bonnet-Guerrini for coordinating the submission of our manuscripts.
This work was supported by the U.S.\ National Science Foundation (NSF) under Cooperative Agreement PHY-2019786 (The NSF AI Institute
for Artificial Intelligence and Fundamental Interactions,
\url{http://iaifi.org/}), by the U.S.\ Department of
Energy (DOE) Office of High Energy Physics under grant
number DE-SC0012567, and by the Simons Foundation through Investigator grant 929241.
JT also thanks the Institut des Hautes \'Etudes Scientifiques (IHES) and the Institut de Physique Th\'eorique (IPhT) for providing an inspiring sabbatical environment to carry out this research.
The computations in this paper were performed using resources from the FASRC Cannon cluster, supported by the FAS Division of Science Research Computing Group at Harvard University, and the SubMIT system \cite{Acosta:2025kgk}, built by the MIT Physics Department.
Claude Opus 4.8 was used for helpful discussions that improved the presentation of several arguments, for streamlining the implementation of bootstrap resampling, and for proofreading; the authors take full responsibility for the content of the manuscript.

\appendix

\begin{figure}[p]
\centering

\subfloat[\label{fig:substructure_distributions_backup_a}]{%
    \includegraphics[width=0.48\textwidth]{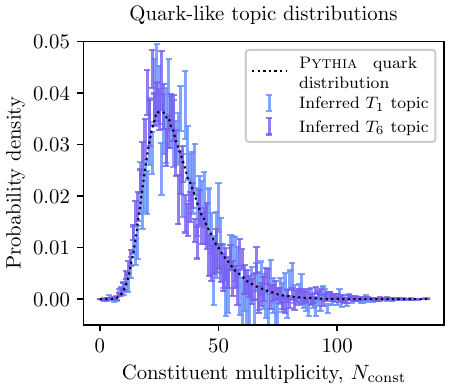}
}
\hfill
\subfloat[\label{fig:substructure_distributions_backup_b}]{%
    \includegraphics[width=0.48\textwidth]{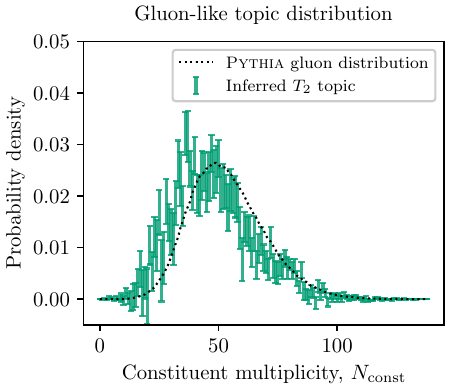}
}

\vfill

\subfloat[\label{fig:substructure_distributions_backup_c}]{%
    \includegraphics[width=0.48\textwidth]{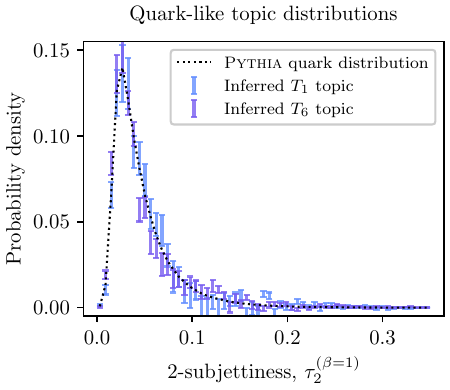}
}
\hfill
\subfloat[\label{fig:substructure_distributions_backup_d}]{%
    \includegraphics[width=0.48\textwidth]{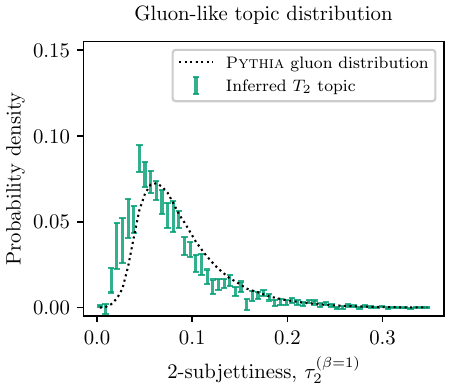}
}

\vfill

\subfloat[\label{fig:substructure_distributions_backup_e}]{%
    \includegraphics[width=0.48\textwidth]{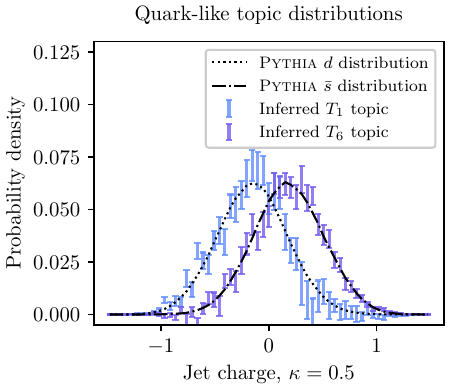}
}
\hfill
\subfloat[\label{fig:substructure_distributions_backup_f}]{%
    \includegraphics[width=0.48\textwidth]{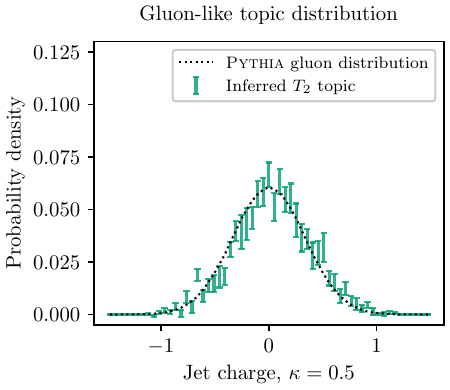}
}

\caption{\label{fig:substructure_distributions_backup}
The same as \Fig{fig:substructure_distributions}, but for the remaining three topics.
(Left column) The quark-like topics $T_1$ and $T_6$, whose distributions are in good agreement with the \tsc{Pythia} $d$ and $\bar{s}$ flavors, respectively.
(Right column) The gluon-like $T_2$ topic, which agrees with the \tsc{Pythia} gluon distributions despite its alignment with the \tsc{Pythia} $\bar{d}$ flavor.
}
\end{figure}

\section{Substructure distributions for the remaining topics} \label{sec:substructure_distributions_for_the_remaining_topics}

In this appendix, we show the jet observable distributions for the remaining topics not already shown in \Sec{subsec:substructure_distributions}.
In \Fig{fig:substructure_distributions_backup}, we show the $T_1$ (\tsc{Pythia} down), $T_2$ (\tsc{Pythia} anti-down), and $T_6$ (\tsc{Pythia} anti-strange) distributions, compared to their \tsc{Pythia} counterparts. 
The features are broadly similar to those of \Fig{fig:substructure_distributions}.
The $T_1$ and $T_6$ operational topics have quark-like shapes in both constituent multiplicity and 2-subjettiness. 
$T_1$ exhibits a negative jet charge, while $T_6$ exhibits a positive jet charge, matching the negative charge of the down quark and the positive charge of the anti-strange quark.
By contrast, $T_2$ has a gluon-like shape in all three observables. 
Much like $T_4$, which receives $\bar{u}$ jets but is contaminated by gluons, $T_2$ receives $\bar{d}$ jets but suffers gluon contamination, though to a lesser degree. 
The $\widehat{p}(\bar{d}|T_2)$ coefficient, shown in \Fig{fig:weak_fourteen_mixture_results_b}, hovers around 0.3, which is consistent with the slight quark-like distortions of the $T_2$ distributions in \Figs{fig:substructure_distributions_backup_b}{fig:substructure_distributions_backup_d}.

\section{Ten-mixture study with PFN-Ex information} \label{sec:ten_mixture_study_with_pfn_ex_information}

In this appendix, we apply the quasi-realistic method for operationally defining light-jet flavors, described in \Sec{sec:realistic_physics_application}, using experimentally realistic particle ID information. 
Specifically, we use the PFN-Ex scheme, introduced in \Reference{Komiske:2018cqr}, which distinguishes the hadron charges but does not include $\pi/K/p$ separation.
In the absence of $\pi/K/p$ discrimination, we expect some \tsc{Pythia} flavors to lose their anchor regions in phase space, thereby increasing the mismatch between the \tsc{Pythia} flavors and the mutually irreducible operational topics.

We start with the same pure samples and ``experimental" dataset as in \Sec{subsec:tag_probe_event_generation}. 
We train a supervised classifier on the seven samples using the same loss and PFN architecture as in \Sec{sec:training_the_three_mixture_demixer}, but with PFN-Ex preprocessing instead.
The classifier distribution forms a sub-simplex within the probability 6-simplex.
Due to the smaller geometry, no jet receives a probability greater than 0.8 for certain outputs of the supervised classifier.
However, such confident assignments are required to construct the correlated jet mixtures used for downstream simplex demixing.
To address this, we replace the standard supervised classifier with a simplex demixer that tightly encloses the distribution.
The demixer's base classifier can then be used to confidently assign jets to each of the resulting topics (vertices).
Conceptually, using the topics extracted from the pure samples for our tag-and-probe method is valid, as long as they are correlated with the topics latent in the ``experimental" data.

\begin{figure}[t]
\centering

\subfloat[\label{fig:t5_pct_ex_entries_a}]{%
    \includegraphics[width=0.9\textwidth]{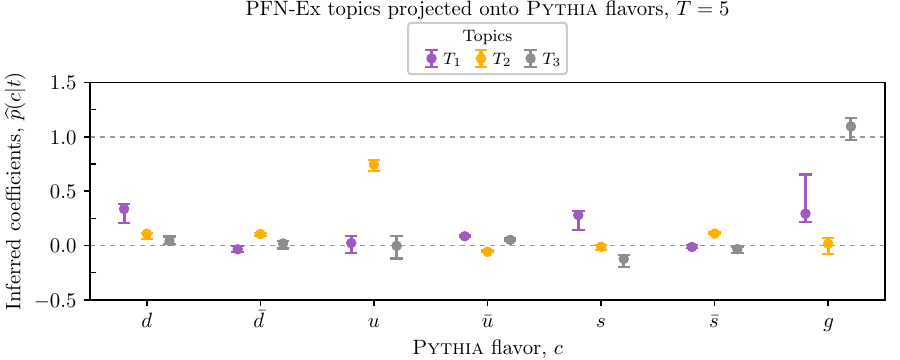}
}
\vspace{1ex}
\subfloat[\label{fig:t6_pct_ex_entries_b}]{%
    \includegraphics[width=0.9\textwidth]{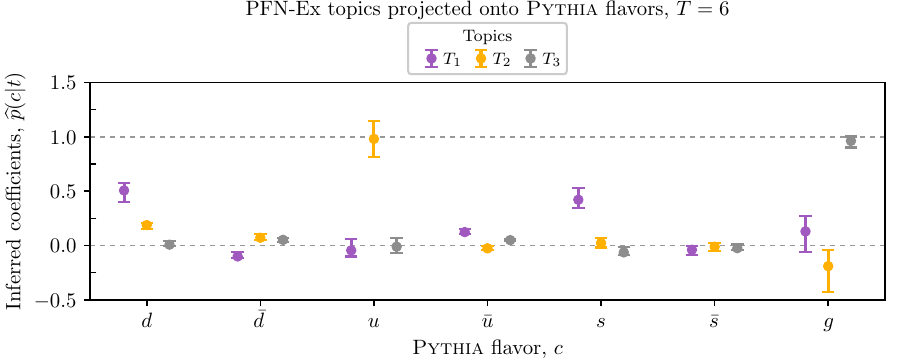}
}

\caption{\label{fig:pct_ex_entries}
Inferred (quasi-)probabilities for the topics inferred from the PFN-Ex study, for (a) $T=5$ and (b) $T=6$.
Both models contain a $u$-dominated vertex, a $g$-dominated vertex, and a vertex comprised of a mixture of $d$ and $s$ jets. 
Without $\pi/K/p$ separation, the classifier cannot distinguish well between the \tsc{Pythia} $d$ and $s$ flavors and therefore places them in the same region of classifier space.}
\end{figure}

\begin{figure}[t]
\centering

\subfloat[\label{fig:substructure_distributions_ex_a}]{%
    \includegraphics[width=0.45\textwidth]{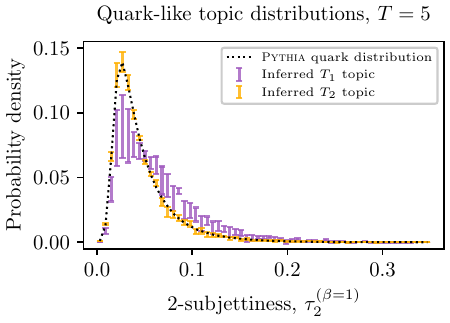}
}
\hfill
\subfloat[\label{fig:substructure_distributions_ex_b}]{%
    \includegraphics[width=0.45\textwidth]{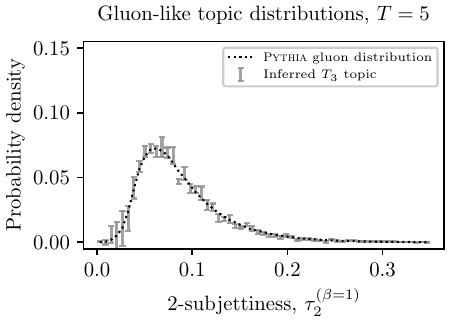}
}

~\\

\subfloat[\label{fig:substructure_distributions_ex_c}]{%
    \includegraphics[width=0.45\textwidth]{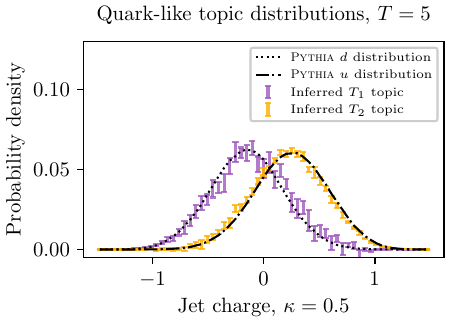}
}
\hfill
\subfloat[\label{fig:substructure_distributions_ex_d}]{%
    \includegraphics[width=0.45\textwidth]{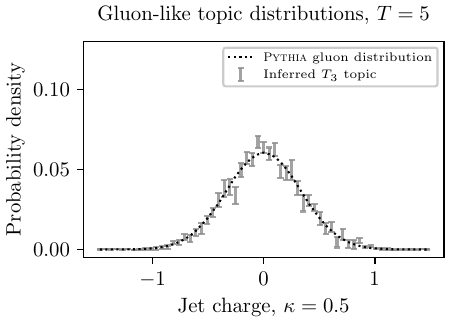}
}

\caption{\label{fig:substructure_distributions_ex5}
Substructure distributions of (top row) 2-subjettiness and (bottom row) jet charge, for the down/strange-, up-, and gluon-like topics from the PFN-Ex study with $T = 5$.
(Left column) The quark-like distributions are compared with the corresponding truth-level \tsc{Pythia} distributions.
(Right column) The analogous comparison is made for the gluon-like distributions.
}
\end{figure}

\begin{figure}[t]
\centering

\subfloat[\label{fig:substructure_distributions_ex_e}]{%
    \includegraphics[width=0.45\textwidth]{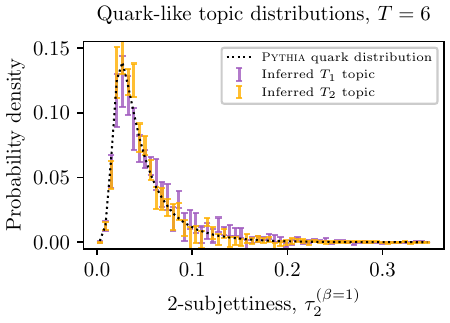}
}
\hfill
\subfloat[\label{fig:substructure_distributions_ex_f}]{%
    \includegraphics[width=0.45\textwidth]{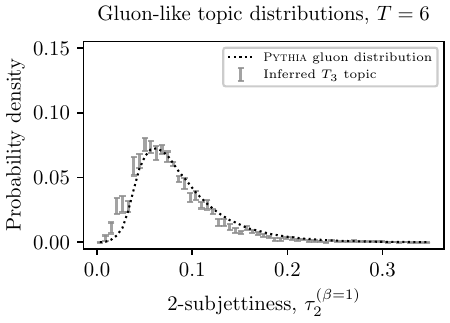}
}

~\\

\subfloat[\label{fig:substructure_distributions_ex_g}]{%
    \includegraphics[width=0.48\textwidth]{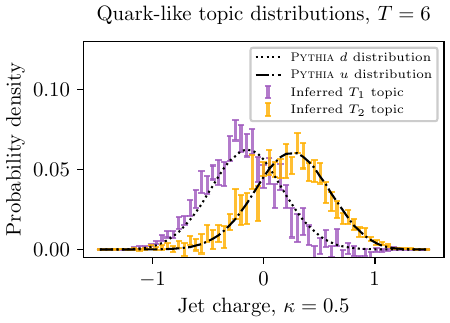}
}
\hfill
\subfloat[\label{fig:substructure_distributions_ex_h}]{%
    \includegraphics[width=0.48\textwidth]{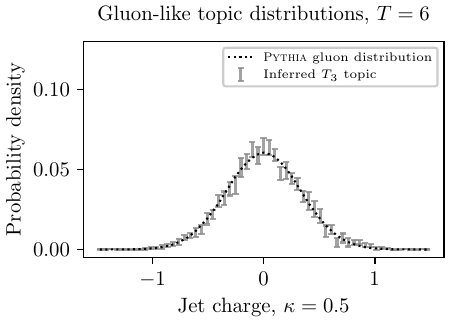}
}

\caption{\label{fig:substructure_distributions_ex6}
Same as \Fig{fig:substructure_distributions_ex5}, but for $T = 6$.
}
\end{figure}

Following \Sec{sec:training_the_fourteen_mixture_demixer}, we tune the demixer's $\alpha$ hyperparameter to the elbow between the flat and rising regimes of the validation cross entropy versus $\alpha$ curve, yielding $\alpha = 6.12 \times 10^{-4}$.
At this $\alpha$, two of the seven vertices remain sparsely populated, yielding mixtures that are smaller by over a factor of $10^3$ than those from their densely populated counterparts. 
The low statistics can degrade the fraction extraction, which approaches the truth only in the asymptotic limit.
To mitigate this, we remove the two sparse vertices by training the prune stage with $\beta = 0.11$.
At the end of training, the base classifier of the demixer has only five active outputs, corresponding to the five surviving vertices.
Then, we use the pruned simplex demixer to implement the tag-and-probe method following \Sec{subsec:tag_and_probe}.
Because the simplex demixer is limited to five outputs, the tag-and-probe procedure gives only 10 jet mixtures.
The truth-level matrix of mixing fractions, $G_{mc}$, is more ill-conditioned than the matrix obtained from the PFN-ID method, with the $u$ and $g$ fractions exhibiting the greatest variance across mixtures.
The resulting training set contains 35.2 million jets in total.

We train a second simplex demixer on the 10 correlated mixtures using the same architecture and procedure as in \Sec{sec:training_the_fourteen_mixture_demixer}.
In the shape stage, we select $\alpha = 1.97 \times 10^{-5}$ using the elbow criterion.
In the prune stage, we scan $\beta$ along a lasso path of 12 log-spaced values from 0.02 to 0.06.
Using $\gamma = 0$, fixed $\alpha$, and a fixed $\beta$ grid, we rerun the procedure on 20 bootstrap resamples to estimate the variance of the $\widehat{p}(t|c)$ and $\widehat{p}(c|t)$ projection matrices.

In this study, we can no longer use domain knowledge to pick $T$, so we instead select $T=5$ and $T=6$ as they are the simplest models on the plateau of the validation cross entropy versus $T$ curve, analogous to \Fig{fig:val_ce_vs_t}.
In \Fig{fig:pct_ex_entries}, we show the $\widehat{p}(c|t)$ entries for three topics obtained from the $T=5$ and $T=6$ models. 
For both models, $T_1$, $T_2$, and $T_3$ denote the down/strange-, up-, and gluon-dominated vertices, respectively.
Notably, without $\pi/K/p$ separation and with nearly degenerate mixtures, the classifier cannot distinguish between the $d$ and $s$ flavors and therefore places them in the same region of phase space. 
In contrast, the classifier still assigns the $u$ and $g$ flavors to distinct regions of phase space.

To confirm the correspondence between topics and \tsc{Pythia} flavors, we calculate the $T_1$, $T_2$, and $T_3$ topic distributions over the 2-subjettiness and jet charge observables following the procedure in \Sec{subsec:substructure_distributions}.
In \Figs{fig:substructure_distributions_ex5}{fig:substructure_distributions_ex6}, we show the resulting binned distributions for $T = 5$ and $T = 6$, respectively. 
In both figures, the left column contains the quark-like distributions ($T_1$ and $T_2$), whereas the right column contains the gluon-like distribution ($T_3$). 
Overall, the features of the topic distributions agree with their down/strange, up, and gluon interpretations. 
The prominent gluon-like distortion in \Fig{fig:substructure_distributions_ex_a} is consistent with the nonzero $p(g|T_1)$ entry in \Fig{fig:pct_ex_entries}.

These results highlight that $\pi/K/p$ separation plays an important role in ensuring that the operational topics align with less well-defined notions of jet flavor. 
In the absence of $\pi/K/p$ separation, charge becomes a more prominent discriminator between quark-like jets.
More generally, simplex demixing extracts the most distinct topics latent in the data, where the meaning of ``most distinct" depends on the features available to the base classifier architecture.

\bibliographystyle{JHEP}
\bibliography{biblio}

@article{JohannRaphael,
    author = "Bonnet-Guerrini, Rapha{\~A}{\guillemotleft}l and Ioannou-Nikolaides, Johann and Petersen, Troels and Piuri, Vincenzo",
    title = "{Multiclass Classification without Labels via Posterior Simplex Geometry}",
    eprint = "2607.24943",
    archivePrefix = "arXiv",
    primaryClass = "cs.LG",
    month = "7",
    year = "2026"
}

@article{Stewart:2022ari,
    author = "Stewart, Iain W. and Yao, Xiaojun",
    title = "{Pure quark and gluon observables in collinear drop}",
    eprint = "2203.14980",
    archivePrefix = "arXiv",
    primaryClass = "hep-ph",
    reportNumber = "MIT-CTP 5365",
    doi = "10.1007/JHEP09(2022)120",
    journal = "JHEP",
    volume = "09",
    pages = "120",
    year = "2022"
}

@article{Srivastava:2014kpo,
    author = "Srivastava, Nitish and Hinton, Geoffrey and Krizhevsky, Alex and Sutskever, Ilya and Salakhutdinov, Ruslan",
    title = "{Dropout: A Simple Way to Prevent Neural Networks from Overfitting}",
    journal = "J. Machine Learning Res.",
    volume = "15",
    pages = "1929--1958",
    year = "2014"
}

@article{Ioffe:2015ovl,
    author = "Ioffe, Sergey and Szegedy, Christian",
    title = "{Batch Normalization: Accelerating Deep Network Training by Reducing  Internal Covariate Shift}",
    eprint = "1502.03167",
    archivePrefix = "arXiv",
    primaryClass = "cs.LG",
    month = "2",
    year = "2015"
}

@article{Nilles:1980ys,
    author = "Nilles, Hans Peter and Streng, K. H.",
    title = "{Quark - Gluon Separation in Three Jet Events}",
    reportNumber = "SLAC-PUB-2586",
    doi = "10.1103/PhysRevD.23.1944",
    journal = "Phys. Rev. D",
    volume = "23",
    pages = "1944",
    year = "1981"
}

@ARTICLE{6200362,
  author={Bioucas-Dias, José M. and Plaza, Antonio and Dobigeon, Nicolas and Parente, Mario and Du, Qian and Gader, Paul and Chanussot, Jocelyn},
  journal={IEEE Journal of Selected Topics in Applied Earth Observations and Remote Sensing}, 
  title={Hyperspectral Unmixing Overview: Geometrical, Statistical, and Sparse Regression-Based Approaches}, 
  year={2012},
  volume={5},
  number={2},
  pages={354-379},
  doi={10.1109/JSTARS.2012.2194696}}

@article{Neyman:1933wgr,
    author = "Neyman, Jerzy and Pearson, Egon Sharpe",
    title = "{On the Problem of the Most Efficient Tests of Statistical Hypotheses}",
    doi = "10.1098/rsta.1933.0009",
    journal = "Phil. Trans. Roy. Soc. Lond. A",
    volume = "231",
    number = "694-706",
    pages = "289--337",
    year = "1933"
}

@article{Gras:2017jty,
    author = {Gras, Philippe and H{\"o}che, Stefan and Kar, Deepak and Larkoski, Andrew and L{\"o}nnblad, Leif and Pl{\"a}tzer, Simon and Si{\'o}dmok, Andrzej and Skands, Peter and Soyez, Gregory and Thaler, Jesse},
    title = "{Systematics of quark/gluon tagging}",
    eprint = "1704.03878",
    archivePrefix = "arXiv",
    primaryClass = "hep-ph",
    reportNumber = "MIT-CTP-4885, COEPP-MN-17-2, MCNET-17-04",
    doi = "10.1007/JHEP07(2017)091",
    journal = "JHEP",
    volume = "07",
    pages = "091",
    year = "2017"
}

@article{Komiske:2018vkc,
    author = "Komiske, Patrick T. and Metodiev, Eric M. and Thaler, Jesse",
    title = "{An operational definition of quark and gluon jets}",
    eprint = "1809.01140",
    archivePrefix = "arXiv",
    primaryClass = "hep-ph",
    reportNumber = "MIT-CTP 5042",
    doi = "10.1007/JHEP11(2018)059",
    journal = "JHEP",
    volume = "11",
    pages = "059",
    year = "2018"
}

@article{Metodiev:2017vrx,
    author = "Metodiev, Eric M. and Nachman, Benjamin and Thaler, Jesse",
    title = "{Classification without labels: Learning from mixed samples in high energy physics}",
    eprint = "1708.02949",
    archivePrefix = "arXiv",
    primaryClass = "hep-ph",
    reportNumber = "MIT--CTP-4922",
    doi = "10.1007/JHEP10(2017)174",
    journal = "JHEP",
    volume = "10",
    pages = "174",
    year = "2017"
}

@article{Metodiev:2018ftz,
    author = "Metodiev, Eric M. and Thaler, Jesse",
    title = "{Jet Topics: Disentangling Quarks and Gluons at Colliders}",
    eprint = "1802.00008",
    archivePrefix = "arXiv",
    primaryClass = "hep-ph",
    reportNumber = "MIT-CTP-4979",
    doi = "10.1103/PhysRevLett.120.241602",
    journal = "Phys. Rev. Lett.",
    volume = "120",
    number = "24",
    pages = "241602",
    year = "2018"
}

@article{Komiske:2022vxg,
    author = "Komiske, Patrick T. and Kryhin, Serhii and Thaler, Jesse",
    title = "{Disentangling quarks and gluons in CMS open data}",
    eprint = "2205.04459",
    archivePrefix = "arXiv",
    primaryClass = "hep-ph",
    reportNumber = "MIT-CTP 5422",
    doi = "10.1103/PhysRevD.106.094021",
    journal = "Phys. Rev. D",
    volume = "106",
    number = "9",
    pages = "094021",
    year = "2022"
}

@article{Frye:2016okc,
    author = "Frye, Christopher and Larkoski, Andrew J. and Schwartz, Matthew D. and Yan, Kai",
    title = "{Precision physics with pile-up insensitive observables}",
    eprint = "1603.06375",
    archivePrefix = "arXiv",
    primaryClass = "hep-ph",
    month = "3",
    year = "2016"
}

@article{Frye:2016aiz,
    author = "Frye, Christopher and Larkoski, Andrew J. and Schwartz, Matthew D. and Yan, Kai",
    title = "{Factorization for groomed jet substructure beyond the next-to-leading logarithm}",
    eprint = "1603.09338",
    archivePrefix = "arXiv",
    primaryClass = "hep-ph",
    doi = "10.1007/JHEP07(2016)064",
    journal = "JHEP",
    volume = "07",
    pages = "064",
    year = "2016"
}

@article{Thaler:2010tr,
    author = "Thaler, Jesse and Van Tilburg, Ken",
    title = "{Identifying Boosted Objects with N-subjettiness}",
    eprint = "1011.2268",
    archivePrefix = "arXiv",
    primaryClass = "hep-ph",
    reportNumber = "MIT-CTP-4191",
    doi = "10.1007/JHEP03(2011)015",
    journal = "JHEP",
    volume = "03",
    pages = "015",
    year = "2011"
}

@article{Komiske:2018cqr,
    author = "Komiske, Patrick T. and Metodiev, Eric M. and Thaler, Jesse",
    title = "{Energy Flow Networks: Deep Sets for Particle Jets}",
    eprint = "1810.05165",
    archivePrefix = "arXiv",
    primaryClass = "hep-ph",
    reportNumber = "MIT-CTP 5064",
    doi = "10.1007/JHEP01(2019)121",
    journal = "JHEP",
    volume = "01",
    pages = "121",
    year = "2019"
}

@article{blei2003latent,
author = {Blei, David M. and Ng, Andrew Y. and Jordan, Michael I.},
title = {Latent dirichlet allocation},
year = {2003},
issue_date = {3/1/2003},
publisher = {JMLR.org},
volume = {3},
number = {null},
issn = {1532-4435},
journal = {J. Mach. Learn. Res.},
month = mar,
pages = {993–1022},
numpages = {30}
}

@inproceedings{arora2012learning,
author = {Arora, Sanjeev and Ge, Rong and Moitra, Ankur},
title = {Learning Topic Models -- Going beyond SVD},
year = {2012},
isbn = {9780769548746},
publisher = {IEEE Computer Society},
address = {USA},
url = {https://doi.org/10.1109/FOCS.2012.49},
doi = {10.1109/FOCS.2012.49},
booktitle = {Proceedings of the 2012 IEEE 53rd Annual Symposium on Foundations of Computer Science},
pages = {1–10},
numpages = {10},
series = {FOCS '12}
}

@InProceedings{pmlr-v28-arora13,
  title = 	 {A Practical Algorithm for Topic Modeling with Provable Guarantees},
  author = 	 {Arora, Sanjeev and Ge, Rong and Halpern, Yonatan and Mimno, David and Moitra, Ankur and Sontag, David and Wu, Yichen and Zhu, Michael},
  booktitle = 	 {Proceedings of the 30th International Conference on Machine Learning},
  pages = 	 {280--288},
  year = 	 {2013},
  editor = 	 {Dasgupta, Sanjoy and McAllester, David},
  volume = 	 {28},
  series = 	 {Proceedings of Machine Learning Research},
  address = 	 {Atlanta, Georgia, USA},
  month = 	 {17--19 Jun},
  publisher =    {PMLR},
  url = 	 {https://proceedings.mlr.press/v28/arora13.html}
}

@misc{EnergyFlow,
    author  = {Komiske, Patrick T. and Metodiev, Eric M. and Thaler, Jesse},
    title   = {{EnergyFlow}},
    howpublished = {\url{https://energyflow.network/}},
    url     = {https://github.com/thaler-lab/EnergyFlow},
    note    = {Python package for particle physics tools},
}

@article{3322706.3361982,
author = {Katz-Samuels, Julian and Blanchard, Gilles and Scott, Clayton},
title = {Decontamination of mutual contamination models},
year = {2019},
issue_date = {January 2019},
publisher = {JMLR.org},
volume = {20},
number = {1},
issn = {1532-4435},
journal = {J. Mach. Learn. Res.},
month = jan,
pages = {1521–1577},
numpages = {57}
}

@InProceedings{pmlr-v30-Scott13,
  title = 	 {Classification with Asymmetric Label Noise: Consistency and Maximal Denoising},
  author = 	 {Scott, Clayton and Blanchard, Gilles and Handy, Gregory},
  booktitle = 	 {Proceedings of the 26th Annual Conference on Learning Theory},
  pages = 	 {489--511},
  year = 	 {2013},
  editor = 	 {Shalev-Shwartz, Shai and Steinwart, Ingo},
  volume = 	 {30},
  series = 	 {Proceedings of Machine Learning Research},
  address = 	 {Princeton, NJ, USA},
  month = 	 {12--14 Jun},
  publisher =    {PMLR},
  url = 	 {https://proceedings.mlr.press/v30/Scott13.html}
}

@article{schwarz1978estimating,
 ISSN = {00905364, 21688966},
 URL = {http://www.jstor.org/stable/2958889},
 author = {Gideon Schwarz},
 journal = {The Annals of Statistics},
 number = {2},
 pages = {461--464},
 publisher = {Institute of Mathematical Statistics},
 title = {Estimating the Dimension of a Model},
 urldate = {2026-06-16},
 volume = {6},
 year = {1978}
}

@ARTICLE{1411995,
  author={Nascimento, J.M.P. and Dias, J.M.B.},
  journal={IEEE Transactions on Geoscience and Remote Sensing}, 
  title={Vertex component analysis: a fast algorithm to unmix hyperspectral data}, 
  year={2005},
  volume={43},
  number={4},
  pages={898-910},
  doi={10.1109/TGRS.2005.844293}}

@article{eugster2009spider,
  title={From spider-man to hero—archetypal analysis in R},
  author={Eugster, Manuel JA and Leisch, Friedrich},
  journal={Journal of Statistical Software},
  volume={30},
  pages={1--23},
  year={2009}
}

@article{cutler1994archetypal,
 ISSN = {00401706},
 URL = {http://www.jstor.org/stable/1269949},
 author = {Adele Cutler and Leo Breiman},
 journal = {Technometrics},
 number = {4},
 pages = {338--347},
 publisher = {[Taylor & Francis, Ltd., American Statistical Association, American Society for Quality]},
 title = {Archetypal Analysis},
 urldate = {2026-06-16},
 volume = {36},
 year = {1994}
}

@inproceedings{suleman2017validation,
  title={Validation of archetypal analysis},
  author={Suleman, Abdul},
  booktitle={2017 IEEE International Conference on Fuzzy Systems (FUZZ-IEEE)},
  pages={1--6},
  year={2017},
  organization={IEEE}
}

@article{
mair2024archetypal,
title={Archetypal Analysis++: Rethinking the Initialization Strategy},
author={Sebastian Mair and Jens Sj{\"o}lund},
journal={Transactions on Machine Learning Research},
issn={2835-8856},
year={2024},
url={https://openreview.net/forum?id=KVUtlM60HM},
note={}
}

@inproceedings{10.1117/12.366289,
author = {Michael E. Winter},
title = {{N-FINDR: an algorithm for fast autonomous spectral end-member determination in hyperspectral data}},
volume = {3753},
booktitle = {Imaging Spectrometry V},
editor = {Michael R. Descour and Sylvia S. Shen},
organization = {International Society for Optics and Photonics},
publisher = {SPIE},
pages = {266 -- 275},
year = {1999},
doi = {10.1117/12.366289},
URL = {https://doi.org/10.1117/12.366289}
}

@INPROCEEDINGS{4779330,
  author={Li, Jun and Bioucas-Dias, Jose M.},
  booktitle={IGARSS 2008 - 2008 IEEE International Geoscience and Remote Sensing Symposium}, 
  title={Minimum Volume Simplex Analysis: A Fast Algorithm to Unmix Hyperspectral Data}, 
  year={2008},
  volume={3},
  number={},
  pages={III - 250-III - 253},
  doi={10.1109/IGARSS.2008.4779330}}

@article{10.1214/08-AOAS227,
author = {Art B. Owen and Patrick O. Perry},
title = {{Bi-cross-validation of the SVD and the nonnegative matrix factorization}},
volume = {3},
journal = {The Annals of Applied Statistics},
number = {2},
publisher = {Institute of Mathematical Statistics},
pages = {564 -- 594},
year = {2009},
doi = {10.1214/08-AOAS227},
URL = {https://doi.org/10.1214/08-AOAS227}
}

@book{2834535,
author = {Hastie, Trevor and Tibshirani, Robert and Wainwright, Martin},
title = {Statistical Learning with Sparsity: The Lasso and Generalizations},
year = {2015},
isbn = {1498712169},
publisher = {Chapman \& Hall/CRC}
}

@article{Hoerl01021970,
author = {Arthur E. Hoerl and Robert W. Kennard},
title = {Ridge Regression: Biased Estimation for Nonorthogonal Problems},
journal = {Technometrics},
volume = {12},
number = {1},
pages = {55--67},
year = {1970},
publisher = {Taylor \& Francis},
doi = {10.1080/00401706.1970.10488634},
URL = { 
        https://doi.org/10.1080/00401706.1970.10488634
},
eprint = {  
        https://doi.org/10.1080/00401706.1970.10488634
}
}

@article{51791361-8fe2-38d5-959f-ae8d048b490d,
 ISSN = {00359246},
 URL = {http://www.jstor.org/stable/2346178},
 author = {Robert Tibshirani},
 journal = {Journal of the Royal Statistical Society. Series B (Methodological)},
 number = {1},
 pages = {267--288},
 publisher = {[Royal Statistical Society, Oxford University Press]},
 title = {Regression Shrinkage and Selection via the Lasso},
 urldate = {2026-07-04},
 volume = {58},
 year = {1996}
}

@misc{chollet2015keras,
  author = {Chollet, Fran\c{c}ois and others},
  title = {Keras},
  year = {2015},
  howpublished = {\url{https://github.com/keras-team/keras}},
}

@misc{tensorflow2015-whitepaper,
title={ {TensorFlow}: Large-Scale Machine Learning on Heterogeneous Systems},
url={https://www.tensorflow.org/},
note={Software available from tensorflow.org},
author={
    Mart\'{i}n~Abadi and
    Ashish~Agarwal and
    Paul~Barham and
    Eugene~Brevdo and
    Zhifeng~Chen and
    Craig~Citro and
    Greg~S.~Corrado and
    Andy~Davis and
    Jeffrey~Dean and
    Matthieu~Devin and
    Sanjay~Ghemawat and
    Ian~Goodfellow and
    Andrew~Harp and
    Geoffrey~Irving and
    Michael~Isard and
    Yangqing Jia and
    Rafal~Jozefowicz and
    Lukasz~Kaiser and
    Manjunath~Kudlur and
    Josh~Levenberg and
    Dandelion~Man\'{e} and
    Rajat~Monga and
    Sherry~Moore and
    Derek~Murray and
    Chris~Olah and
    Mike~Schuster and
    Jonathon~Shlens and
    Benoit~Steiner and
    Ilya~Sutskever and
    Kunal~Talwar and
    Paul~Tucker and
    Vincent~Vanhoucke and
    Vijay~Vasudevan and
    Fernanda~Vi\'{e}gas and
    Oriol~Vinyals and
    Pete~Warden and
    Martin~Wattenberg and
    Martin~Wicke and
    Yuan~Yu and
    Xiaoqiang~Zheng},
  year={2015},
}

@inproceedings{He:2015dtg,
    author = "He, Kaiming and Zhang, Xiangyu and Ren, Shaoqing and Sun, Jian",
    title = "{Delving Deep into Rectifiers: Surpassing Human-Level Performance on ImageNet Classification}",
    eprint = "1502.01852",
    archivePrefix = "arXiv",
    primaryClass = "cs.CV",
    doi = "10.1109/ICCV.2015.123",
    month = "2",
    year = "2015"
}

@inproceedings{Kingma:2014vow,
    author = "Kingma, Diederik P. and Ba, Jimmy",
    title = "{Adam: A Method for Stochastic Optimization}",
    booktitle = "{International Conference on Learning Representations}",
    eprint = "1412.6980",
    archivePrefix = "arXiv",
    primaryClass = "cs.LG",
    month = "12",
    year = "2014"
}

@article{10.1214/aos/1176344552,
author = {B. Efron},
title = {{Bootstrap Methods: Another Look at the Jackknife}},
volume = {7},
journal = {The Annals of Statistics},
number = {1},
publisher = {Institute of Mathematical Statistics},
pages = {1 -- 26},
year = {1979},
doi = {10.1214/aos/1176344552},
URL = {https://doi.org/10.1214/aos/1176344552}
}

@book{Efron_Hastie_2016, place={Cambridge}, series={Institute of Mathematical Statistics Monographs}, title={Computer Age Statistical Inference: Algorithms, Evidence, and Data Science}, publisher={Cambridge University Press}, author={Efron, Bradley and Hastie, Trevor}, year={2016}, collection={Institute of Mathematical Statistics Monographs}}

@article{Amram:2020ykb,
    author = "Amram, Oz and Suarez, Cristina Mantilla",
    title = "{Tag N{\textquoteright} Train: a technique to train improved classifiers on unlabeled data}",
    eprint = "2002.12376",
    archivePrefix = "arXiv",
    primaryClass = "hep-ph",
    doi = "10.1007/JHEP01(2021)153",
    journal = "JHEP",
    volume = "01",
    pages = "153",
    year = "2021"
}

@article{EuropeanMuon:1984xji,
    author = "Albanese, J. P. and others",
    collaboration = "European Muon",
    title = "{Quark Charge Retention in Final State Hadrons From Deep Inelastic Muon Scattering}",
    reportNumber = "CERN-EP-84-63",
    doi = "10.1016/0370-2693(84)91825-2",
    journal = "Phys. Lett. B",
    volume = "144",
    pages = "302--308",
    year = "1984"
}

@article{Krohn:2012fg,
    author = "Krohn, David and Schwartz, Matthew D. and Lin, Tongyan and Waalewijn, Wouter J.",
    title = "{Jet Charge at the LHC}",
    eprint = "1209.2421",
    archivePrefix = "arXiv",
    primaryClass = "hep-ph",
    doi = "10.1103/PhysRevLett.110.212001",
    journal = "Phys. Rev. Lett.",
    volume = "110",
    number = "21",
    pages = "212001",
    year = "2013"
}

@misc{fastjet_contrib,
  author       = {{FastJet Collaboration}},
  title        = {{FastJet contrib}},
  howpublished = {\url{https://fastjet.hepforge.org/contrib/}},
  note         = {Accessed: 2026-07-05}
}

@article{CMS:2023fix,
    author = "Hayrapetyan, Aram and others",
    collaboration = "CMS",
    title = "{Measurement of multidifferential cross sections for dijet production in proton{\textendash}proton collisions at $\sqrt{s} = 13\,\text {Te}\hspace{-.08em}\text {V} $}",
    eprint = "2312.16669",
    archivePrefix = "arXiv",
    primaryClass = "hep-ex",
    reportNumber = "CMS-SMP-21-008, CERN-EP-2023-257",
    doi = "10.1140/epjc/s10052-024-13606-8",
    journal = "Eur. Phys. J. C",
    volume = "85",
    number = "1",
    pages = "72",
    year = "2025"
}

@article{ATLAS:2017ble,
    author = "Aaboud, M. and others",
    collaboration = "ATLAS",
    title = "{Measurement of inclusive jet and dijet cross-sections in proton-proton collisions at $\sqrt{s}=13$ TeV with the ATLAS detector}",
    eprint = "1711.02692",
    archivePrefix = "arXiv",
    primaryClass = "hep-ex",
    reportNumber = "CERN-EP-2017-157",
    doi = "10.1007/JHEP05(2018)195",
    journal = "JHEP",
    volume = "05",
    pages = "195",
    year = "2018"
}

@article{ALICE:2020nkc,
    author = "Acharya, Shreyasi and others",
    collaboration = "ALICE",
    title = "{Multiplicity dependence of $\pi $, K, and p production in pp collisions at $\sqrt{s} = 13$ TeV}",
    eprint = "2003.02394",
    archivePrefix = "arXiv",
    primaryClass = "nucl-ex",
    reportNumber = "CERN-EP-2020-024",
    doi = "10.1140/epjc/s10052-020-8125-1",
    journal = "Eur. Phys. J. C",
    volume = "80",
    number = "8",
    pages = "693",
    year = "2020"
}

@article{Calabrese:2022eju,
    author = "Calabrese, R. and others",
    title = "{Performance of the LHCb RICH detectors during LHC~Run~2}",
    eprint = "2205.13400",
    archivePrefix = "arXiv",
    primaryClass = "physics.ins-det",
    reportNumber = "LHCb-DP-2021-004",
    doi = "10.1088/1748-0221/17/07/P07013",
    journal = "JINST",
    volume = "17",
    number = "07",
    pages = "P07013",
    year = "2022"
}

@article{CMS:2017eoq,
    author = "Sirunyan, Albert M and others",
    collaboration = "CMS",
    title = "{Measurement of charged pion, kaon, and proton production in proton-proton collisions at $\sqrt{s}=13$ TeV}",
    eprint = "1706.10194",
    archivePrefix = "arXiv",
    primaryClass = "hep-ex",
    reportNumber = "CMS-FSQ-16-004, CERN-EP-2017-091",
    doi = "10.1103/PhysRevD.96.112003",
    journal = "Phys. Rev. D",
    volume = "96",
    number = "11",
    pages = "112003",
    year = "2017"
}

@article{Larkoski:2019nwj,
    author = "Larkoski, Andrew J. and Metodiev, Eric M.",
    title = "{A Theory of Quark vs. Gluon Discrimination}",
    eprint = "1906.01639",
    archivePrefix = "arXiv",
    primaryClass = "hep-ph",
    reportNumber = "MIT-CTP 5049",
    doi = "10.1007/JHEP10(2019)014",
    journal = "JHEP",
    volume = "10",
    pages = "014",
    year = "2019"
}

@article{yu2019bootstrapping,
author = {Yu, Han and Chapman, Brian and Di Florio, Arianna and Eischen, Ellen and Gotz, David and Jacob, Mathews and Blair, Rachael Hageman},
title = {Bootstrapping estimates of stability for clusters, observations and model selection},
year = {2019},
issue_date = {March     2019},
publisher = {Kluwer Academic Publishers},
address = {USA},
volume = {34},
number = {1},
issn = {0943-4062},
url = {https://doi.org/10.1007/s00180-018-0830-y},
doi = {10.1007/s00180-018-0830-y},
journal = {Comput. Stat.},
month = mar,
pages = {349–372},
numpages = {24}
}

@article{Nakai:2020kuu,
    author = "Nakai, Yuichiro and Shih, David and Thomas, Scott",
    title = "{Strange Jet Tagging}",
    eprint = "2003.09517",
    archivePrefix = "arXiv",
    primaryClass = "hep-ph",
    month = "3",
    year = "2020"
}

@article{meinshausen2010stability,
author = {Meinshausen, Nicolai and Bühlmann, Peter},
title = {Stability selection},
journal = {Journal of the Royal Statistical Society: Series B (Statistical Methodology)},
volume = {72},
number = {4},
pages = {417-473},
doi = {https://doi.org/10.1111/j.1467-9868.2010.00740.x},
url = {https://rss.onlinelibrary.wiley.com/doi/abs/10.1111/j.1467-9868.2010.00740.x},
eprint = {https://rss.onlinelibrary.wiley.com/doi/pdf/10.1111/j.1467-9868.2010.00740.x},
year = {2010}
}

@article{ROUSSEEUW198753,
title = {Silhouettes: A graphical aid to the interpretation and validation of cluster analysis},
journal = {Journal of Computational and Applied Mathematics},
volume = {20},
pages = {53-65},
year = {1987},
issn = {0377-0427},
doi = {https://doi.org/10.1016/0377-0427(87)90125-7},
url = {https://www.sciencedirect.com/science/article/pii/0377042787901257},
author = {Peter J. Rousseeuw}
}

@article{tibshirani2001gap,
  author  = {Tibshirani, Robert and Walther, Guenther and Hastie, Trevor},
  title   = {Estimating the Number of Clusters in a Data Set via the Gap Statistic},
  journal = {Journal of the Royal Statistical Society: Series B (Statistical Methodology)},
  volume  = {63},
  number  = {2},
  pages   = {411--423},
  year    = {2001},
  doi     = {10.1111/1467-9868.00293}
}

@article{Brewer:2020och,
    author = "Brewer, Jasmine and Thaler, Jesse and Turner, Andrew P.",
    title = "{Data-driven quark and gluon jet modification in heavy-ion collisions}",
    eprint = "2008.08596",
    archivePrefix = "arXiv",
    primaryClass = "hep-ph",
    reportNumber = "MIT-CTP 5219",
    doi = "10.1103/PhysRevC.103.L021901",
    journal = "Phys. Rev. C",
    volume = "103",
    number = "2",
    pages = "L021901",
    year = "2021"
}

@article{Buckley:2015gua,
    author = "Buckley, Andy and Pollard, Chris",
    title = "{QCD-aware partonic jet clustering for truth-jet flavour labelling}",
    eprint = "1507.00508",
    archivePrefix = "arXiv",
    primaryClass = "hep-ph",
    reportNumber = "MCNET-15-12, GLA-PPE-2015-03",
    doi = "10.1140/epjc/s10052-016-3925-z",
    journal = "Eur. Phys. J. C",
    volume = "76",
    number = "2",
    pages = "71",
    year = "2016"
}

@article{Caletti:2022glq,
    author = "Caletti, Simone and Larkoski, Andrew J. and Marzani, Simone and Reichelt, Daniel",
    title = "{A fragmentation approach to jet flavor}",
    eprint = "2205.01117",
    archivePrefix = "arXiv",
    primaryClass = "hep-ph",
    reportNumber = "SLAC-PUB-17663, IPPP/22/28",
    doi = "10.1007/JHEP10(2022)158",
    journal = "JHEP",
    volume = "10",
    pages = "158",
    year = "2022"
}

@article{Banfi:2006hf,
    author = "Banfi, Andrea and Salam, Gavin P. and Zanderighi, Giulia",
    title = "{Infrared safe definition of jet flavor}",
    eprint = "hep-ph/0601139",
    archivePrefix = "arXiv",
    reportNumber = "FERMILAB-PUB-06-003-T, BICOCCA-FT-05-28, CAVENDISH-HEP-05-25, CERN-PH-TH-2006-002, DAMTP-2005-134, LPTHE-05-34",
    doi = "10.1140/epjc/s2006-02552-4",
    journal = "Eur. Phys. J. C",
    volume = "47",
    pages = "113--124",
    year = "2006"
}

@article{Gallicchio:2011xc,
    author = "Gallicchio, Jason and Schwartz, Matthew D.",
    title = "{Pure Samples of Quark and Gluon Jets at the LHC}",
    eprint = "1104.1175",
    archivePrefix = "arXiv",
    primaryClass = "hep-ph",
    doi = "10.1007/JHEP10(2011)103",
    journal = "JHEP",
    volume = "10",
    pages = "103",
    year = "2011"
}

@article{Czakon:2022wam,
    author = "Czakon, Michal and Mitov, Alexander and Poncelet, Rene",
    title = "{Infrared-safe flavoured anti-k$_{T}$ jets}",
    eprint = "2205.11879",
    archivePrefix = "arXiv",
    primaryClass = "hep-ph",
    reportNumber = "Cavendish-HEP-22/06, P3H-22-056, TTK-22-17",
    doi = "10.1007/JHEP04(2023)138",
    journal = "JHEP",
    volume = "04",
    pages = "138",
    year = "2023"
}

@article{Gauld:2022lem,
    author = "Gauld, Rhorry and Huss, Alexander and Stagnitto, Giovanni",
    title = "{Flavor Identification of Reconstructed Hadronic Jets}",
    eprint = "2208.11138",
    archivePrefix = "arXiv",
    primaryClass = "hep-ph",
    reportNumber = "BONN-TH-2022-17, CERN-TH-2022-121, ZU-TH 31/22",
    doi = "10.1103/PhysRevLett.130.161901",
    journal = "Phys. Rev. Lett.",
    volume = "130",
    number = "16",
    pages = "161901",
    year = "2023",
    note = "[Erratum: Phys.Rev.Lett. 132, 159901 (2024)]"
}

@article{Caola:2023wpj,
    author = "Caola, Fabrizio and Grabarczyk, Radoslaw and Hutt, Maxwell L. and Salam, Gavin P. and Scyboz, Ludovic and Thaler, Jesse",
    title = "{Flavored jets with exact anti-kt kinematics and tests of infrared and collinear safety}",
    eprint = "2306.07314",
    archivePrefix = "arXiv",
    primaryClass = "hep-ph",
    reportNumber = "MIT-CTP 5564, OUTP-23-06P",
    doi = "10.1103/PhysRevD.108.094010",
    journal = "Phys. Rev. D",
    volume = "108",
    number = "9",
    pages = "094010",
    year = "2023"
}

@article{Caletti:2022hnc,
    author = "Caletti, Simone and Larkoski, Andrew J. and Marzani, Simone and Reichelt, Daniel",
    title = "{Practical jet flavour through NNLO}",
    eprint = "2205.01109",
    archivePrefix = "arXiv",
    primaryClass = "hep-ph",
    reportNumber = "SLAC-PUB-17674, IPPP/22/27",
    doi = "10.1140/epjc/s10052-022-10568-7",
    journal = "Eur. Phys. J. C",
    volume = "82",
    number = "7",
    pages = "632",
    year = "2022"
}

@article{Bierlich:2022pfr,
    author = "Bierlich, Christian and others",
    title = "{A comprehensive guide to the physics and usage of PYTHIA 8.3}",
    eprint = "2203.11601",
    archivePrefix = "arXiv",
    primaryClass = "hep-ph",
    reportNumber = "LU-TP 22-16, MCNET-22-04, FERMILAB-PUB-22-227-SCD",
    doi = "10.21468/SciPostPhysCodeb.8",
    journal = "SciPost Phys. Codeb.",
    volume = "2022",
    pages = "8",
    year = "2022"
}

@article{Cacciari:2011ma,
    author = "Cacciari, Matteo and Salam, Gavin P. and Soyez, Gregory",
    title = "{FastJet User Manual}",
    eprint = "1111.6097",
    archivePrefix = "arXiv",
    primaryClass = "hep-ph",
    reportNumber = "CERN-PH-TH-2011-297",
    doi = "10.1140/epjc/s10052-012-1896-2",
    journal = "Eur. Phys. J. C",
    volume = "72",
    pages = "1896",
    year = "2012"
}

@article{Cacciari:2005hq,
    author = "Cacciari, Matteo and Salam, Gavin P.",
    title = "{Dispelling the $N^{3}$ myth for the $k_t$ jet-finder}",
    eprint = "hep-ph/0512210",
    archivePrefix = "arXiv",
    reportNumber = "LPTHE-05-32",
    doi = "10.1016/j.physletb.2006.08.037",
    journal = "Phys. Lett. B",
    volume = "641",
    pages = "57--61",
    year = "2006"
}

@article{Cacciari:2008gp,
    author = "Cacciari, Matteo and Salam, Gavin P. and Soyez, Gregory",
    title = "{The anti-$k_t$ jet clustering algorithm}",
    eprint = "0802.1189",
    archivePrefix = "arXiv",
    primaryClass = "hep-ph",
    reportNumber = "LPTHE-07-03",
    doi = "10.1088/1126-6708/2008/04/063",
    journal = "JHEP",
    volume = "04",
    pages = "063",
    year = "2008"
}

@article{ATLAS:2019rqw,
    author = "Aad, Georges and others",
    collaboration = "ATLAS",
    title = "{Properties of jet fragmentation using charged particles measured with the ATLAS detector in $pp$ collisions at $\sqrt{s}=13$ TeV}",
    eprint = "1906.09254",
    archivePrefix = "arXiv",
    primaryClass = "hep-ex",
    reportNumber = "CERN-EP-2019-090",
    doi = "10.1103/PhysRevD.100.052011",
    journal = "Phys. Rev. D",
    volume = "100",
    number = "5",
    pages = "052011",
    year = "2019"
}

@article{LlorenteMerino:2019zov,
    author = "Llorente Merino, Javier",
    editor = "Roig Garc{\'e}s, Pablo and Bautista Guzman, Irais and Fern{\'a}ndez T{\'e}llez, Arturo and Mart{\'\i}nez Hern{\'a}ndez, Mario Iv{\'a}n",
    collaboration = "ATLAS, CMS",
    title = "{Production of top quarks, jets and photons}",
    reportNumber = "ATL-PHYS-PROC-2019-131",
    doi = "10.22323/1.350.0087",
    journal = "PoS",
    volume = "LHCP2019",
    pages = "087",
    year = "2019"
}

@article{VillaplanaPerez:2019yxt,
    author = "Villaplana Perez, Miguel",
    collaboration = "ATLAS",
    title = "{Measurement of jet fragmentation using the ATLAS detector}",
    reportNumber = "ATL-PHYS-PROC-2019-148",
    doi = "10.22323/1.364.0498",
    journal = "PoS",
    volume = "EPS-HEP2019",
    pages = "498",
    year = "2020"
}

@article{ATLAS:2025ycg,
    author = "Aad, Georges and others",
    collaboration = "ATLAS",
    title = "{Performance and efficiency of a transformer-based quark/gluon jet tagger in the ATLAS experiment}",
    eprint = "2512.03949",
    archivePrefix = "arXiv",
    primaryClass = "hep-ex",
    reportNumber = "CERN-EP-2025-252",
    month = "12",
    year = "2025"
}

@article{Lenz:2020eam,
    author = "Lenz, T.",
    collaboration = "ATLAS",
    title = "{Measurement of jet substructure observables and jet fragmentation properties using the ATLAS detector}",
    reportNumber = "ATL-PHYS-PROC-2019-106",
    doi = "10.1016/j.nuclphysbps.2019.11.002",
    journal = "Nucl. Part. Phys. Proc.",
    volume = "309-311",
    pages = "1--8",
    year = "2020"
}

@article{Dillon:2019cqt,
    author = "Dillon, Barry M. and Faroughy, Darius A. and Kamenik, Jernej F.",
    title = "{Uncovering latent jet substructure}",
    eprint = "1904.04200",
    archivePrefix = "arXiv",
    primaryClass = "hep-ph",
    doi = "10.1103/PhysRevD.100.056002",
    journal = "Phys. Rev. D",
    volume = "100",
    number = "5",
    pages = "056002",
    year = "2019"
}

@article{Dillon:2020quc,
    author = "Dillon, B. M. and Faroughy, D. A. and Kamenik, J. F. and Szewc, M.",
    title = "{Learning the latent structure of collider events}",
    eprint = "2005.12319",
    archivePrefix = "arXiv",
    primaryClass = "hep-ph",
    doi = "10.1007/JHEP10(2020)206",
    journal = "JHEP",
    volume = "10",
    pages = "206",
    year = "2020"
}

@article{Alvarez:2021zje,
    author = "Alvarez, Ezequiel and Spannowsky, Michael and Szewc, Manuel",
    title = "{Unsupervised Quark/Gluon Jet Tagging With Poissonian Mixture Models}",
    eprint = "2112.11352",
    archivePrefix = "arXiv",
    primaryClass = "hep-ph",
    reportNumber = "IPPP/21/58, ICAS 70/21",
    doi = "10.3389/frai.2022.852970",
    journal = "Front. Artif. Intell.",
    volume = "5",
    pages = "852970",
    year = "2022"
}

@article{LeBlanc:2022bwd,
    author = "LeBlanc, Matt and Nachman, Benjamin and Sauer, Christof",
    title = "{Going off topics to demix quark and gluon jets in {\ensuremath{\alpha}}$_{S}$ extractions}",
    eprint = "2206.10642",
    archivePrefix = "arXiv",
    primaryClass = "hep-ph",
    doi = "10.1007/JHEP02(2023)150",
    journal = "JHEP",
    volume = "02",
    pages = "150",
    year = "2023"
}

@article{Dolan:2023abg,
    author = "Dolan, Matthew J. and Gargalionis, John and Ore, Ayodele",
    title = "{Quark-versus-gluon tagging in CMS Open Data with CWoLa and TopicFlow}",
    eprint = "2312.03434",
    archivePrefix = "arXiv",
    primaryClass = "hep-ph",
    doi = "10.1007/JHEP08(2025)024",
    journal = "JHEP",
    volume = "08",
    pages = "024",
    year = "2025"
}

@ARTICLE{2020SciPy-NMeth,
  author  = {Virtanen, Pauli and Gommers, Ralf and Oliphant, Travis E. and
            Haberland, Matt and Reddy, Tyler and Cournapeau, David and
            Burovski, Evgeni and Peterson, Pearu and Weckesser, Warren and
            Bright, Jonathan and {van der Walt}, St{\'e}fan J. and
            Brett, Matthew and Wilson, Joshua and Millman, K. Jarrod and
            Mayorov, Nikolay and Nelson, Andrew R. J. and Jones, Eric and
            Kern, Robert and Larson, Eric and Carey, C J and
            Polat, {\.I}lhan and Feng, Yu and Moore, Eric W. and
            {VanderPlas}, Jake and Laxalde, Denis and Perktold, Josef and
            Cimrman, Robert and Henriksen, Ian and Quintero, E. A. and
            Harris, Charles R. and Archibald, Anne M. and
            Ribeiro, Ant{\^o}nio H. and Pedregosa, Fabian and
            {van Mulbregt}, Paul and {SciPy 1.0 Contributors}},
  title   = {{{SciPy} 1.0: Fundamental Algorithms for Scientific
            Computing in Python}},
  journal = {Nature Methods},
  year    = {2020},
  volume  = {17},
  pages   = {261--272},
  adsurl  = {https://rdcu.be/b08Wh},
  doi     = {10.1038/s41592-019-0686-2},
}

@article{Jones:1988ay,
    author = "Jones, Lorella M.",
    title = "{Tests for Determining the Parton Ancestor of a Hadron Jet}",
    reportNumber = "UTHEP-189",
    doi = "10.1103/PhysRevD.39.2550",
    journal = "Phys. Rev. D",
    volume = "39",
    pages = "2550",
    year = "1989"
}

@article{Fodor:1989ir,
    author = "Fodor, Z.",
    title = "{How to See the Differences Between Quark and Gluon Jets}",
    reportNumber = "ITP-472-BUDAPEST",
    doi = "10.1103/PhysRevD.41.1726",
    journal = "Phys. Rev. D",
    volume = "41",
    pages = "1726",
    year = "1990"
}

@article{Jones:1990rz,
    author = "Jones, Lorella",
    title = "{TOWARDS A SYSTEMATIC JET CLASSIFICATION}",
    reportNumber = "ILL-TH-90-01",
    doi = "10.1103/PhysRevD.42.811",
    journal = "Phys. Rev. D",
    volume = "42",
    pages = "811--814",
    year = "1990"
}

@article{Lonnblad:1990qp,
    author = "Lonnblad, Leif and Peterson, Carsten and Rognvaldsson, Thorsteinn",
    title = "{Using neural networks to identify jets}",
    reportNumber = "LU-TP-90-8",
    doi = "10.1016/0550-3213(91)90392-B",
    journal = "Nucl. Phys. B",
    volume = "349",
    pages = "675--702",
    year = "1991"
}

@article{Pumplin:1991kc,
    author = "Pumplin, Jon",
    title = "{How to tell quark jets from gluon jets}",
    reportNumber = "MSUTH-91-03",
    doi = "10.1103/PhysRevD.44.2025",
    journal = "Phys. Rev. D",
    volume = "44",
    pages = "2025--2032",
    year = "1991"
}

@article{Gallicchio:2011xq,
    author = "Gallicchio, Jason and Schwartz, Matthew D.",
    title = "{Quark and Gluon Tagging at the LHC}",
    eprint = "1106.3076",
    archivePrefix = "arXiv",
    primaryClass = "hep-ph",
    doi = "10.1103/PhysRevLett.107.172001",
    journal = "Phys. Rev. Lett.",
    volume = "107",
    pages = "172001",
    year = "2011"
}

@article{Gallicchio:2012ez,
    author = "Gallicchio, Jason and Schwartz, Matthew D.",
    title = "{Quark and Gluon Jet Substructure}",
    eprint = "1211.7038",
    archivePrefix = "arXiv",
    primaryClass = "hep-ph",
    doi = "10.1007/JHEP04(2013)090",
    journal = "JHEP",
    volume = "04",
    pages = "090",
    year = "2013"
}

@article{Bhattacherjee:2015psa,
    author = "Bhattacherjee, Biplob and Mukhopadhyay, Satyanarayan and Nojiri, Mihoko M. and Sakaki, Yasuhito and Webber, Bryan R.",
    title = "{Associated jet and subjet rates in light-quark and gluon jet discrimination}",
    eprint = "1501.04794",
    archivePrefix = "arXiv",
    primaryClass = "hep-ph",
    reportNumber = "CAVENDISH-HEP-15-01, KEK-TH-1790, IPMU15-0008",
    doi = "10.1007/JHEP04(2015)131",
    journal = "JHEP",
    volume = "04",
    pages = "131",
    year = "2015"
}

@article{FerreiradeLima:2016gcz,
    author = "Ferreira de Lima, Danilo and Petrov, Petar and Soper, Davison and Spannowsky, Michael",
    title = "{Quark-Gluon tagging with Shower Deconstruction: Unearthing dark matter and Higgs couplings}",
    eprint = "1607.06031",
    archivePrefix = "arXiv",
    primaryClass = "hep-ph",
    reportNumber = "DCPT-16-140, IPPP-16-70",
    doi = "10.1103/PhysRevD.95.034001",
    journal = "Phys. Rev. D",
    volume = "95",
    number = "3",
    pages = "034001",
    year = "2017"
}

@article{Bhattacherjee:2016bpy,
    author = "Bhattacherjee, Biplob and Mukhopadhyay, Satyanarayan and Nojiri, Mihoko M. and Sakaki, Yasuhito and Webber, Bryan R.",
    title = "{Quark-gluon discrimination in the search for gluino pair production at the LHC}",
    eprint = "1609.08781",
    archivePrefix = "arXiv",
    primaryClass = "hep-ph",
    reportNumber = "CAVENDISH-HEP-16-16, KEK-TH-1936, PITT-PACC-1609, IPMU16-0146",
    doi = "10.1007/JHEP01(2017)044",
    journal = "JHEP",
    volume = "01",
    pages = "044",
    year = "2017"
}

@article{Davighi:2017hok,
    author = "Davighi, Joe and Harris, Philip",
    title = "{Fractal based observables to probe jet substructure of quarks and gluons}",
    eprint = "1703.00914",
    archivePrefix = "arXiv",
    primaryClass = "hep-ph",
    reportNumber = "DAMTP-2017-03-06",
    doi = "10.1140/epjc/s10052-018-5819-8",
    journal = "Eur. Phys. J. C",
    volume = "78",
    number = "4",
    pages = "334",
    year = "2018"
}

@article{ATLAS:2023dyu,
    author = "Aad, Georges and others",
    collaboration = "ATLAS",
    title = "{Performance and calibration of quark/gluon-jet taggers using 140 fb$^{−1}$ of pp collisions at TeV with the ATLAS detector*}",
    eprint = "2308.00716",
    archivePrefix = "arXiv",
    primaryClass = "hep-ex",
    reportNumber = "CERN-EP-2023-151",
    doi = "10.1088/1674-1137/acf701",
    journal = "Chin. Phys. C",
    volume = "48",
    number = "2",
    pages = "023001",
    year = "2024"
}

@article{ATLAS:2021suo,
    author = "Aad, Georges and others",
    collaboration = "ATLAS",
    title = "{Search for heavy particles in the $b$-tagged dijet mass distribution with additional $b$-tagged jets in proton-proton collisions at $\sqrt s$= 13{\,}{\,}TeV with the ATLAS experiment}",
    eprint = "2108.09059",
    archivePrefix = "arXiv",
    primaryClass = "hep-ex",
    reportNumber = "CERN-EP-2021-119",
    doi = "10.1103/PhysRevD.105.012001",
    journal = "Phys. Rev. D",
    volume = "105",
    number = "1",
    pages = "012001",
    year = "2022"
}

@article{CMS:2019emo,
    author = "Sirunyan, Albert M and others",
    collaboration = "CMS",
    title = "{Search for low mass vector resonances decaying into quark-antiquark pairs in proton-proton collisions at $\sqrt{s}=$ 13 TeV}",
    eprint = "1909.04114",
    archivePrefix = "arXiv",
    primaryClass = "hep-ex",
    reportNumber = "CMS-EXO-18-012, CERN-EP-2019-176",
    doi = "10.1103/PhysRevD.100.112007",
    journal = "Phys. Rev. D",
    volume = "100",
    number = "11",
    pages = "112007",
    year = "2019"
}

@article{ATLAS:2017zuf,
    author = "Aaboud, Morad and others",
    collaboration = "ATLAS",
    title = "{Search for diboson resonances with boson-tagged jets in $pp$ collisions at $\sqrt{s}=13$ TeV with the ATLAS detector}",
    eprint = "1708.04445",
    archivePrefix = "arXiv",
    primaryClass = "hep-ex",
    reportNumber = "CERN-EP-2017-147",
    doi = "10.1016/j.physletb.2017.12.011",
    journal = "Phys. Lett. B",
    volume = "777",
    pages = "91--113",
    year = "2018"
}

@article{CMS:2017qxu,
    author = "Sirunyan, Albert M and others",
    collaboration = "CMS",
    title = "{Search for supersymmetry in proton-proton collisions at 13 TeV using identified top quarks}",
    eprint = "1710.11188",
    archivePrefix = "arXiv",
    primaryClass = "hep-ex",
    reportNumber = "CMS-SUS-16-050, CERN-EP-2017-269",
    doi = "10.1103/PhysRevD.97.012007",
    journal = "Phys. Rev. D",
    volume = "97",
    number = "1",
    pages = "012007",
    year = "2018"
}

@article{ATLAS:2021kxv,
    author = "Aad, Georges and others",
    collaboration = "ATLAS",
    title = "{Search for new phenomena in events with an energetic jet and missing transverse momentum in $pp$ collisions at $\sqrt {s}$ =13  TeV with the ATLAS detector}",
    eprint = "2102.10874",
    archivePrefix = "arXiv",
    primaryClass = "hep-ex",
    reportNumber = "CERN-EP-2020-238",
    doi = "10.1103/PhysRevD.103.112006",
    journal = "Phys. Rev. D",
    volume = "103",
    number = "11",
    pages = "112006",
    year = "2021"
}

@article{CMS:2017asf,
    author = "Sirunyan, A. M. and others",
    collaboration = "CMS",
    title = "{Search for vectorlike light-flavor quark partners in proton-proton collisions at $\sqrt s$ =8  TeV}",
    eprint = "1708.02510",
    archivePrefix = "arXiv",
    primaryClass = "hep-ex",
    reportNumber = "CMS-B2G-12-016, CERN-EP-2017-145",
    doi = "10.1103/PhysRevD.97.072008",
    journal = "Phys. Rev. D",
    volume = "97",
    pages = "072008",
    year = "2018"
}

@article{CMS:2017mbm,
    author = "Sirunyan, Albert M and others",
    collaboration = "CMS",
    title = "{Search for direct production of supersymmetric partners of the top quark in the all-jets final state in proton-proton collisions at $ \sqrt{s}=13 $ TeV}",
    eprint = "1707.03316",
    archivePrefix = "arXiv",
    primaryClass = "hep-ex",
    reportNumber = "CMS-SUS-16-049, CERN-EP-2017-129",
    doi = "10.1007/JHEP10(2017)005",
    journal = "JHEP",
    volume = "10",
    pages = "005",
    year = "2017"
}

@article{Adams:2015hiv,
    author = "Adams, D. and others",
    title = "{Towards an Understanding of the Correlations in Jet Substructure}",
    eprint = "1504.00679",
    archivePrefix = "arXiv",
    primaryClass = "hep-ph",
    reportNumber = "FERMILAB-PUB-15-670-CMS, SLAC-PUB-16703",
    doi = "10.1140/epjc/s10052-015-3587-2",
    journal = "Eur. Phys. J. C",
    volume = "75",
    number = "9",
    pages = "409",
    year = "2015"
}

@article{Altheimer:2013yza,
    author = "Altheimer, A. and others",
    title = "{Boosted Objects and Jet Substructure at the LHC. Report of BOOST2012, held at IFIC Valencia, 23rd-27th of July 2012}",
    eprint = "1311.2708",
    archivePrefix = "arXiv",
    primaryClass = "hep-ex",
    reportNumber = "FERMILAB-PUB-13-669",
    doi = "10.1140/epjc/s10052-014-2792-8",
    journal = "Eur. Phys. J. C",
    volume = "74",
    number = "3",
    pages = "2792",
    year = "2014"
}

@article{Altheimer:2012mn,
    author = "Altheimer, A. and others",
    title = "{Jet Substructure at the Tevatron and LHC: New Results, New Tools, New Benchmarks}",
    eprint = "1201.0008",
    archivePrefix = "arXiv",
    primaryClass = "hep-ph",
    reportNumber = "SLAC-R-990, FERMILAB-PUB-12-897-T",
    doi = "10.1088/0954-3899/39/6/063001",
    journal = "J. Phys. G",
    volume = "39",
    pages = "063001",
    year = "2012"
}

@article{Abdesselam:2010pt,
    author = "Abdesselam, A. and others",
    title = "{Boosted Objects: A Probe of Beyond the Standard Model Physics}",
    eprint = "1012.5412",
    archivePrefix = "arXiv",
    primaryClass = "hep-ph",
    reportNumber = "SLAC-PUB-15081, FERMILAB-PUB-10-617-CMS",
    doi = "10.1140/epjc/s10052-011-1661-y",
    journal = "Eur. Phys. J. C",
    volume = "71",
    pages = "1661",
    year = "2011"
}

@article{Salam:2010nqg,
    author = "Salam, Gavin P.",
    title = "{Towards Jetography}",
    eprint = "0906.1833",
    archivePrefix = "arXiv",
    primaryClass = "hep-ph",
    doi = "10.1140/epjc/s10052-010-1314-6",
    journal = "Eur. Phys. J. C",
    volume = "67",
    pages = "637--686",
    year = "2010"
}

@inproceedings{Shelton:2013an,
    author = "Shelton, Jessie",
    title = "{Jet Substructure}",
    booktitle = "{Theoretical Advanced Study Institute in Elementary Particle Physics}: {Searching for New Physics at Small and Large Scales}",
    eprint = "1302.0260",
    archivePrefix = "arXiv",
    primaryClass = "hep-ph",
    doi = "10.1142/9789814525220_0007",
    pages = "303--340",
    year = "2013"
}

@article{Plehn:2011tg,
    author = "Plehn, Tilman and Spannowsky, Michael",
    title = "{Top Tagging}",
    eprint = "1112.4441",
    archivePrefix = "arXiv",
    primaryClass = "hep-ph",
    doi = "10.1088/0954-3899/39/8/083001",
    journal = "J. Phys. G",
    volume = "39",
    pages = "083001",
    year = "2012"
}

@article{Cacciari:2015jwa,
    author = "Cacciari, Matteo",
    title = "{Phenomenological and theoretical developments in jet physics at the LHC}",
    eprint = "1509.02272",
    archivePrefix = "arXiv",
    primaryClass = "hep-ph",
    doi = "10.1142/S0217751X1546001X",
    journal = "Int. J. Mod. Phys. A",
    volume = "30",
    number = "31",
    pages = "1546001",
    year = "2015"
}

@book{Marzani:2019hun,
    author = "Marzani, Simone and Soyez, Gregory and Spannowsky, Michael",
    title = "{Looking inside jets: an introduction to jet substructure and boosted-object phenomenology}",
    eprint = "1901.10342",
    archivePrefix = "arXiv",
    primaryClass = "hep-ph",
    doi = "10.1007/978-3-030-15709-8",
    publisher = "Springer",
    volume = "958",
    year = "2019"
}

@article{Larkoski:2017jix,
    author = "Larkoski, Andrew J. and Moult, Ian and Nachman, Benjamin",
    title = "{Jet Substructure at the Large Hadron Collider: A Review of Recent Advances in Theory and Machine Learning}",
    eprint = "1709.04464",
    archivePrefix = "arXiv",
    primaryClass = "hep-ph",
    doi = "10.1016/j.physrep.2019.11.001",
    journal = "Phys. Rept.",
    volume = "841",
    pages = "1--63",
    year = "2020"
}

@article{Larkoski:2024uoc,
    author = "Larkoski, Andrew J.",
    title = "{QCD masterclass lectures on jet physics and machine learning}",
    eprint = "2407.04897",
    archivePrefix = "arXiv",
    primaryClass = "hep-ph",
    doi = "10.1140/epjc/s10052-024-13341-0",
    journal = "Eur. Phys. J. C",
    volume = "84",
    number = "10",
    pages = "1117",
    year = "2024"
}

@article{Kogler:2018hem,
    author = "Kogler, Roman and others",
    title = "{Jet Substructure at the Large Hadron Collider: Experimental Review}",
    eprint = "1803.06991",
    archivePrefix = "arXiv",
    primaryClass = "hep-ex",
    reportNumber = "FERMILAB-PUB-18-123-PPD",
    doi = "10.1103/RevModPhys.91.045003",
    journal = "Rev. Mod. Phys.",
    volume = "91",
    number = "4",
    pages = "045003",
    year = "2019"
}

@article{Guest:2018yhq,
    author = "Guest, Dan and Cranmer, Kyle and Whiteson, Daniel",
    title = "{Deep Learning and its Application to LHC Physics}",
    eprint = "1806.11484",
    archivePrefix = "arXiv",
    primaryClass = "hep-ex",
    doi = "10.1146/annurev-nucl-101917-021019",
    journal = "Ann. Rev. Nucl. Part. Sci.",
    volume = "68",
    pages = "161--181",
    year = "2018"
}

@article{Albertsson:2018maf,
    author = "Albertsson, Kim and others",
    title = "{Machine Learning in High Energy Physics Community White Paper}",
    eprint = "1807.02876",
    archivePrefix = "arXiv",
    primaryClass = "physics.comp-ph",
    reportNumber = "FERMILAB-PUB-18-318-CD-DI-PPD",
    doi = "10.1088/1742-6596/1085/2/022008",
    journal = "J. Phys. Conf. Ser.",
    volume = "1085",
    number = "2",
    pages = "022008",
    year = "2018"
}

@article{Radovic:2018dip,
    author = "Radovic, Alexander and Williams, Mike and Rousseau, David and Kagan, Michael and Bonacorsi, Daniele and Himmel, Alexander and Aurisano, Adam and Terao, Kazuhiro and Wongjirad, Taritree",
    title = "{Machine learning at the energy and intensity frontiers of particle physics}",
    reportNumber = "FERMILAB-PUB-18-436-ND",
    doi = "10.1038/s41586-018-0361-2",
    journal = "Nature",
    volume = "560",
    number = "7716",
    pages = "41--48",
    year = "2018"
}

@article{Carleo:2019ptp,
    author = "Carleo, Giuseppe and Cirac, Ignacio and Cranmer, Kyle and Daudet, Laurent and Schuld, Maria and Tishby, Naftali and Vogt-Maranto, Leslie and Zdeborov{\'a}, Lenka",
    title = "{Machine learning and the physical sciences}",
    eprint = "1903.10563",
    archivePrefix = "arXiv",
    primaryClass = "physics.comp-ph",
    doi = "10.1103/RevModPhys.91.045002",
    journal = "Rev. Mod. Phys.",
    volume = "91",
    number = "4",
    pages = "045002",
    year = "2019"
}

@article{Bourilkov:2019yoi,
    author = "Bourilkov, Dimitri",
    title = "{Machine and Deep Learning Applications in Particle Physics}",
    eprint = "1912.08245",
    archivePrefix = "arXiv",
    primaryClass = "physics.data-an",
    doi = "10.1142/S0217751X19300199",
    journal = "Int. J. Mod. Phys. A",
    volume = "34",
    number = "35",
    pages = "1930019",
    year = "2020"
}

@article{Schwartz:2021ftp,
    author = "Schwartz, Matthew D.",
    title = "{Modern Machine Learning and Particle Physics}",
    eprint = "2103.12226",
    archivePrefix = "arXiv",
    primaryClass = "hep-ph",
    doi = "10.1162/99608f92.beeb1183",
    month = "3",
    year = "2021"
}

@article{Karagiorgi:2022qnh,
    author = "Karagiorgi, Georgia and Kasieczka, Gregor and Kravitz, Scott and Nachman, Benjamin and Shih, David",
    title = "{Machine learning in the search for new fundamental physics}",
    doi = "10.1038/s42254-022-00455-1",
    journal = "Nature Rev. Phys.",
    volume = "4",
    number = "6",
    pages = "399--412",
    year = "2022"
}

@article{Boehnlein:2021eym,
    author = "Boehnlein, Amber and others",
    title = "{Colloquium: Machine learning in nuclear physics}",
    eprint = "2112.02309",
    archivePrefix = "arXiv",
    primaryClass = "nucl-th",
    doi = "10.1103/RevModPhys.94.031003",
    journal = "Rev. Mod. Phys.",
    volume = "94",
    number = "3",
    pages = "031003",
    year = "2022"
}

@article{Plehn:2022ftl,
    author = "Plehn, Tilman and Butter, Anja and Dillon, Barry and Heimel, Theo and Krause, Claudius and Winterhalder, Ramon",
    title = "{Modern Machine Learning for LHC Physicists}",
    eprint = "2211.01421",
    archivePrefix = "arXiv",
    primaryClass = "hep-ph",
    month = "11",
    year = "2022"
}

@article{Bonilla:2022wzp,
    author = "Bonilla, Johan and others",
    title = "{Jets and Jet Substructure at Future Colliders}",
    eprint = "2203.07462",
    archivePrefix = "arXiv",
    primaryClass = "hep-ph",
    reportNumber = "FERMILAB-PUB-22-186-SCD-T",
    doi = "10.3389/fphy.2022.897719",
    journal = "Front. in Phys.",
    volume = "10",
    pages = "897719",
    year = "2022"
}

@article{DeZoort:2023vrm,
    author = "DeZoort, Gage and Battaglia, Peter W. and Biscarat, Catherine and Vlimant, Jean-Roch",
    title = "{Graph neural networks at the Large Hadron Collider}",
    doi = "10.1038/s42254-023-00569-0",
    journal = "Nature Rev. Phys.",
    volume = "5",
    number = "5",
    pages = "281--303",
    year = "2023"
}

@article{Zhou:2023pti,
    author = "Zhou, Kai and Wang, Lingxiao and Pang, Long-Gang and Shi, Shuzhe",
    title = "{Exploring QCD matter in extreme conditions with Machine Learning}",
    eprint = "2303.15136",
    archivePrefix = "arXiv",
    primaryClass = "hep-ph",
    doi = "10.1016/j.ppnp.2023.104084",
    journal = "Prog. Part. Nucl. Phys.",
    volume = "135",
    pages = "104084",
    year = "2024"
}

@article{Belis:2023mqs,
    author = "Belis, Vasilis and Odagiu, Patrick and Aarrestad, Thea Klaeboe",
    title = "{Machine learning for anomaly detection in particle physics}",
    eprint = "2312.14190",
    archivePrefix = "arXiv",
    primaryClass = "physics.data-an",
    doi = "10.1016/j.revip.2024.100091",
    journal = "Rev. Phys.",
    volume = "12",
    pages = "100091",
    year = "2024"
}

@article{Mondal:2024nsa,
    author = "Mondal, Spandan and Mastrolorenzo, Luca",
    title = "{Machine learning in high energy physics: a review of heavy-flavor jet tagging at the LHC}",
    eprint = "2404.01071",
    archivePrefix = "arXiv",
    primaryClass = "hep-ex",
    doi = "10.1140/epjs/s11734-024-01234-y",
    journal = "Eur. Phys. J. ST",
    volume = "233",
    number = "15-16",
    pages = "2657--2686",
    year = "2024"
}

@article{Feickert:2021ajf,
    author = "Feickert, Matthew and Nachman, Benjamin",
    title = "{A Living Review of Machine Learning for Particle Physics}",
    eprint = "2102.02770",
    archivePrefix = "arXiv",
    primaryClass = "hep-ph",
    month = "2",
    year = "2021"
}

@article{Behring:2025ilo,
    author = "Behring, Arnd and others",
    title = "{Flavoured jet algorithms: a comparative study}",
    eprint = "2506.13449",
    archivePrefix = "arXiv",
    primaryClass = "hep-ph",
    reportNumber = "CERN-TH-2025-113, IFJPAN-IV-2025-13, MCNET-25-14, MPP-2025-118, PUBDB-2025-01862",
    doi = "10.1007/JHEP09(2025)149",
    journal = "JHEP",
    volume = "09",
    pages = "149",
    year = "2025"
}

@article{Thaler:2011gf,
    author = "Thaler, Jesse and Van Tilburg, Ken",
    title = "{Maximizing Boosted Top Identification by Minimizing N-subjettiness}",
    eprint = "1108.2701",
    archivePrefix = "arXiv",
    primaryClass = "hep-ph",
    reportNumber = "MIT-CTP-4287",
    doi = "10.1007/JHEP02(2012)093",
    journal = "JHEP",
    volume = "02",
    pages = "093",
    year = "2012"
}

@inproceedings{10.1145/275487.275505,
author = {Papadimitriou, Christos H. and Tamaki, Hisao and Raghavan, Prabhakar and Vempala, Santosh},
title = {Latent semantic indexing: a probabilistic analysis},
year = {1998},
isbn = {0897919963},
publisher = {Association for Computing Machinery},
address = {New York, NY, USA},
url = {https://doi.org/10.1145/275487.275505},
doi = {10.1145/275487.275505},
booktitle = {Proceedings of the Seventeenth ACM SIGACT-SIGMOD-SIGART Symposium on Principles of Database Systems},
pages = {159–168},
numpages = {10},
location = {Seattle, Washington, USA},
series = {PODS '98}
}

@article{Acosta:2025kgk,
    author = "Acosta, P. and others",
    title = "{SubMIT: A Physics Analysis Facility at MIT}",
    eprint = "2506.01958",
    archivePrefix = "arXiv",
    primaryClass = "cs.DC",
    reportNumber = "MIT-CTP/5848",
    doi = "10.1007/s41781-025-00156-1",
    journal = "EPJ RI",
    volume = "10",
    number = "1",
    pages = "2",
    year = "2026"
}

@inproceedings{Li:2021LabelNoise,
    author    = {Li, Xuefeng and Liu, Tongliang and Han, Bo and Niu, Gang and Sugiyama, Masashi},
    title     = {Provably End-to-End Label-Noise Learning without Anchor Points},
    booktitle = {Proceedings of the 38th International Conference on Machine Learning},
    series    = {Proceedings of Machine Learning Research},
    volume    = {139},
    pages     = {6403--6413},
    year      = {2021},
    publisher = {PMLR},
    editor    = {Meila, Marina and Zhang, Tong},
    url       = {https://proceedings.mlr.press/v139/li21l.html}
}

@misc{SimplexDemixingCode,
    author  = {de la Fuente, Gregorio and Thaler, Jesse},
    title   = {{Simplex demixing code}},
    howpublished = {\url{https://github.com/gregofuente/simplex_demixing}},
    url     = {https://github.com/gregofuente/simplex_demixing},
    note    = {ML code for simplex demixing},
}

@dataset{SimplexDemixingDatasets,
  author       = {de la Fuente Simarro, Gregorio and
                  Thaler, Jesse},
  title        = {Pythia8 Jet Mixtures for Simplex Demixing},
  month        = jul,
  year         = 2026,
  publisher    = {Zenodo},
  doi          = {10.5281/zenodo.21542473},
  url          = {https://doi.org/10.5281/zenodo.21542473},
}

@article{Larkoski:2026zvn,
    author = "Larkoski, Andrew J.",
    title = "{More QCD Masterclass Lectures on Jet Physics and Machine Learning}",
    eprint = "2609.09785",
    archivePrefix = "arXiv",
    primaryClass = "hep-ph",
    month = "9",
    year = "2026"
}

\end{document}